\RequirePackage{afterpackage}
\AfterPackage{amsthm}{
    \RequirePackage{hyperref}
    \RequirePackage[capitalise]{cleveref}

    \usepackage{proof-at-the-end}
    \pgfkeys{/prAtEnd/default/.style={restate, proof at the end, stared, no link to proof}}
    \pgfkeys{/prAtEnd/allend/.style={all end, no link to theorem}}
    
}

\documentclass{acmart}

\usepackage{amsmath}
\usepackage{mathpartir}
\usepackage{mathtools}
\usepackage{xcolor}
\usepackage{graphicx}
\usepackage{calc}
\usepackage{xspace}

\usepackage{amsthm}
\usepackage{thmtools}
\usepackage{stmaryrd}
\usepackage{cases}

\usepackage[capitalise]{cleveref}

\usepackage{xfakebold}
\usepackage{bm}

\AtEndPreamble{%
   \theoremstyle{acmdefinition}
   \newtheorem{remark}[theorem]{Remark}}

\newcommand{\changed}[1]{{\color{red} #1}}
\newcommand{\revOne}[1]{{\color{blue} #1}}
\newcommand{\revTwo}[1]{{\color{orange} #1}}

\newcommand{\DMon}{\mathsf{DMon}}
\newcommand{\CMon}{\mathsf{CMon}}
\newcommand{\LMon}{\mathsf{LMon}}
\newcommand{\Act}{\mathsf{Act}}
\newcommand{\Com}{\mathsf{Com}}
\newcommand{\Cvar}{\mathsf{Con}}
\newlength{\NOTskip} 
\def\NOT#1{\settowidth{\NOTskip}{\ensuremath{#1}}%
            \hspace{0.5\NOTskip}\mathclap{/}\hspace{-0.5\NOTskip}#1}
            
\newcommand{\rec}[1]{\mathsf{rec}~#1}

\newcommand{\xxoverset}[3]{%
\resizebox{#1+\widthof{\scriptsize #2}}{\height}{$#3$}}
\newcommand{\extoverset}[3][0pt]{%
  \overset{\textup{#2}}{\xxoverset{#1}{#2}{#3}}}

\newcommand{\muHML}{$\mu$HML\xspace}
\newcommand{\HypermuHML}{Hyper-\ensuremath{\mathsf{rec}}HML\xspace}

\newcommand{\hyperLTLtwo}{\ensuremath{\textsc{Hyper}^{2}\textsc{LTL}}\xspace}
\newcommand{\hyperLTLtwofp}{\ensuremath{\textsc{Hyper}^{2}\textsc{LTL}_\text{fp}}\xspace}

\newcommand{\ttt}{\mathsf{tt}}
\newcommand{\ff}{\mathsf{ff}}

\newcommand{\maxx}[1]{\textrm{max}~#1}
\newcommand{\minn}[1]{\textrm{min}~#1}

\newcommand{\Trc}{\mathsf{Trc}}
\newcommand{\Htrc}{\mathsf{HTrc}_I}

\newcommand{\sem}[1]{\left\llbracket#1\right\rrbracket}

\newcommand{\fx}{\mathtt{fix}}
\newcommand{\cl}{\mathtt{cls}}

\newcommand{\R}{\mathcal{R}}

\newcommand{\boldarrow}{\setBold[0.6] \rightarrow\unsetBold}
\newcommand{\boldarrowt}{\setBold[0.6]\Rightarrow_{\!\unsetBold\mathnormal{A}}\unsetBold}


\newcommand{\MdLocName}{{}{\textsc{d}}\mathsf{m}}
\newcommand{\MdName}{{}{\textsc{d}}\mathsf{M}}
\newcommand{\synone}[1]{\MdLocName^\ell_\sigma(#1)}
\newcommand{\Mone}[1]{\Monev{\sigma}{#1}}
\newcommand{\Monev}[2]{\MdName_{#1}(#2)}
\newcommand{\synvone}[3]{\MdLocName^{#1}_{#2}(#3)}

\newcommand{\McName}{{}{\textsc{c}}\mathsf{m}}
\renewcommand{\Mc}[1]{\Mcv{#1}{\sigma}}
\newcommand{\Mcv}[2]{\McName_{#2}(#1)}

\newcommand{\MdecVarName}{\mathcal{M}}
\newcommand{\MdecVar}[2]{\MdecVarName_{#1} (#2 )}
\newcommand{\MdecVarDefSigma}[1]{\MdecVar{\sigma}{#1}}
\newcommand{\MdecVarDefault}{\MdecVarDefSigma{\varphi}}

\newcommand{\PHypermuHML}{PHyper-$\mathsf{rec}$HML\xspace}
\newcommand{\FPHypermuHML}{PHyper-maxHML\xspace}
\newcommand{\WTFPHypermuHML}{Hyper-maxHML\xspace}

\makeatletter
\newcommand\bigDiamond{\mathop{\mathpalette\bigDi@mond\relax}}
\newcommand\bigDi@mond[2]{%
  \vcenter{\hbox{\m@th
    \scalebox{\ifx#1\displaystyle 2\else1.2\fi}{$#1\Diamond$}%
  }}%
}
\newcommand{\op}{\diamond}
\newcommand{\bigop}{\bigDiamond}
\newcommand{\xrightsquigarrow}[1]{\extoverset{$#1$}\rightsquigarrow}

\newcommand{\yes}{\mathsf{yes}}
\newcommand{\no}{\mathsf{no}}
\newcommand{\vend}{\mathsf{end}}

\newcommand{\quantfree}{\mathsf{Qf}}
\newcommand{\FVloc}{\mathsf{fv}_{\mathsf{l}}}
\newcommand{\FVrec}{\mathsf{fv}_{\mathsf{r}}}

\newcommand{\range}{\mathsf{rng}}
\newcommand{\dom}{\mathsf{dom}}

\makeatletter
\def\bRightarrowfill@{\arrowfill@{}{}{\Rightarrow}}
\def\xyarrow#1#2#3#4#5{\ext@arrow 0959{#1}{#3}{#2}{\!\!}^{#4}_{#5}}
\def\xybRightarrow{\xyarrow\bRightarrowfill@}
\makeatother

\newcommand*{\exqed}{\hfill$\blacktriangleleft$}

\usepackage{soul}

\newif\ifextendedversion 
\extendedversiontrue

\newif\ifarxiversion 
\arxiversiontrue 

\newif\ifnotarxiversion
\ifarxiversion
\notarxiversionfalse
\else
\notarxiversiontrue
\fi

\ifarxiversion
\renewcommand{\changed}[1]{#1}
\renewcommand{\revOne}[1]{#1}
\renewcommand{\revTwo}[1]{#1}
\fi

\title{Centralized vs Decentralized Monitors for Hyperproperties}

\author{Luca Aceto}
\orcid{https://orcid.org/0000-0001-8554-6907}
\affiliation{\institution{Dept. of Computer Science, Reykjavik University} \country{ Iceland}}
\affiliation{\institution{Gran Sasso Science Institute}\city{ L'Aquila}\country{ Italy}}
\email{luca@ru.is}

\author{Antonis Achilleos}
\orcid{https://orcid.org/0000-0002-1314-333X}
\affiliation{\institution{Dept. of Computer Science, Reykjavik University} \country{Iceland}}
\email{antonios@ru.is}

\author{Elli Anastasiadi}
\orcid{https://orcid.org/0000-0001-7526-9256}
\affiliation{\institution{Aalborg University} \city{Aalborg} \country{Denmark}}
\email{elli.anastasiadi@it.uu.se}

\author{Adrian Francalanza}
\orcid{https://orcid.org/0000-0003-3829-7391}
\affiliation{ \institution{University of Malta} \country{Malta}}
\email{adrian.francalanza@um.edu.mt}

\author{Daniele Gorla}
\orcid{https://orcid.org/0000-0001-8859-9844}
\affiliation{\institution{Dept. of Computer Science, ``Sapienza'' University of Rome} \country{Italy}}
\email{gorla@di.uniroma1.it}

\author{Jana Wagemaker}
\orcid{https://orcid.org/0000-0002-8616-3905}
\affiliation{\institution{iCIS, Radboud University} \city{Nijmegen} \country{Netherlands}}
\email{jana.wagemaker@ru.nl}

    \ccsdesc[500]{Theory of computation~Operational semantics}
    \ccsdesc[500]{Theory of computation~Modal and temporal logics}
    \ccsdesc[500]{Theory of computation~Logic and verification}
    
\keywords{Runtime Verification, hyperlogics, decentralization}

\begin{document}

\sloppy

\newcommand{\pandor}{\odot}
\newcommand{\por}{\oplus}
\newcommand{\pand}{\otimes}
\newcommand{\Loc}{{\mathcal L}}
\renewcommand{\Htrc}{\mathsf{HTrc}}
\newcommand{\valto}{\Rrightarrow}

\marginparsep 5pt
\marginparwidth 40pt
 \newcommand{\olds}[1]{\oldstylenums{#1}}
 \newcommand{\oldsb}[1]{{\bfseries\olds{#1}}}
 \newcommand{\mnote}[1]{\stepcounter{ncomm}%
 \vbox to0pt{\vss\llap{\tiny\oldsb{\arabic{ncomm}}}\vskip6pt}%
 \marginpar{\tiny\bf\raggedright%
 {\oldsb{\arabic{ncomm}}}.\hskip0.5em#1}}
 \newcounter{ncomm}%

\newcommand{\dan}[1]{\marginpar{ \textbf{Da:} {\footnotesize\em #1}}}
\newcommand{\jana}[1]{\marginpar{ \textbf{Ja:} {\footnotesize\em #1}}}
\newcommand{\antonis}[1]{\marginpar{ \textbf{An:} {\footnotesize\em #1}}}
\newcommand{\elli}[1]{\marginpar{ \textbf{El:} {\footnotesize\em #1}}}
\newcommand{\af}[1]{\marginpar{\textbf{Ad:} {\footnotesize\em #1}}}

\newcommand{\hd}{{\mathit hd}}
\newcommand{\tl}{{\mathit tl}}

\begin{abstract}
This paper focuses on the runtime verification of hyperproperties expressed in \HypermuHML, an expressive yet simple logic for describing properties of sets of traces. 
To this end, we consider a simple language of monitors that observe sets of system executions and report verdicts w.r.t. a given \HypermuHML formula. 
%
We first employ a unique omniscient monitor that centrally observes all system traces.
Since centralised monitors are not ideal for distributed settings, we also provide a language for decentralized monitors, where each trace has a dedicated monitor;
these monitors yield a unique verdict by communicating their observations to one another. 
For both the centralized and the decentralized settings, we provide a synthesis procedure that, given a formula, yields a monitor that is correct (i.e., sound and violation complete). 
A key step in proving the correctness of the synthesis for decentralized monitors is a result showing that, for each formula, the synthesized centralized monitor and its corresponding decentralized one are weakly bisimilar for a suitable notion of weak bisimulation.


\end{abstract}

\maketitle

\renewcommand\shortauthors{Aceto, Achilleos, Anastasiadi, Francalanza, Gorla, Wagemaker}

\begin{acks}
    This work has been supported by the project `Mode(l)s of Verification and Monitorability' (MoVeMent) (grant No \grantnum{irf}{217987}) of the \grantsponsor{irf}{Icelandic Research Fund}{}. Elli Anastasiadi's research has been supported by grant \grantnum{src}{VR 2020-04430}
    of the \grantsponsor{src}{Swedish Research Council}{}.
\end{acks}

\section{Introduction}
Runtime verification (RV)~\cite{Bartocci2018} is a verification technique that observes system executions to determine whether some given specification is satisfied or violated.
%
This runtime analysis is usually conducted by a computational entity called a \textit{monitor}~\cite{Francalanza21}.
RV is a lightweight verification technique 
that is carried out as the system under observation executes,
%
thereby avoiding scalability issues caused by the state-explosion problem, as is the case for model checking.
Recently, RV has been extended to parallel set-ups \cite{Bocchi_Tuosto_DistInteraction,Cassar2017ReliabilityAF,choreogr_monitors_as_mem}, and a large body of work in that setting aims to verify \textit{hyperproperties} at runtime \cite{circuit_monitors,monitoring_hyperLTL,complex_monitor_hyper,hyperproperties,monitor_hyper}.

Hyperproperties~\cite{hyperproperties} are sets of \emph{hypertraces}, \emph{i.e.}~sets of traces that may be seen as describing different system executions or the contributions of different sequential processes to a system execution.
As argued in~\cite{BonakdarpourSS18}, many properties of concurrent and distributed systems can be viewed as hyperproperties.
%
When verifying hyperproperties at runtime, several traces (i.e. several execution sequences) can be observed instead of just one, possibly at the same time. 
Several extensions of temporal logics, such as HyperLTL, HyperCTL$^*$~\cite{hyperLTL}, Hyper$^2$LTL \cite{BeutnerFFM23}, have been defined to express hyperproperties. 
%
Extensions of standard logics to hyper properties also include variations of the $\mu$-calculus, such as \cite{circuit_monitors}, setting the basis for the logic used in this paper, and \cite{Gustfeld21}, which studies an asynchronous semantics.
%

Since they were proposed by Clarkson and Schneider in~\cite{hyperproperties}, hyperproperties have become a fundamental, trace-based formalism for expressing  security and privacy properties, 
verified using static and dynamic techniques~\cite{AgrawalB16,BeutnerFFM23,BeutnerFFM24,monitoring_hyperLTL,BonakdarpourSS18,brett17,ChalupaH23,monitor_hyper} 
implemented in a variety of tools~\cite{BeutnerFFM24,FinkbeinerHST18,BeutnerF22}. 
There is a large body of work, such as~\cite{AgrawalB16,brett17,10.1007/978-3-030-17465-1_7}, detailing several algorithms for monitoring (fragments of) hyperlogics under  different assumptions and providing several correctness guarantees.
%
However, these proposals either construct a centralized monitoring algorithm that has access to all traces in the observed hypertrace, or verify single trace properties, over a distributed set-up\footnote{See e.g.~\cite{
BonakdarpourFRR22, BonakdarpourFRT16, FraigniaudRT20,FraGP13} for distributed monitoring algorithms for classic trace-based logics.}. 
Having an omniscient monitor simplifies the runtime analysis since the monitoring algorithm can compare all traces as needed by simply accessing different parts of its local memory. 
But this power comes with drawbacks. 
For starters, centralized monitors are unrealistic for distributed systems, where trace analysis is typically localised to network nodes so as to minimize communication across locations.
Moreover, centralized monitors
create single points of failure during verification~\cite{AcetoAFI24}. 
Furthermore,  it can be problematic to store all the traces locally, especially in light of the wide availability of multi-core systems.
The goal of the decentralized monitor synthesis from logical specifications presented in this paper is to permit distributed monitor choreographies with \emph{local} trace views whose components communicate in order to verify \emph{global} properties (such as hyperproperties). 
Decentralized monitors have been shown to avoid high contentions  leading to vastly improved scalability~\cite{AcetoAFI24}. 
They also offer better privacy guarantees whenever they are stationed locally at the nodes where the respective traces are generated~\cite{FraGP13,JiaGP16}. 
To the best of our knowledge, 
%
such a message-passing monitoring set-up has never been studied for the purpose of verifying hyperproperties so far.

In this paper, we study procedures for the 
\emph{automated synthesis of
centralized and decentralized monitors} from hyperproperties described
in the logic \HypermuHML~\cite{circuit_monitors}. 
%
This logic extends the linear-time~\cite{Vardi88} $\mu$-calculus~\cite{Kozen-paper}  (also known as Hennessy-Milner logic with recursion~\cite{Larsen-paper}) 
with constructs to describe properties
of hypertraces inspired by the work on HyperLTL (namely variables
ranging over traces, modal operators parametrized by trace variables, matching/mismatching between trace variables,
and existential and universal quantification over them).
%
\HypermuHML
can 
describe hyperproperties not expressible in HyperLTL or HyperCTL*, 
%
such as properties that speak about consensus (see Example \ref{ex:One}) and periodicity (see Example \ref{ex:Two}).
Furthermore, \HypermuHML\ supports a general, syntax-driven monitor synthesis that can handle both the aforementioned hyperproperties, at least in the centralized case (see also the discussion in Section \ref{sec:related_work}).

In both the centralized and decentralized set-ups, we work in the parallel model~\cite{monitor_hyper}, where a fixed number
of system executions is processed in parallel by monitors in an online
fashion. We specify monitors using a process-algebraic formalism that
builds on the one presented in \cite{AcetoPOPL19,HMLrec_mon} to define a class of monitors
called regular. Such monitors are easy to describe, resemble
(alternating) automata, and have sufficient expressive power to
provide standard monitoring guarantees. 
Moreover, their algebraic
structure supports the compositional definition of their operational
semantics and monitor synthesis procedures from formulas,
building on
previous work relating algebraic process calculi with RV~\cite{BocchiCDHY17,InoueY17,Fra17,AcetoAFIL19,Francalanza21,LanotteMM21,AcetoCFI23,LanotteMM23}.

In the centralized case, for each formula in the fragment of
\HypermuHML limited to greatest-fixed-point operators, our synthesis procedure yields a monolithic monitor that has access to all
the traces in an observed hypertrace. 
However, in order to synthesize decentralized monitors for a sufficiently
expressive fragment of the logic, it is necessary to extend the
monitor capabilities with communication, as
shown already in \cite{circuit_monitors}. 
For instance, to monitor for
the property
 ``If there is a trace where event $a$ occurs, then there
exists another trace where event $b$ does not occur thereafter'',
monitors observing different traces need to
communicate to record that event $a$ occurred in some
trace at some point and that there is some trace where $b$ does not
occur from that point onwards. Allowing monitors to send and receive messages
significantly complicates their operational semantics (see Section \ref{sec:decMon}), the monitor
synthesis procedure (see Section \ref{sec:decSyn}), and all consequent proofs. 
The operational semantics for communicating monitors is one of the main contributions of the paper since its design is crucial to obtain the correctness guarantees provided by the synthesis procedure for decentralized monitors.
In
particular, the semantics of decentralized monitors and their
synthesis from formulas have to be designed carefully to ensure that
monitors are reactive (they are always ready to process any
system event) and input-enabled (they can always receive any
input from other monitors in their environment), properties that are desirable in any decentralized RV set-up.

We show that both \emph{the centralized and the decentralized monitor
synthesis procedures are correct}.
More precisely, the monitors
synthesized from formulas are \emph{sound} and
\emph{violation-complete}, meaning that (1) if the monitor synthesized from a formula $\varphi$ reports a positive (resp., negative) verdict when observing a hypertrace $T$, then $T$ does (resp., does not) satisfy $\varphi$, and (2) if $T$ does not satisfy $\varphi$, then its associated monitor will report a negative verdict when observing $T$ (see Theorems \ref{thm:soundCentr} and \ref{thm:complCentr}, and Corollaries \ref{cor:soundDec} and \ref{cor:compDec}). The proof of correctness in the decentralized case is considerably more technical than the corresponding proof in the centralized setting, due to the intricate communication semantics. To address the resulting technical challenges, we develop a proof strategy where we prove the correctness of the decentralized monitor
synthesis procedure using the centralized one as a yardstick.

This methodology is one of the key contributions we offer in this study. More
precisely, in Section \ref{sec:princ} \emph{we identify six properties of a decentralized monitor
synthesis that make it `principled'} (see Definition \ref{def:principled-synthesis}) and we show that,
when a decentralized monitor synthesis is principled, the centralized
and decentralized monitors synthesized from a formula are related by a
suitable notion of weak bisimulation (Theorem \ref{thm:principled-bisimulation}). Apart from
supporting the definition of decentralized monitor synthesis
procedures, this result allows us to reduce the correctness of our
decentralized monitor synthesis to that of the centralized one, which can in turn drive the definition of further synthesis procedures in future work. 
We also conjecture that our methodology provides a path to proving similar results for other models of communicating monitors independent of the monitoring strategy. 
In summary, our contributions are the following:
\begin{itemize}
    \item a framework for monitoring hyperproperties by a central monitor that has access to all locations (\Cref{sec.cenMon})
    and 
    a decentalized monitoring set-up for hyperproperties, with monitors that communicate (\Cref{sec:decMon});
    \item a synthesis function that returns a correct centralized monitor for every formula without least fixed points (\Cref{sec.cenMon});
    \item a synthesis function that returns a correct (decentralized) choreography of communicating monitors  for every formula without least fixed points that has no location quantifier within a fixed point operator (\Cref{sec:decMon}); and
    \item a methodology to prove the correctness of a synthesis of communicating monitors, by establishing a list of desirable properties and relating the behavior of the decentralized monitors to that 
    of the corresponding centralized monitor (\Cref{def:principled-synthesis,thm:principled-bisimulation}).
\end{itemize}

This is an extended and revised version of \cite{AAAFGW24}. With respect to the extended abstract, here we provide 
\ifarxiversion
all the details and proofs,
\fi
we better justify \HypermuHML (through the formulation of a version of the non-interference property taken from \cite{hyperproperties} -- see Example \ref{ex:GMNI}), we elaborate on the safety properties monitorable in our framework (\Cref{sec:monHyperprop}), and we also explicitly include the decentralized synthesis of the diamond operator, that was omitted from the extended abstract (see however the drawback of omitting it in \Cref{rem:diamond}).

\section{The Model and the Logic}
\label{sec:model}

Let $\Act$ be a finite set of actions with at least two elements\footnote{When $\Act$ is a singleton, every property in the logic becomes 
equivalent to true or false.
}, ranged over by $a,b$;
the set of (infinite) traces over $\Act$ is $\Trc=\Act^\omega$, ranged over by $t$.
Given a finite and non-empty set of locations $\Loc$ ranged over by $\ell$, a hypertrace $T$ on $\Loc$ is a function from $\Loc$ to $\Trc$; the set of hypertraces on $\Loc$ is denoted by $\Htrc_\Loc$.
$\Loc$ and $\Act$ are fixed throughout this paper.
%
A hypertrace describes
a (distributed) system with $ |\Loc|$ users, and every user is located at a unique location chosen from $\Loc$.
A system behavior is 
captured
by a hypertrace $T$ on $\Loc$, mapping every user to the trace they perform.

For $t,t'\in \Trc$, we write $t\xrightarrow{a}t'$ whenever $t=at'$.
Let $A: \Loc \rightarrow \Act$; for $T, T' \in \Htrc_\Loc$, we write $T\xrightarrow{A}T'$ whenever $T(\ell)\xrightarrow{A(\ell)}T'(\ell)$, for every $\ell \in \Loc$.
Notice that, for each $T$, there is a {\em unique} pair $A$ and $T'$ such that $T\xrightarrow{A}T'$:
more precisely, for every $\ell \in \Loc$, we have that $A(\ell) = a$ and $T'(\ell) = t'$, whenever $T(\ell) = at'$. We denote the $A$ and $T'$ just defined by $\hd(T)$ and $\tl(T)$ respectively.
For a partial function $f: D \rightharpoonup E$ (where $D$ and $E$ are sets ranged over by $d$ and $e$, respectively), we denote by $\dom(f)$ the set $\{d \in D \mid f(d) \mbox{ is defined}\}$ and
by $\range(f)$ the set $\{e \mid \exists d \in \dom(f).~f(d) = e\}$.
Notation $f[d \mapsto e]$  denotes the (partial) function mapping $d$ to $e$ and behaving like $f$ otherwise.

\subsection{The Logic \HypermuHML}
\label{sec:fulllogic}

We consider \HypermuHML as the logic to specify {\em hyperproperties}.
 We assume two disjoint and countably infinite sets $\Pi$ and $V$ of {\em location variables} and {\em recursion variables}, ranged over by $\pi$ and $x$, respectively. Formulas of \HypermuHML are constructed as follows:
\[
	\varphi ::= \ttt \mid \ff \mid \varphi \land \varphi \mid \varphi \lor \varphi \mid \max x. \varphi \mid \min x. \varphi \mid x
\mid \exists \pi. \varphi \mid \forall \pi. \varphi \mid \pi=\pi \mid \pi\neq\pi \mid [a_\pi]\varphi \mid \langle a_\pi\rangle\varphi
\]

Apart from the basic boolean constructs,
we include the greatest and and least fixed-point operators to describe unbounded and/or infinite behaviors in a finitary manner,\footnote{In LTL, this behavior is captured by the ‘Until' and ‘Release' operators, but these are less expressive than fixed-points; see \cite{AcetoAFIL21:sosym}.} existential/universal quantifiers and equality/inequality tests on location variables, and the usual Hennessy-Milner modalities where
$[a_\pi]$ stands for `necessarily after $a$ at the location bound to $\pi$',  and $\langle a_\pi\rangle$ denotes  `possibly after $a$ at the location bound to $\pi$'.
%
A formula is said to be {\em guarded} if every
recursion variable appears within the scope of a modality within its fixed-point binding. All formulas are assumed to be guarded (without loss of expressiveness 
\cite{DBLP:journals/jacm/KupfermanVW00}).
We write $\FVloc(\varphi)$ to denote the free location variables of $\varphi$, and $\FVrec(\varphi)$ for the free recursion variables.

\begin{remark} We consider formulas where bound location variables are all pairwise distinct (and different from the free variables); hence, the formula $\forall\pi.[a_\pi]\exists\pi.\varphi$ denotes the formula $\forall\pi.[a_\pi]\exists\pi'.(\varphi\{^{\pi'}\!/_\pi\})$, where $\varphi\{^{\pi'}\!/_\pi\}$ stands for the capture-avoiding substitution of $\pi'$ for $\pi$ in $\varphi$. A similar notation for other kinds of substitutions is used throughout the paper. 
\exqed
\end{remark}

The semantics for formulas $\varphi$ in \HypermuHML is defined over $\Htrc_\Loc$ through the function $\sem{-}^\rho_\sigma$, as shown in \Cref{tab:semantics_hyper_mu_HML}, by exploiting two partial functions: $\rho\colon V\rightharpoonup 2^{\Htrc_\Loc}$, which assigns a set of hypertraces on $\Loc$ to all free recursion variables of $\varphi$, and $\sigma\colon \Pi\rightharpoonup  \Loc$, which assigns a location to all free location variables of $\varphi$. In what follows, we tacitly assume that the free recursion and location variables in a formula $\varphi$ are always included in $\dom(\rho)$ and $\dom(\sigma)$, respectively. 
\revOne{Function $\rho$ gives the semantics of recursion variables. Intuitively, $\rho(x)$ is the set of hypertraces that are assumed to satisfy $x$, as formally described on the first line of Table 1. Formulae of the form $\maxx {x}.\psi$ and $\min x.\psi$ are interpreted as the greatest and least fixed points of the function induced by their body $\psi$. Function $\sigma$ keeps track of the location associated with each location
variable; it is extended every time a new variable $\pi$ is introduced (through $\exists\pi$ or $\forall\pi$, resp.) to check whether the body of the quantification holds for some (resp., for all) locations.

All the other operators have quite a straightforward semantics. In particular, }
a formula $\langle a_\pi\rangle\varphi$ holds true at hypertrace $T$ if the trace in $T$ at the location bound to $\pi$ starts with an $a$ and $\tl(T)$ satisfies $\varphi$; by contrast, a formula $[ a_\pi]\varphi$ can also hold true if the trace in $T$ at the location associated to $\pi$ does not start with an $a$. 

\begin{table}
    \hrule
\begin{align*}
    & \sem{\ttt}^\rho_\sigma=\Htrc_\Loc \qquad\qquad\qquad\qquad \sem{\ff}^\rho_\sigma=\emptyset
    && \!\!\!\!\sem{x}^\rho_\sigma=\rho(x)
    \\
    & \sem{\varphi\wedge\varphi'}^\rho_\sigma =\sem{\varphi}^\rho_\sigma \cap\sem{\varphi'}^\rho_\sigma
    &&\!\!\!\! \sem{\varphi\vee\varphi'}^\rho_\sigma =\sem{\varphi}^\rho_\sigma \cup\sem{\varphi'}^\rho_\sigma
    \\
    &\sem{\maxx {x}.\psi}^\rho_\sigma=\bigcup\{S\mid S\subseteq \sem{\psi}^{\rho[x\mapsto S]}_\sigma\}
    && \!\!\!\!\sem{\minn {x}.\psi}^\rho_\sigma =\bigcap\{S\mid S\supseteq \sem{\psi}^{\rho[x\mapsto S]}_\sigma\}
   \\
       & \sem{\exists \pi.\varphi}^\rho_\sigma =\bigcup_{\ell\in \Loc} \sem{\varphi}^\rho_{\sigma[\pi\mapsto \ell]}
    &&
     \!\!\!\!\sem{\forall \pi.\varphi}^\rho_\sigma =\bigcap_{\ell\in \Loc} \sem{\varphi}^\rho_{\sigma[\pi\mapsto \ell]}
    \\
    & \sem{\pi=\pi'}^\rho_\sigma =\begin{cases}
        \Htrc_\Loc & \text{if } \sigma(\pi)=\sigma(\pi') \\ \emptyset & \text{otherwise}
        \end{cases} &&
     \!\!\!\!\sem{\pi\neq\pi'}^\rho_\sigma =\begin{cases}
            \Htrc_\Loc & \text{if } \sigma(\pi)\neq\sigma(\pi') \\ \emptyset & \text{otherwise}
            \end{cases}
     \\
    &  \sem{[a_\pi]\varphi}^\rho_\sigma = \{T \mid \hd(T)(\sigma(\pi)) = a \ \text{ implies } \tl(T)\in\sem{\varphi}^\rho_\sigma\}
    \qquad &&  
    \!\!\!\!\sem{\langle a_\pi\rangle\varphi}^\rho_\sigma = \{T \mid \hd(T)(\sigma(\pi))=a \ \wedge\ \tl(T) \in\sem{\varphi}^\rho_\sigma)\}
\end{align*}
\hrule\vspace{0.5em}
\caption{The semantics of \HypermuHML.}
\label{tab:semantics_hyper_mu_HML}
\end{table}

Whenever $\varphi$ is {\em closed} (i.e., without any free variable), the semantics is given by $\sem\varphi^\emptyset_\emptyset$, where $\emptyset$ denotes the partial function with empty domain. Notationally, we shall simply write $\sem\varphi$ instead of $\sem\varphi^\emptyset_\emptyset$. We say that $T$ satisfies the closed formula $\varphi$ if $T \in \sem{\varphi}$.

\begin{example}
\label{ex:One}
For example, consider the set of actions $\{a,b\}$; then, 
the hyperproperty 
\begin{equation}
\label{exp:property}
\varphi_a=
\forall \pi. \max x. \bigl(\langle b_\pi \rangle x \ \vee\ \exists \pi'. (\pi'\neq\pi\ \wedge\ \langle a_{\pi'} \rangle x)\bigr)
\end{equation}
is a consensus-type property stating that, 
at every position of every trace, whenever there is an $a$ there is another trace that also has $a$. 
Using the semantic definition 
of the logic, it is not hard to see that the hypertrace $T_1$ over the set of locations $\{\ell_1,\ell_2,\ell_3\}$ that maps $\ell_1$ to $a^\omega$, $\ell_2$ to $b a^\omega$ and $\ell_3$ to $(b a)^\omega$ does not satisfy the property $\varphi_a$: what breaks the property is the first position. On the other hand, the hypertrace $T_2$ that maps $\ell_1$ to $a^\omega$, $\ell_2$ to $(ab)^\omega$ and $\ell_3$ to $(b a)^\omega$ does satisfy $\varphi_a$ because at each position there are two traces that exhibit an $a$.
\exqed
\end{example}

\begin{example}
\label{ex:GMNI}
As a second example, we show how the well-known property of {\em non-interference}, in the formulation by Goguen and Meseguer from \cite{hyperproperties}, can be written in our logic.
The property requires that ``commands issued by users holding high clearances be removable without affecting observations of users holding low clearances''. They model this security policy as a hyperproperty by treating commands as inputs and observations as outputs, and by requiring that a system contains, for any trace $t$, a corresponding trace $t'$ with no high inputs, yet with the same low outputs as $t$:
$$
GMNI = \{ T \mid \forall t \in T .\exists t' \in T\ .\ (H_{in}(t') = \epsilon \wedge L_{out}(t) = L_{out}(t'))\} ~,
$$ 

In our setting, we assume that actions can be either high or low, i.e. that $\Act$ can be bipartitioned into $\{l^1,\ldots,l^m\}$ and $\{h^1,\ldots,h^n\}$.
Then, we define the auxiliary formula
$$
low\_agree(\pi,\pi') = 
max~x. \left(\ 
\bigwedge_{i=1}^m \big( \langle l^i_{\pi}\rangle\ttt \leftrightarrow \langle l^i_{\pi'}\rangle\ttt \big) 
\ \wedge\ 
\bigvee_{a \in \Act } [a_\pi]x
\ \right).
$$
This formula says that the two given locations $\pi$ and $\pi'$ exhibit the same low behaviour\revOne{, and this holds repeatedly, after any possible action (being it high or low)}.
By using it, {\em GMNI} can be formulated in \HypermuHML\ as follows:
$$
\varphi_{GMNI} = \forall \pi . \exists \pi'. \left(\ 
low\_agree(\pi,\pi')\ \wedge\ 
max~x. \left( \bigwedge_{i=1}^n[h^i_{\pi'}]\ff\ \wedge\ \bigwedge_{i=1}^m [l^i_{\pi'}] x \right)\ 
\right). 
$$
Formula $\varphi_{GMNI}$ states that, for every trace, there must exist another trace such that:
(1) the two completely agrees on the low actions (as expressed by the $low\_agree$ conjunct);
and (2) the second one cannot perform any high action (as expressed by the $max~x$ conjunct). 
\end{example}

\subsection{On the Expressiveness of \HypermuHML}
\label{section:expressiveness}

The logic \HypermuHML adapts linear-time \muHML~\cite{Larsen-paper} to express properties of hypertraces, just as HyperLTL and HyperCTL*~\cite{hyperLTL} are variations on LTL \cite{LTL} and CTL* \cite{CTLstar}, respectively, interpreted over hypertraces. 
It is well known that \muHML is more expressive than LTL and CTL*~\cite{LTL_expressivenes};
in a companion paper, we are working to formally show that the same result holds also in the hyper-setting, i.e. that 
\HypermuHML can express all properties that can be described using HyperLTL and HyperCTL*.
Here we confine ourselves to show that
\HypermuHML can express properties that cannot be described using HyperLTL and HyperCTL*. 

First, we recall that Wolper showed in \cite{LTL_expressivenes} that the property ``event $a$ occurs at all even positions in a trace'' cannot be expressed in LTL (see \cite[Corollary 4.2] {LTL_expressivenes} that is based on Theorem 4.1 in that reference).
 We will refer to this property as $\varphi_e$, where ``$e$'' stands for even,
 %
%
and 
adapt it to a hypertrace setting for proving strictness of the inclusion of HyperLTL in \HypermuHML.

\begin{example}
\label{ex:Two}

Let $\varphi_e^h$ be the hyperproperty on the set of actions $\{a,b\}$ that results 
from adding an existential trace quantifier $\exists \pi$ at the beginning of $\varphi_e$, and 
replacing all modalities with $\pi$-indexed ones: 
\begin{equation}
\label{ex:Whprop}
\varphi^h_e = \exists \pi.
\max x. \bigl([a_{\pi}]\langle a_{\pi} \rangle x \wedge [b_{\pi}]\langle a_{\pi} \rangle x\bigr) 
\end{equation}
This is a liveness property that 
describes the periodicity of events; when evaluated over singleton hypertraces, it coincides with the evaluation of $\varphi_e$. 
\exqed
\end{example}

\begin{theorem}
\label{prop:expressiveness}
\HypermuHML is more expressive than HyperLTL. 
\end{theorem}
\begin{proof}
  
By contradiction,
assume that a formula $\varphi_{h-LTL}$ expresses exactly that there exists a trace in a hypertrace such that $a$ holds on even positions.  
We cannot trivially claim that this would be expressed with a single quantifier $\exists \pi$, but we know that over singleton hypertraces, say for $T= \{t_0 \}$, $T \models \varphi_{h-LTL}$ iff $T \models \varphi^h_e $.
However, since $T$ contains only a single trace, we know that all the trace variables in $\varphi_{h-LTL}$ must be mapped to $t_0$.
Consequently, all propositional variables that occur in $\varphi_{h-LTL} $ are now interpreted to hold on the same trace. 
Therefore, for this variable mapping, we get an LTL formula which expresses exactly that the single trace (that all trace variables have been mapped to) satisfies Wolper's property. 

Therefore, we can express the same formula in plain LTL, by replacing all propositional variables with non-trace quantified ones, and arrive at a contradiction. 
%
\end{proof}

For HyperCTL* we exploit the possibility of the latter to use quantifiers in any part of a formula.
This has already been shown in the hyperproperty $\varphi_a$ defined in \eqref{exp:property}: such a formula can potentially spawn an unbounded number of quantifiers,  by unfolding the recursion when encountering $a$ events. 

\begin{theorem}
\HypermuHML is more expressive than HyperCTL*. 
\end{theorem}
\begin{proof}
To prove this, we refer the reader to \cite{epistemic_hyper}, where it is shown that the property ``there
is an $n \geq 0$ such that $a \neq t(n)$ for every $t \in T$'' is not expressible in HyperCTL* (notice that this property was also shown not to be expressible in HyperLTL by \cite{Finkbeiner017}).
In our logic this property is expressible (over the set of actions $\{a,b\}$) with the formula: 
$$\min x. ((\forall \pi. \langle b_{\pi}\rangle \ttt )\vee (\forall \pi'. ( [ a_{\pi'}]x \wedge [b_{\pi'}]x))~. $$

Note that this formula can also be simplified. We choose to write it like this so that it is explicit that either all traces have $b$, or all traces take a step. However the following formulation is also valid, and only uses one quantifier:
$$\min x. ((\forall \pi . ((\langle b_{\pi}\rangle \ttt )\vee ( [ a_{\pi}]x \wedge [b_{\pi}]x))~. 
\vspace*{-.7cm}
$$
\end{proof}

To conclude, we remark that this additional expressiveness of \HypermuHML\ is present in the fragments for which we synthesize monitors, as we shall see in the next sections. 

\section{Centralized Monitoring}
\label{sec.cenMon}

We first present a centralized, omniscient monitor that enforces hyperproperties expressed in \HypermuHML; here, the monitoring algorithm can compare all traces as needed by simply accessing different parts of its local memory. 
 
 \subsection{Centralized Monitors: Syntax and Operational Semantics}
 \label{sec:cmonitors}

The set of centralized monitors $\CMon$ is given by the following grammar:
	\begin{align*}
\CMon \ni m &::= \yes ~\mid~ \no ~\mid~ \vend ~\mid~ a_\ell.m ~\mid~ m+m ~\mid~ m \por m ~\mid~ m\pand m ~\mid~ \rec{x}.m ~\mid~ x
\end{align*}

{\em Verdicts}, taken from the set $\{\yes, \no, \vend\}$ and ranged over by $v$, are the possible outcomes of the monitor observation: it can be $\yes$, if the hyperproperty associated to the monitor is satisfied, $\no$, it if is not, or $\vend$, if the monitor has not enough evidence to conclude either $\yes$ or $\no$. The prefixed monitor $a_\ell.m$ checks if $\ell$ is currently performing an $a$ and, in case, continues as $m$. We can compose monitors by using
\revOne{
non-deterministic choice, used to specify alternative and mutually exclusive observation capabilities, and
parallel-and and parallel-or, used to build parallel monitors that combine their verdicts according to a conjunctive or disjunctive strategy, respectively.
}
Notationally, we denote with $\pandor$ any of $\pand$ and $\por$.
Finally, monitors can also be recursively defined.

\begin{table}
\hrule
\begin{mathpar}
 v \xrightarrow A v \and
\inferrule{A(\ell) = a}{a_\ell.m \xrightarrow A m}
\and
\inferrule{A(\ell) \neq a}{a_\ell.m \xrightarrow A \vend}
\and
\inferrule{m\{^{\rec x.m}/_x\} \xrightarrow A m'}{\rec x.m \xrightarrow A m'}
\and
\inferrule{m \xrightarrow A m'}{m+n \xrightarrow A m'}
 \and
\inferrule{n \xrightarrow A n'}{m+n \xrightarrow A n'}
\and
\inferrule{m \xrightarrow A m' \\ n \xrightarrow A n'}{m \pandor n \xrightarrow A m' \pandor n'}
\end{mathpar}
\hrule\vspace{0.5em}
\caption{The operational semantics for centralized monitors, where $\pandor \in \{\pand,\por\}$.}
\label{tab:centralized_operational_semantics}
\end{table}

\begin{table}
 \hrule
\begin{minipage}{0.67\textwidth}
\center
\[\begin{array}{llll}
    v \valto v
    &
      \inferrule{m \valto \vend \\\\ n\valto \vend}{m\pandor n\valto \vend}
   &
    \inferrule{m \valto \yes}{m \por n \valto \yes}
     & 
     \inferrule{m \valto \no}{m\pand n\valto \no}
     \vspace*{.2cm}
     \\
    \inferrule{m \valto v}{m + n\valto v}
     & 
     \inferrule{m \valto \no \\\\ n\valto v}{m \por n \valto v}
     &
     \inferrule{m \valto \yes \\\\ n\valto v}{m \pand n \valto v}
      &
      \inferrule{m\{^{\rec x.m}/_x\} \valto v}{\rec x.m \valto v}
      \vspace{0.5em}
\end{array}\]
\caption{Verdict evaluation for centralized monitors (up to commutativity of $+$, $\pand$, and $\por$).}
\label{tab:verdict_evaluation_semantics}
\end{minipage}
\quad
\begin{minipage}{0.3\textwidth}
\vspace{0.2em}
\[\begin{array}{lll}
    \inferrule{m \xrightarrow{A} m' \\ T\xrightarrow{A}T'}
    { m\triangleright T \rightarrowtail m'\triangleright T'}
     \vspace*{.5cm}
   \\
   \inferrule{m \valto v}
   {m\triangleright T \rightarrowtail v}
\vspace{0.5em}
\end{array}\]
\caption{The instrumentation rules for centralized monitors.}
\label{tab:instrumentation}
\end{minipage}
\hrule
\end{table}

The operational semantics of centralized monitors is given in \Cref{tab:centralized_operational_semantics}
\revOne{and follows the intuitive explanations we gave above for the different operators. In particular, a monitor that waits for an action at some location (as prescribed by $a_\ell$) can either see that action at $\ell$  (as stated by $A$) and proceed in its observation, or it does not observe that action and stops its monitoring activity, by reporting $\vend$. Recursive monitors should be unfolded to evolve. A non-deterministic choice $m + n$ can initially behave like $m$ and discard $n$ in doing so or vice versa. Finally, once issued, a verdict never changes; possibly different verdicts obtained in different parallel branches of the monitor can be composed conjunctively or disjunctively.}
This is represented by the judgement $\valto$, whose intended use is to evaluate monitors and reach a verdict, whenever possible. The rules are given in \Cref{tab:verdict_evaluation_semantics}.
\revOne{They state that an $\vend$ verdict can be obtained only when combining two (or more) end verdicts, that $\yes$ and $\no$ are the absorbing elements for $\por$ and $\pand$, respectively, and the neutral elements for $\pand$ and $\por$, respectively, that verdict evaluation of a non-deterministic monitor is non-deterministic as well, and that recursive monitors must be unfolded before they can be properly evaluated.}
Also notice that there can be multiple ways to infer the same verdict for the same monitor: \changed{e.g., for $\yes \por \no$ we can either use the third rule of the first line or the (symmetric version of the) second rule from the second line of \Cref{tab:verdict_evaluation_semantics}.} However, the inferred value is of course the same (i.e., $\yes$, in the previous situation).

We instrument a monitor $m$
on a hypertrace $T$  based on the rules of \Cref{tab:instrumentation}. 
As usual, we write $\rightarrowtail^* $ for the reflexive-transitive closure of $\rightarrowtail$. 

\subsection{From Formulas to Centralized Monitors}
\label{sec:centrSyn}

We derive monitors for the subset of formulas without least fixed-points, denoted with \WTFPHypermuHML.
More precisely, given a formula $\varphi$, we want to derive a monitor that, when monitoring a hypertrace $T$, returns $\no$ if and only if $T$ does not belong to the semantics of $\varphi$; furthermore, if it returns $\yes$, then $T$ belongs to the semantics of $\varphi$.
All regular properties of infinite traces that can be monitored for violations with the aforementioned guarantees can be expressed without using least fixed-point operators 
(see the maximality results presented in
~\cite[Proposition~4.18]{AcetoPOPL19} and~\cite[Theorem~5.2]{AcetoAFIL21:sosym} in the setting of logics interpreted over infinite traces).
%
Intuitively, we use least fixed-points to describe liveness properties, whose violation does not have a finite witness in general.

The definition of the synthetized monitor is given by induction on $\varphi$. This definition is parametrized by a  partial function $\sigma$, assigning a location to all the free location variables of $\varphi$; when $\varphi$ is closed, we consider $\Mcv{\varphi}{\emptyset}$. The formal definition is 
given in \Cref{tab:synthesis}.
\revOne{A monitor synthetized from a greatest-fixed-point formula is itself recursive and intuitively checks whether some unfolding of the formula is violated. The monitor for $\varphi\wedge \varphi'$ (respectively, $\varphi\vee \varphi'$) is obtained as the parallel-and (respectively, the parallel-or) of the monitors synthesized from $\varphi$ and  $\varphi'$. Universal and existential quantifiers are treated as conjunctions and disjunctions, respectively, and use the function $\sigma$ to assign values to the newly introduced location variable. 
Finally, the monitors for formulae of the form $[a_\pi].\varphi$ and  $\langle a_\pi \rangle.\varphi$ look for an occurrence of action $a$ at (the location bound to) $\pi$; if such action is observed at that location, then the monitor proceeds by checking the rest of the formula, otherwise the monitor for  $\langle a_\pi \rangle.\varphi$ returns $\no$ whereas the monitor for $[a_\pi].\varphi$ returns $\yes$.
Notice that we are using the choice operator ‘$+$' here to consider only one of the possible observations (corresponding to the action occurring at the location $\sigma(\pi)$) and discarding all the other ones.}


\begin{example}
\label{ex:Wcentr}
Let 
$\mathcal{L} = \{ 1, 2\}$ and $\Act = \{a,b\}$, and consider
the formula \eqref{ex:Whprop}.
The 
monitor synthesis in \Cref{tab:synthesis} produces the following monitor $m$ when applied to that formula:
$$ m  =  \bigoplus_{{\ell \in \{1,2\}}}  \rec x. (
                    (a_{\ell}. (a_{\ell} .x + b_{\ell}.\no) + b_{\ell} . \yes) \otimes (b_{\ell}. (a_{\ell} .x + b_{\ell}.\no) + a_{\ell} . \yes)) .
$$
When monitor $m$ is instrumented with the hypertrace $T$ mapping
location $1$ to $a^\omega$ and location $2$ to $(ab)^\omega$, the
verdict $\no$ cannot be reached: indeed, $T$ satisfies the formula $\varphi$
since the trace at location $1$ has $a$ at all positions. On the
other hand, when $m$ is instrumented with the hypertrace $T'$ mapping
location $1$ to $b^\omega$ and location $2$ to $(ab)^\omega$, the $\no$
verdict is reached after the monitor has observed the first two
actions at locations 1 and 2; this is in line with the fact that $T'$
does not satisfy $\varphi^h_e$. \exqed
\end{example}

\begin{table}
\hrule
\[\def\arraystretch{1.3}\begin{array}{r@{\,}lr@{\,}l}
    & \Mc{\ttt} = \yes \quad\quad
    \Mc{\ff} = \no    \quad
& \Mc{x} = x \qquad
\Mc{\max x.\varphi} = \rec x.\Mc{\varphi} 
\\
& \Mc{\varphi\wedge \varphi'} = \Mc{\varphi} \pand \Mc{\varphi'} \quad
& \Mc{\varphi\vee\varphi'} =\Mc{\varphi} \por \Mc{\varphi'}
\\ 
& \Mc{\forall\pi.\varphi} =  \bigotimes_{\ell \in \Loc} \Mcv{\varphi}{\sigma[\pi \mapsto \ell]}  
& \Mc{\exists\pi.\varphi} =  \bigoplus_{\ell \in \Loc} \Mcv{\varphi}{\sigma[\pi \mapsto \ell]}
\\
& \Mc{\pi=\pi'} = \begin{cases}
\yes~ \mbox{if } \sigma(\pi)=\sigma(\pi')
\\
\no~ \mbox{otherwise }
\end{cases}
&
\Mc{\pi\neq\pi'}  = \begin{cases}
\yes ~ \mbox{if } \sigma(\pi)\neq\sigma(\pi') \\
\no ~ \mbox{otherwise }
\end{cases}
\\ & \Mc{[a_\pi]\varphi} = 
a_{\sigma(\pi)}.\Mc{\varphi} + \sum_{b \neq a} b_{\sigma(\pi)}.\yes &
\Mc{\langle a_\pi \rangle\varphi} =  
    a_{\sigma(\pi)}.\Mc{\varphi} + \sum_{b \neq a} b_{\sigma(\pi)}.\no 
\end{array}\]
\hrule \vspace{0.5em}
\caption{Centralized monitor synthesis.}
\label{tab:synthesis}
\end{table}

The main results of this section are that the centralized monitors
synthesized from formulas report sound verdicts and their verdicts are
complete for formula violations. We refer the reader to~\cite{AcetoAFIL21:sosym} for a discussion on notions of correctness for monitors and the significance of soundness and violation-completeness.
\revTwo{The main point is that, already in the setting of  \muHML, it is possible to have soundness and completeness only for formulae without fixed points; hence,  satisfaction completeness and violation completeness were introduced to weaken the completeness requirement and having the possibility of monitoring also some recursive formulae.}

To prove soundness (\Cref{thm:soundCentr}), we first need three lemmata: 
the first one claims soundness of synthesized monitors when they have not taken any transition yet;
the second one claims that  synthesized monitors always transition to other synthesized monitors;
the last one claims that the relation between the hypertrace and the semantics of the formula are carried over through the instrumentation steps.
The proofs of these results have been relegated to \Cref{app:soundness}, to streamline reading.


\begin{textAtEnd}[allend, category=soundness]
We need a few preliminary lemmata to prove  \Cref{thm:soundCentr}.
We first show that centralized monitors never output the $\vend$ verdict if they are $\vend$-free themselves :
\begin{lemma}\label{lemma:valtoverdict}
    If $m$ is an $\vend$-free monitor and $m\valto v$, then $v=\no$ or $v=\yes$.
\end{lemma}

We link the semantics of a formula to the special case where the synthesized centralized monitor is a verdict, which we use in various subsequent proofs:
\begin{lemma}\label{lemma:verdict}
If $\Mc{\varphi}=\yes$, then $\sem{\varphi}_\sigma^\rho=\Htrc_\Loc$, for all $\rho$; if $\Mc{\varphi}=\no$, then $\sem{\varphi}_\sigma^\rho=\emptyset$, for all $\rho$.
\end{lemma}
\ifarxiversion
\begin{proof}
    We proceed by induction on $\varphi$. If $\Mc{\varphi}=\yes$, then the definition of the synthesis function (see \Cref{tab:synthesis}) yields the following base cases: $\varphi=\ttt$, or $\varphi=(\pi=\pi')$ and $\sigma(\pi)=\sigma(\pi')$, or $\varphi=(\pi\neq\pi')$ and $\sigma(\pi)
    \neq\sigma(\pi')$. In all cases we immediately conclude that $\sem{\varphi}_\sigma^\rho=\Htrc_\Loc$, for all $\rho$. 

    In the inductive cases we only need to consider $\Loc=\{\ell\}$ and $\varphi=\forall \pi.\varphi'$ or $\varphi=\exists\pi.\varphi'$ where $\Mc{\varphi}=\Mcv{\varphi'}{\sigma[\pi\mapsto\ell]}=\yes$ (because otherwise $\Mc{\varphi}\neq \yes$). In both cases we obtain from the induction hypothesis that $\sem{\varphi}_\sigma^\rho=\sem{\varphi'}_{\sigma[\pi\mapsto\ell]}^\rho=\Htrc_\Loc$. 

    The proof where $\Mc{\varphi}=\no$ is similar.
\end{proof}
\fi

We show that if synthesized monitors of two formulas are equal, then these two formulas must have the same semantics, which is used to prove \Cref{cor:semantics}.
\begin{lemma}\label{lemma:nochange}
    Let $\varphi$ and $\psi$ be (possibly open) formulas. If $\Mc{\varphi}=\Mcv{\psi}{\sigma'}$, then $\sem{\varphi}_\sigma^\rho=\sem{\psi}_{\sigma'}^\rho$, for all $\rho$.
\end{lemma}
\begin{proof}
We proceed by induction on $\varphi$. As a first base case we consider $\varphi=\ttt$. We then have $\Mc{\varphi}=\yes=\Mcv{\psi}{\sigma'}$. We can immediately conclude that, for all $\rho$, we have $\sem{\varphi}_\sigma^\rho=\sem{\psi}_{\sigma'}^\rho$ using \Cref{lemma:verdict}.
The base cases for $\varphi=\ff$, $\varphi=(\pi=\pi')$ and $\varphi=(\pi\neq\pi')$ are similar. 

The last base case is for $\varphi=x$ for $x$ a recursion variable. Then $\Mc{\varphi}=x=\Mcv{\psi}{\sigma'}$, from which we can conclude that $\psi=x$ or $\psi=Q_1\dots Q_n.\psi'$ for $Q_i$ a universal or existential quantifier, $n\geq 1$, $\Loc=\{\ell\}$ and $\psi'=x$. In the first case it follows immediately that, for all $\rho$, we have $\sem{\varphi}_\sigma^\rho=\sem{\psi}_{\sigma'}^\rho$. In the second case we obtain that $\sem{\psi}_\sigma^\rho=\sem{\psi'}_{\sigma[\pi\mapsto\ell]}^\rho=\rho(x)=\sem{\varphi}$.

\ifarxiversion
We now proceed with the inductive cases. 
\fi 
We first consider $\varphi=[a_\pi]\varphi'$. Thus $\Mc{\varphi}=a_{\sigma(\pi)}.\Mc{\varphi'} + \sum_{b \neq a} b_{\sigma(\pi)}.\yes=\Mcv{\psi}{\sigma'}$. We distinguish two cases: either $\psi=[a_{\pi'}]\psi'$, $\sigma(\pi)=\sigma'(\pi')$ and $\Mc{\varphi'}=\Mcv{\psi'}{\sigma'}$, or $\Act=\{a,b\}$, $\psi=\langle b_{\pi'}\rangle \psi'$, $\sigma(\pi)=\sigma'(\pi')$, $\Mc{\varphi'}=\no$ and $\Mcv{\psi'}{\sigma'}=\yes$. In the first case we use the induction hypothesis to obtain that, for all $\rho$, we have $\sem{\varphi'}_\sigma^\rho=\sem{\psi'}_{\sigma'}^\rho$. The definition of the semantics gives the desired result immediately.
In the second case, we use \Cref{lemma:verdict} to obtain that, for all $\rho$,
we have $\sem{\varphi}_\sigma^\rho=\{T \mid \hd(T)(\sigma(\pi)) = a \ \text{ 
    implies }\tl(T)\in\emptyset\}=\{T \mid \hd(T)(\sigma'(\pi'))=b \ \wedge\ \tl
    (T)  \in\Htrc_\Loc\}=\sem{\psi}_{\sigma'}^\rho$.                 

\ifarxiversion
The case for $\varphi=\langle a_\pi\rangle\varphi'$ is similar. 

Next we consider $\varphi=\maxx {x}.\varphi'$. Hence, $\Mc{\varphi}=\rec x.\Mc{\varphi'}=\Mcv{\psi}{\sigma'}$. From this we conclude that $\psi=\maxx{x}.\psi'$ and $\Mc{\varphi'}=\Mcv{\psi'}{\sigma'}$ for some $\psi'$ . We use the induction hypothesis to obtain that, for all $\rho$, we have $\sem{\varphi'}_\sigma^\rho=\sem{\psi'}_{\sigma'}^\rho$. We derive 
\begin{align*}
    \sem{\maxx {x}.\varphi'}^\rho_\sigma=\bigcup\{S\mid S\subseteq \sem{\varphi'}^{\rho[x\mapsto S]}_\sigma\}\stackrel{IH}{=}\bigcup\{S\mid S\subseteq \sem{\psi'}^{\rho[x\mapsto S]}_\sigma\}= \sem{\maxx {x}.\psi'}^\rho_\sigma
\end{align*}
\fi 

For the case where $\varphi=\varphi_1\wedge \varphi_2$ and the case where $\varphi=\forall \pi.\varphi'$, we need an intermediary result: 
\begin{align}\label{result:intermediary}
\text{If $\Mc{\varphi}=\Mcv{\varphi_1}{\sigma_1}\pand\dots\pand \Mcv{\varphi_k}{\sigma_k}$ for $k\geq 2$, then }
    \sem{\varphi}_\sigma^\rho=\bigcap_{i=1}^k \sem{\varphi_i}_{\sigma_i}^\rho,
    \text{for all $\rho$.}
\end{align}
We will prove this claim at the end, but first finish the main proof using this result. Consider $\varphi=\varphi_1\wedge \varphi_2$. Hence, $\Mc{\varphi}=\Mc{\varphi_1}\pand \Mc{\varphi_2}=\Mcv{\psi}{\sigma'}$. We use intermediary result (\ref{result:intermediary}) to conclude. 
Now consider the case where $\varphi=\forall \pi.\varphi'$ and $|\Loc|=1$:
we get $\Mc{\varphi}=\Mcv{\varphi'}{\sigma[\pi\mapsto \ell]}=\Mcv{\psi}{\sigma'}$. From the induction hypothesis we obtain that, for all $\rho$,
we have that $\sem{\varphi'}_{\sigma[\pi\mapsto \ell]}^\rho=\sem{\psi}_{\sigma'}^\rho$. 
Since $\Loc=\{\ell\}$, we have, for all $\rho$, that $\sem{\varphi}_\sigma^\rho=\sem{\varphi'}_{\sigma[\pi\mapsto \ell]}^\rho$, 
and we are done. 
If $\Loc=\{\ell_1,\ldots,\ell_k\}$ for $k\geq 2$, we have that $\Mc{\varphi}=\Mcv{\varphi_1}{\sigma[\pi\mapsto \ell_1]}\pand\dots\pand \Mcv{\varphi_k}{\sigma[\pi\mapsto \ell_k]}=\Mcv{\psi}{\sigma'}$ 
and we use intermediary result (\ref{result:intermediary}) to conclude. 

\ifarxiversion
The case for $\varphi=\varphi_1\vee\varphi_2$ and $\varphi=\exists \pi.\varphi'$ can be handled in a similar manner as the two previous cases, with an additional result dual to (\ref{result:intermediary}).
\fi 

Now we prove claim (\ref{result:intermediary}) by induction on $\varphi$. Without loss of generality we can assume that for all $i\in\{1,\dots,k\}$ we have that $\Mcv{\varphi_i}{\sigma_i}\neq m\pand n$ for all $m,n$ (we refer to this as maximality).

We distinguish the case where $k=2$ or $k>2$. If $k=2$, we get from the synthesis function that $\varphi$ is a conjunction or a universal formula. In case $\varphi=\varphi'\wedge \varphi''$, we obtain that $\Mc{\varphi'}\pand \Mcv{\varphi''}{\sigma}=\Mcv{\varphi_1}{\sigma_1}\pand \Mcv{\varphi_2}{\sigma_2}$. We use the induction hypothesis of the main proof to conclude that, for all $\rho$, we have $\sem{\varphi'}_\sigma^\rho=\sem{\varphi_1}_{\sigma_1}^\rho$ and, for all $\rho$, we have $\sem{\varphi''}_\sigma^\rho=\sem{\varphi_2}_{\sigma_2}^\rho$. The desired result now follows immediately from the semantics:
\begin{align*}
\sem{\varphi}_\sigma^\rho=\sem{\varphi'}_\sigma^\rho\cap\sem{\varphi''}_\sigma^\rho\stackrel{IH}=\sem{\varphi_1}_{\sigma_1}^\rho\cap\sem{\varphi_2}_{\sigma_2}^\rho.
\end{align*}

The case for $\varphi=\forall \pi.\varphi'$ and $k=2$ implies that $\Loc=\{\ell_1,\ell_2\}$ and it is handled similarly.
Then we treat the case where $k>2$.
From the synthesis function, we can again conclude that $\varphi$ is a conjunction or a universal formula. The universal case is a complicated version of the conjunctive case, so we only threat the former. Let $\varphi=\forall \pi.\varphi'$ and $K=\{1,\dots,k\}$. Then, 
    $\Mc{\varphi} = \otimes_{\ell \in \Loc} \Mcv{\varphi'}{\sigma[\psi \mapsto \ell]}$, and from our assumption of maximality of $\otimes_{i\in K}\Mcv{\varphi_i}{\sigma_i}$, there is a partition $(K_\ell)_{\ell \in \Loc}$, such that for every $\ell \in \Loc$, 
    $\Mcv{\varphi'}{\sigma[\pi \mapsto \ell]} = \otimes_{i\in K_\ell}\Mcv{\varphi_i}{\sigma_i}$.
    By the induction hypothesis, 
    for all $\rho$ we have that $\sem{\varphi'}_{\sigma[\pi\mapsto\ell]}^\rho =\bigcap_{i\in K_\ell}\sem{\varphi_i}_{\sigma_i}^{\rho}$ and we can conclude that $\sem{\varphi}_\sigma^\rho=\bigcap_{\ell\in\Loc}\sem{\varphi'}_{\sigma[\pi\mapsto\ell]}^\rho = \bigcap_{\ell\in\Loc}\bigcap_{i\in K_\ell}\sem{\varphi_i}_{\sigma_i}^{\rho}=\bigcap_{i\in K}\sem{\varphi_i}_{\sigma_i}^{\rho}$.
\qedhere
\end{proof}

The following immediate fact is used to prove in \Cref{lemma:valto}, which shows soundness of synthesized monitors when they have not taken any transition yet, and \Cref{lemma:steps}:
\begin{corollary}\label{lem:unfolding}
$\sem{\max x.\psi}_\sigma=\sem{\psi\{^{\max x.\psi}/_x\}}_\sigma$.
\end{corollary}
\ifarxiversion
\begin{proof}
Immediate, since we interpret recursive formulas as fixed points.
\end{proof}
\fi 

The next lemma shows that, for any combination of synthesized parallel monitors, we can construct an equal synthesized monitor, and link their semantics. It is used in various lemmas below. 
\begin{lemma}\label{lemma:combine-sigma}
If $n=\Mcv{\varphi_1}{\sigma_1}\pand \dots \pand \Mcv{\varphi_k}{\sigma_k}$ with $k\geq 1$, then there exist $\sigma$ and $\psi$ such that $n=\Mc{\psi}$ and $\sem{\psi}_\sigma=\bigcap_{i\in I}\sem{\varphi_i}_{\sigma_i}$. Similarly, if $n=\Mcv{\varphi_1}{\sigma_1}\por \dots \por \Mcv{\varphi_k}{\sigma_k}$, there exist $\sigma$ and $\psi$ such that $n=\Mc{\psi}$ and $\sem{\psi}_\sigma=\bigcup_{i\in I}\sem{\varphi_i}_{\sigma_i}$.
\end{lemma}
\begin{proof}
    We only prove the case for $\pand$, the case for $\por$ is similar.
    We proceed by induction on $k$. The base case is trivial. For the inductive case we consider $n=\Mcv{\varphi_1}{\sigma_1}\pand \Mcv{\varphi_2}{\sigma_2}\pand \dots \pand \Mcv{\varphi_k}{\sigma_k}$, where we know from the induction hypothesis that there exist $\sigma',\psi'$ such that
    $ \Mcv{\varphi_2}{\sigma_2}\pand \dots \pand \Mcv{\varphi_k}{\sigma_k}=\Mcv{\psi'}{\sigma'}$ and
    $\sem{\psi'}_{\sigma'}=\bigcap_{i\in\{2,\dots,k\}}\sem{\varphi_i}_{\sigma_i}$. Let $\pi_1,\dots,\pi_j$ be all the variables occuring in $\varphi_1$. Let $\pi_1',\dots \pi_j'$ all be fresh variables not occuring in $\varphi_1$ or in $\psi'$. Lastly, let $\varphi_1'=\varphi_1\{^{\pi_1'} /_{\pi_1} \}\dots\{^{\pi_j'} /_{\pi_j} \}$.

    Then $\Mcv{\varphi_1}{\sigma_1}=\Mcv{\varphi_1'}{\sigma'[\pi_1'\mapsto \sigma_1(\pi_1)]\dots[\pi_j'\mapsto \sigma_1(\pi_j)]}$
    and $\sem{\varphi_1}_{\sigma_1}=\sem{\varphi_1'}_{\sigma'[\pi_1'\mapsto \sigma_1(\pi_1)]\dots[\pi_j'\mapsto \sigma_1(\pi_j)]}$.
    Furthermore $\Mcv{\psi'}{\sigma'}=\Mcv{\psi'}{\sigma'[\pi_1'\mapsto \sigma_1(\pi_1)]\dots[\pi_j'\mapsto \sigma_1(\pi_j)]}$
    and $\sem{\psi'}_{\sigma'}=\sem{\psi'}_{\sigma'[\pi_1'\mapsto \sigma_1(\pi_1)]\dots[\pi_j'\mapsto \sigma_1(\pi_j)]}$.
    Thus we can conclude that $n=\Mcv{\varphi_1'\wedge\psi'}{\sigma'[\pi_1'\mapsto \sigma_1(\pi_1)]\dots[\pi_j'\mapsto \sigma_1(\pi_j)]}$. We also have that
    \begin{align*}
        &\sem{\varphi_1'\wedge\psi'}_{\sigma'[\pi_1'\mapsto \sigma_1(\pi_1)]\dots[\pi_j'\mapsto \sigma_1(\pi_j)]}\\
        &=\sem{\varphi_1'}_{\sigma'[\pi_1'\mapsto \sigma_1(\pi_1)]\dots[\pi_j'\mapsto \sigma_1(\pi_j)]}\cap\sem{\psi'}_{\sigma'[\pi_1'\mapsto \sigma_1(\pi_1)]\dots[\pi_j'\mapsto \sigma_1(\pi_j)]}\\
        &=\sem{\varphi_1}_{\sigma_1}\cap\sem{\psi'}_{\sigma'}\\
        &=\sem{\varphi_1}_{\sigma_1}\cap\bigcap_{i\in\{2,\dots,k\}}\sem{\varphi_i}_{\sigma_i}=\bigcap_{i\in I}\sem{\varphi_i}_{\sigma_i}
\qedhere
    \end{align*}
\end{proof}

\ifarxiversion
Then we obtain that, for any synthesized monitor that is a combination of synthesized parallel monitors, we can relate the semantics of the corresponding formulas:
\fi 
\begin{corollary}\label{cor:semantics}
For all $\varphi$ and $\sigma$, if $\Mc{\varphi}=m\pandor n$, and there exist $\psi_1, \psi_2$ and $\sigma_1,\sigma_2$ such that $m=\Mcv{\psi_1}{\sigma_1}$ and $n=\Mcv{\psi_2}{\sigma_2}$, then $\sem{\varphi}_\sigma=\sem{\psi_1}_{\sigma_1}\cap\sem{\psi_2}_{\sigma_2}$ if $\pandor=\pand$, and $\sem{\varphi}_\sigma=\sem{\psi_1}_{\sigma_1}\cup\sem{\psi_2}_{\sigma_2}$ if $\pandor=\por$.
\end{corollary}
\ifarxiversion
\begin{proof}
Immediate from the synthesis definition, by \Cref{lemma:nochange} and \Cref{lemma:combine-sigma}.
\end{proof}
\fi

Now we can prove soundness of synthesized monitors when they have not taken any transition yet:
\end{textAtEnd}

\begin{theoremEnd}[default,category=soundness]{lemma}
\label{lemma:valto}
    If $\Mc{\varphi}\valto v$, then $\sem{\varphi}_\sigma=\Htrc_\Loc$, if $v=\yes$, and $\sem{\varphi}_{\sigma}=\emptyset$, if $v=\no$.
\end{theoremEnd}
\begin{proofEnd}
    \ifarxiversion
We proceed by induction on $\valto$.
\fi 

The base is for $\Mc{\varphi}=v$: here the result immediately follows from \Cref{lemma:valtoverdict} and \Cref{lemma:verdict}.

For the first inductive case, we consider $\Mc{\varphi}=m\por n\valto \yes$ and $m\valto \yes$. From the synthesis function and possibly \Cref{lemma:combine-sigma} we conclude that there exist $\varphi_1,\varphi_2$ and $\sigma_1, \sigma_2$ such that $m=\Mcv{\varphi_1}{\sigma_1}$ and $n=\Mcv{\varphi_2}{\sigma_2}$. We apply the induction hypothesis to obtain that $\sem{\varphi_1}_{\sigma_1}=\Htrc_\Loc$. The result then follows from \Cref{cor:semantics}.

\ifarxiversion
Let $\Mc{\varphi}=m\pand n\valto \no$ and $m\valto \no$. From the synthesis function and possibly \Cref{lemma:combine-sigma} we conclude that there exist $\varphi_1,\varphi_2$ and $\sigma_1, \sigma_2$ such that $m=\Mcv{\varphi_1}{\sigma_1}$ and $n=\Mcv{\varphi_2}{\sigma_2}$. We apply the induction hypothesis to obtain that $\sem{\varphi_1}_{\sigma_1}=\emptyset$. The result then follows from \Cref{cor:semantics}.
\fi 



The case for $\vend$ can be skipped, as no formula is synthesized into monitors with $\vend$.

\ifarxiversion
Let $\Mc{\varphi}=m\pand n\valto v$ because $m\valto \yes$ and $n\valto v$. From the synthesis function and possibly \Cref{lemma:combine-sigma} we conclude that there exist $\varphi_1,\varphi_2$ and $\sigma_1, \sigma_2$ such that $m=\Mcv{\varphi_1}{\sigma_1}$ and $n=\Mcv{\varphi_2}{\sigma_2}$. We can apply the induction hypothesis to obtain that $\sem{\varphi_1}_{\sigma_1}=\Htrc_\Loc$ and as we know that $v\neq\vend$, we also get that $\sem{\varphi_1}_{\sigma_1}=\Htrc_\Loc$ if $v=\yes$ and $\sem{\varphi_1}_{\sigma_1}=\emptyset$ if $v=\no$. The result then follows from \Cref{cor:semantics}.

The case for $\Mc{\varphi}=m\por n\valto v$ because $m\valto \no$ and $n\valto v$ is similar to the previous case. 

Let $\Mc{\varphi}=m+n\valto v$ because $m\valto v$. This case can be skipped, because, by the definition of the synthesis function, we know that neither $m\valto v$ nor $n\valto v$ can hold whenever $\Mc{\varphi}=m+n$.


For the last inductive case, we consider $\Mc{\varphi}=\rec x.m\valto v$ because $m\{^{\rec x.m}/_x\}\valto v$. By definition of the synthesis, we know that
$\varphi=\maxx {x}.\psi$ and $m=\Mc{\psi}$. Note that $\Mc{\psi}\{^{\rec x.\Mc{\psi}}/_x\}=\Mc{\psi\{^{\maxx {x}.\psi}/_x\}}$. From the induction hypothesis and \Cref{lem:unfolding}, we obtain that
$\sem{\varphi}_\sigma=\sem{\psi\{^{\maxx {x}.\psi}/_x\}}_\sigma=\emptyset$, if $v=\no$, and $\sem{\varphi}_\sigma=\sem{\psi\{^{\maxx {x}.\psi}/_x\}}_\sigma=\Htrc_\Loc$, if $v=\yes$.
\fi 
\end{proofEnd}

\begin{textAtEnd}[allend, category=soundness]

    The following lemma is a basic fact regarding substitution, and we use it in various places throughout the paper. 
\begin{lemma}\label{lemma:recursion-substitution}
    If $\Mc{\varphi} = m$, $\Mc{\psi} = n$, and the free location variables of $\psi$ are not bound in $\varphi$, then $\Mc{\varphi\{{}^\psi/{}_x\}} = m\{{}^n/{}_x\}$.
\end{lemma}
\begin{proof}
    The proof is by induction on $\varphi$ and the only interesting case is the one where $\varphi = \exists \pi.\varphi'$ or $\varphi = \forall \pi.\varphi'$. We describe the case for $\varphi = \exists \pi.\varphi'$, as the other one is dual.

    In this case, $m = \Mc{\varphi} = \Mc{\exists \pi.\varphi'} = \bigoplus_{\ell \in \Loc} \Mcv{\varphi'}{\sigma[\pi \mapsto \ell]}$.
    By the induction hypothesis,
    $\Mcv{\varphi'\{^\psi /_x \}}{\sigma[\pi \mapsto \ell]} = \Mcv{\varphi'}{\sigma[\pi \mapsto \ell]}\{^{\Mcv{\psi}{\sigma[\pi \mapsto \ell]}} /_x \}.$
    Since $\pi$ does not appear free in $\psi$,
    $\Mcv{\psi}{\sigma[\pi \mapsto \ell]} = \Mc{\psi} = n.$
    Therefore, $\Mcv{\varphi\{^\psi /_x \}}{\sigma} = \bigoplus_{\ell \in \Loc} \Mcv{\varphi'\{^\psi /_x \}}{\sigma[\pi \mapsto \ell]} = \bigoplus_{\ell \in \Loc} \Mcv{\varphi'}{\sigma[\pi \mapsto \ell]}\{^n/_x \} = m\{^n/_x\}$.
\end{proof}

The next lemma shows that synthesized monitors always transition to other synthesized monitors:
\end{textAtEnd}

\begin{theoremEnd}[default,category=soundness]{lemma}
\label{lemma:chain}
    If $\Mc{\varphi}\xrightarrow{A} m'$, then $m'=\Mcv{\varphi'}{\sigma'}$ for some $\sigma'$ and  $\varphi'$.
\end{theoremEnd}
\begin{proofEnd}
We proceed by induction on $\xrightarrow{A}$. The interesting cases are the inductive cases. We first treat both inductive cases of the sum. If $\Mc{\varphi}=m+n\xrightarrow{A} m'$, we have two cases: either $\varphi=[a_\pi]\psi$ or $\varphi=\langle a_\pi\rangle\psi$. We treat the case for $\varphi=[a_\pi]\psi$, the other case is similar. Thus we are in the case where $\Mc{\varphi}=a_{\sigma(\pi)}.\Mc{\psi} + \sum_{b \neq a} b_{\sigma(\pi)}.\yes$. This means that $m'$ can be $\Mc{\psi}$ or $\yes$. In both cases we have obtained the required result.

If $\Mc{\varphi}=\rec {x}.m\xrightarrow{A} m'$ because $m\{^{\rec x.m}/_x\} \xrightarrow A m'$, we know that $m=\Mc{\psi}$ and $\varphi=\maxx x.\psi$. We use
\Cref{lemma:recursion-substitution} to conclude that $m\{^{\rec x.m}/_x\} = \Mcv{\psi\{^{\varphi}/_x\}}{\sigma}$ (note that the free location variables of $\varphi$ are not bound in $\psi$). Thus we can use the induction hypothesis to know that $m'=\Mcv{\varphi'}{\sigma'}$ for some $\sigma'$ and formula $\varphi'$.

If $\Mc{\varphi}=m\pandor n\xrightarrow{A} m'\pandor n'$, from the synthesis function and possibly \Cref{lemma:combine-sigma} we obtain $\sigma_1,\sigma_2$ and $\psi_1$, $\psi_2$ such that $m=\Mcv{\psi_1}{\sigma_1}$ and $n=\Mcv{\psi_2}{\sigma_2}$. We then use the induction hypothesis to obtain $\sigma_1',\sigma_2'$ and $\psi_1'$, $\psi_2'$ such that $m'=\Mcv{\psi_1'}{\sigma_1'}$ and $n'=\Mcv{\psi_2'}{\sigma_2'}$. The required result now follows from \Cref{lemma:combine-sigma}.
\end{proofEnd}

\begin{textAtEnd}[allend, category=soundness]

    The last result we need to prove soundness is the following, which essentially shows that the relation between the hypertrace and the semantics of the formula are carried over through the instrumentation steps:
\end{textAtEnd}

\begin{theoremEnd}[default,category=soundness]{lemma}
\label{lemma:steps}
    Let $\Mc{\varphi}\triangleright T \rightarrowtail \Mcv{\varphi'}{\sigma'}\triangleright T'$. If
    $T'\notin\sem{\varphi'}_{\sigma'}$ then $T\notin\sem{\varphi}_\sigma$; if
    $T'\in\sem{\varphi'}_{\sigma'}$ then $T\in\sem{\varphi}_\sigma$.
\end{theoremEnd}
\begin{proofEnd}
From $\Mc{\varphi}\triangleright T \rightarrowtail \Mcv{\varphi'}{\sigma'}\triangleright T'$ we deduce that $\Mc{\varphi} \xrightarrow{A} \Mcv{\varphi'}{\sigma'}$ and $T\xrightarrow{A}T'$. We proceed by induction on $\Mc{\varphi} \xrightarrow{A} \Mcv{\varphi'}{\sigma'}$.

In the first base case we have $\Mc{\varphi}=v=\Mcv{\varphi'}{\sigma'}$. The result follows immediately from \Cref{lemma:verdict}.

The cases where $\Mc{\varphi}=a_\ell.m$ are not applicable, as the synthesis function cannot produce $a_\ell.m$ (being $|\Act|>1$).

The first inductive case is for $\Mc{\varphi}=m+n$ and $\Mcv{\varphi'}{\sigma'}=m'$ or $\Mcv{\varphi'}{\sigma'}=n'$. We distinguish two cases.
\begin{enumerate}
    \item $\varphi=[a_\pi]\psi$. Here we have two more cases:
    \begin{enumerate}
    \item $m'=\Mc{\psi}$, $m=a_{\sigma(\pi)}.m'$, and $A(\sigma(\pi))=a$.
    Assume that $T'\notin\sem{\psi}_\sigma$. Then $T\notin\sem{[a_\pi]\psi}_\sigma$ follows immediately from the definitions. If we assume instead that $T'\in\sem{\psi}_\sigma$, we also know immediately that $T\in\sem{[a_\pi]\psi}_\sigma$.
    \item $m'=\yes$, $m=b_{\sigma(\pi)}.m'$ for $b\neq a$,
    and $A(\sigma(\pi))\neq a$. It follows immediately that $T\in\sem{[a_\pi]\psi}_\sigma$ as $A(\sigma(\pi))\neq a$, and that $T'\notin \sem{\varphi'}_{\sigma'}$ never holds.
    \end{enumerate}
    \item $\varphi=\langle a_\pi\rangle\psi$. Also here we have two more cases:
    \begin{enumerate}
    \item $m'=\Mc{\psi}$, $m=a_{\sigma(\pi)}.m'$, and $A(\sigma(\pi))=a$. Assume that $T'\notin\sem{\psi}_\sigma$. Then $T\notin\sem{\langle a_\pi\rangle\psi}_\sigma$ follows immediately from the definitions. If we assume instead that $T'\in\sem{\psi}_\sigma$, we also know immediately that $T\in\sem{\langle a_\pi\rangle\psi}_\sigma$.
    \item $m'=\no$, $m=b_{\sigma(\pi)}.m'$ for $b\neq a$, and $A(\sigma(\pi))\neq a$. It follows immediately that $T\notin\sem{\langle a_\pi\rangle\psi}_\sigma$ as $A(\sigma(\pi))\neq a$, and that $T'\in \sem{\varphi'}_{\sigma'}$ never holds.
    \end{enumerate}
\end{enumerate}

Consider $\Mc{\varphi}=\rec x. m\xrightarrow{A}m'$. Then we have $\varphi=\maxx {x}.\psi$ and $m=\Mc{\psi}$. We know $\Mcv{\varphi'}{\sigma'}=m'$, and $\Mc{\psi}\{^{\rec x.\Mc{\psi}}/_x\}\xrightarrow{A} m'$. Assume that $T'\in\sem{\varphi'}_{\sigma'}$. Note that $\Mc{\psi}\{^{\rec x.\Mc{\psi}}/_x\}=\Mc{\psi\{^{\maxx {x}.\psi}/_x\}}$. From the induction hypothesis and \Cref{lem:unfolding}, we obtain that $T\in \sem{\psi\{^{\maxx {x}.\psi}/_x\}}_\sigma=\sem{\varphi}_\sigma$. The case for $T'\notin\sem{\varphi'}_{\sigma'}$ is similar.

Last we consider $\Mc{\varphi}=m\pandor n\xrightarrow{A}m'\pandor n'$, because $m\xrightarrow{A}m'$, $n\xrightarrow{A}n'$, and  $m'\pandor n'=\Mcv{\varphi'}{\sigma'}$. From \Cref{lemma:combine-sigma}, we know that there exist $\psi_1$,$\psi_2$ and
$\tau_1,\tau_2$ such that $m'=\Mcv{\psi_1}{\tau_1}$ and $n'=\Mcv{\psi_2}{\tau_2}$. Similarly, we obtain also that $m=\Mcv{\varphi_1}{\sigma_1}$ and $n=\Mcv{\varphi_2}{\sigma_2}$ for some  $\varphi_1,\varphi_2$ and $\sigma_1, \sigma_2$.
\ifarxiversion
We distinguish two cases.
\begin{enumerate}
    \item $\pandor=\pand$. We assume $T'\in \sem{\varphi'}_{\sigma'}$. Then by \Cref{cor:semantics} we know that $T'\in\sem{\psi_1}_{\tau_1}$ and $T'\in\sem{\psi_2}_{\tau_2}$.
    From the induction hypothesis we get that $T\in\sem{\varphi_1}_{\sigma_1}$ and $T\in\sem{\varphi_2}_{\sigma_2}$. By \Cref{cor:semantics} we know that $T\in\sem{\varphi}_\sigma$.

    If on the other hand we assume $T'\notin \sem{\varphi'}_{\sigma'}$, the argument is dual.

    \item $\pandor=\por$. We assume $T'\in \sem{\varphi'}_{\sigma'}$. Then by \Cref{cor:semantics}
    we know that $T'\in\sem{\psi_1}_{\tau_1}$ or $T'\in\sem{\psi_2}_{\tau_2}$. From the induction hypothesis we get that $T\in\sem{\varphi_1}_{\sigma_1}$ or $T\in\sem{\varphi_2}_{\sigma_2}$.
    By \Cref{cor:semantics} we know that $T\in\sem{\varphi}_\sigma$. The case for $T'\notin \sem{\varphi'}_{\sigma'}$ is dual.
\qedhere
\end{enumerate}
\fi
\end{proofEnd}

\begin{theorem}[Soundness]
\label{thm:soundCentr}
Let $\varphi \in $ \WTFPHypermuHML be a closed formula and $T \in \Htrc_\Loc$.
If $\Mcv{\varphi}{\emptyset} \triangleright T \rightarrowtail^* \no$, then $T \not\in \sem\varphi$;
if $\Mcv{\varphi}{\emptyset} \triangleright T \rightarrowtail^* \yes$, then $T \in \sem\varphi$.
\end{theorem}
\begin{proof}
By definition, there exist an integer $h>0$ and $m_1 \triangleright T_1, \ldots, m_h \triangleright T_h$ such that $m_1 = \Mcv{\varphi}{\emptyset}$, $T_1 = T$,
$m_i \triangleright T_i \rightarrowtail m_{i+1} \triangleright T_{i+1}$ (for every $i = 1,\ldots,h-1$), and
$m_h \triangleright T_h \rightarrowtail v$.
We proceed by induction on $h$.
The base case holds because of \Cref{lemma:valto}.
The inductive case holds because of \Cref{lemma:steps}, which can be applied because of \Cref{lemma:chain}.
\end{proof}


We now move to violation completeness, i.e. that all hypertraces that do not belong to $\sem\varphi$ are rejected by the monitor for $\varphi$. This is proved in \Cref{thm:complCentr}, that requires two preliminary easy lemmata, whose proofs hold by definition of the operational semantics.

\begin{lemma}\label{lem:pandNO}
If $m \triangleright T \rightarrowtail^* \no$, then
$(m_1 \pand \ldots \pand m_h \pand m \pand m_{h+1} \pand\ldots\pand m_k)\triangleright T \rightarrowtail^* \no$, for every $m_1,\ldots,m_h,m_{h+1},\ldots,m_k$ 
without free variables.
\end{lemma}

\begin{lemma}\label{lem:porNO}
If $m_i \triangleright T \rightarrowtail^* \no$ for every $i = 1,\ldots,k$, then
$(m_1 \por \ldots \por m_k)\triangleright T \rightarrowtail^* \no$.
\end{lemma}

\begin{theorem}[Violation Completeness]\label{thm:complCentr}
Let $\varphi \in $ \WTFPHypermuHML be a closed formula and $T \in \Htrc_\Loc$.
If $T \not\in \sem\varphi$, then $\Mcv{\varphi}{\emptyset} \triangleright T \rightarrowtail^* \no$.
\end{theorem}
\begin{proof}
    We assume that in $\varphi$ every recursive variable $x$ appears in the scope of a unique (max) fixed-point subformula of $\varphi$ of the form $\fx(x) \triangleq \max x.\psi$. Let $\leq_\varphi$ be a partial order of the recursive variables in $\varphi$, such that $x \leq_\varphi y$ if $\fx(x)$ is a subformula of $\fx(y)$ in $\varphi$.
    We now define the closure $\cl(\psi)$ of a subformula $\psi$ of $\varphi$ by induction on the number of the free recursion variables in $\psi$:
    if $\psi$ is closed, then $\cl(\psi) \triangleq \psi$; otherwise, $\cl(\psi) \triangleq \cl(\psi\{^{\fx(x)}/_x\})$, where $x$ is $\leq_\varphi$-minimal in $\psi$.

    For each $\sigma$, let $\rho_\sigma$ be such that $\rho_\sigma(x) = \{ T \mid \Mcv{\cl(x)}{\sigma} \triangleright T \not\rightarrowtail^* \no \}$ for every $x$ in its domain.
By induction on the structure of $\psi$, we prove that:
\begin{quote}
for every (not necessarily closed) subformula $\psi$ of $\varphi$, $T \not\in \sem\psi_\sigma^{\rho_\sigma}$ implies  $\Mcv{\cl(\psi)}{\sigma} \triangleright T \rightarrowtail^* \no$, 
for $\sigma$ and $\rho_\sigma$ such that
$\FVloc(\psi) \subseteq \dom(\sigma)$ and
$\FVrec(\psi) \subseteq \dom(\rho_\sigma)$.
\end{quote}
Clearly, the theorem is proved if we consider $\varphi$.
There are seven possible base cases:
\begin{itemize}
\item $\psi = \ttt$, or $\psi = (\pi = \pi')$ for $\sigma(\pi) = \sigma(\pi')$,
or $\psi = (\pi \neq \pi')$ for $\sigma(\pi) \neq \sigma(\pi')$: in these cases, all $T$'s belong to $\sem\psi^{\rho_\sigma}_\sigma$, so there is nothing to prove.

\item $\psi = \ff$, or $\psi = (\pi = \pi')$ for $\sigma(\pi) \neq \sigma(\pi')$,
or $\psi = (\pi \neq \pi')$ for $\sigma(\pi) = \sigma(\pi')$: in these cases, 
$\Mc{\psi} \triangleq \no$ and the claim is then trivial.

\item $\psi = x$: immediate from the definition of $\rho_\sigma$.
\end{itemize}
For the inductive step, we distinguish the outmost operator in $\psi$:
\begin{itemize}
\item $\psi = \psi_1 \wedge \psi_2$: by definition of the semantics, $T \not\in \sem{\psi}^{\rho_\sigma}_\sigma$ if and only if $T \not\in \sem{\psi_i}^{\rho_\sigma}_\sigma$, for some $i \in \{1,2\}$.
By inductive hypothesis, $\Mcv{\cl(\psi_i)}{\sigma} \triangleright T \rightarrowtail^* \no$.
By definition of $\Mcv{\cl(\psi)}{\sigma}$ and Lemma~\ref{lem:pandNO}, $\Mcv{\cl(\psi)}{\sigma} \triangleright T \rightarrowtail^* \no$.

\item $\psi = \psi_1 \vee \psi_2$: by definition of the semantics, $T \not\in \sem{\psi}^{\rho_\sigma}_\sigma$ if and only if $T \not\in \sem{\psi_i}^{\rho_\sigma}_\sigma$, for $i \in \{1,2\}$.
By inductive hypothesis, $\Mcv{\cl(\psi_i)}{\sigma} \triangleright T \rightarrowtail^* \no$.
By definition of $\Mcv{\cl(\psi)}{\sigma}$ and Lemma~\ref{lem:porNO},
$\Mcv{\cl(\psi)}{\sigma} \triangleright T \rightarrowtail^* \no$.

\item $\psi = [a_\pi]\chi$: by definition of the semantics, $T \not\in \sem{\psi}^{\rho_\sigma}_\sigma$ if and only if $\hd(T)(\sigma(\pi)) = a$ and $\tl(T)\not\in\sem{\chi}^{\rho_\sigma}_\sigma$.
By definition, $\Mcv{\cl(\psi)}{\sigma} = a_{\sigma(\pi)}.\Mcv{\cl(\chi)}{\sigma} + \sum_{b \neq a} b_{\sigma(\pi)}.\yes$; hence, $\Mcv{\cl(\psi)}{\sigma} \triangleright T \rightarrowtail \Mcv{\cl(\chi)}{\sigma} \triangleright \tl(T)$. By inductive hypothesis, $\Mcv{\cl(\chi)}{\sigma} \triangleright \tl(T) \rightarrowtail^* \no$, and we easily conclude.

\item $\psi = \langle a_\pi \rangle\chi$: by definition of the semantics, $T \not\in \sem{\psi}^{\rho_\sigma}_\sigma$ if and only if either $\hd(T)(\sigma(\pi)) \neq a$ or $\tl(T)\not\in\sem{\chi}^{\rho_\sigma}_\sigma$.
By definition, $\Mcv{\cl(\psi)}{\sigma} = a_{\sigma(\pi)}.\Mcv{\cl(\chi)}{\sigma} + \sum_{b \neq a} b_{\sigma(\pi)}.\no$. If $\hd(T)(\sigma(\pi)) \neq a$, then 
$\Mcv{\cl(\psi)}{\sigma} \triangleright T \rightarrowtail \no \triangleright \tl(T)$ and we easily conclude.
Otherwise, we work like in the previous case.

\item $\psi = \forall\pi.\chi$: by definition of the semantics, $T \not\in \sem{\psi}^\rho_\sigma$ if and only if there exists $\ell \in \Loc$ such that $T \not\in\sem{\chi}^\rho_{\sigma[\pi \mapsto \ell]}$.
By inductive hypothesis, $\Mcv{\cl(\chi)}{\sigma[\pi \mapsto \ell]} \triangleright T \rightarrowtail^* \no$.
We conclude by definition of $\Mcv{\cl(\psi)}{\sigma}$ and Lemma~\ref{lem:pandNO}.

\item $\psi = \exists\pi.\chi$: by definition of the semantics, $T \not\in \sem{\psi}^\rho_\sigma$ if and only if for all $\ell \in \Loc$ it holds that $T \not\in\sem{\chi}^\rho_{\sigma[\pi \mapsto \ell]}$.
By inductive hypothesis, $\Mcv{\cl(\chi)}{\sigma[\pi \mapsto \ell]} \triangleright T \rightarrowtail^* \no$, for all $\ell$. We conclude by definition of $\Mcv{\cl(\psi)}{\sigma}$ and Lemma~\ref{lem:porNO}.

\item $\psi = \max x.\chi$: we prove the contrapositive, i.e. that $\Mcv{\cl(\psi)}{\sigma}\triangleright T \not\rightarrowtail^* \no$ implies that $T \in \sem{\psi}_\sigma^{\rho_\sigma}$.
Let us consider $S \triangleq \{T'\mid \Mcv{\cl(\chi)}{\sigma} \triangleright T' \not\rightarrowtail^* \no\}$. 
We observe that
$\rho_\sigma[x \mapsto S]$ satisfies the assumptions required by the inductive statement; so,
the inductive hypothesis gives us that $S \subseteq \sem{\chi}_\sigma^{\rho_\sigma[x \mapsto S]}$, and therefore $S \subseteq \sem{\psi}_\sigma^{\rho_\sigma}$.
By the definition of the monitor semantics, $\Mcv{\cl(\psi)}{\sigma}\triangleright T \not\rightarrowtail^* \no$ yields $\Mcv{\cl(\chi)}{\sigma}\triangleright T \not\rightarrowtail^* \no$, meaning that $T \in S$, and therefore $T \in \sem{\psi}_\sigma^{\rho_\sigma}$, which completes the proof.
\qedhere 
\end{itemize}
\end{proof}




\subsection{On the Monitorable Hypersafety Properties}
\label{sec:monHyperprop}

The monitor synthesis we presented above does not cover the whole \HypermuHML, but just \WTFPHypermuHML; nevertheless, this fragment is able to capture all {\em hypersafety properties}, as defined in \cite{hyperproperties}.
Using this definition, one can show that hypersafety properties are subset-closed, as shown in \cite[Theorem 1]{hyperproperties}.

In our setting, $S$ ($\subseteq \Htrc_\Loc$) is a hypersafety property whenever
$$ 
\forall T \in \Htrc_\Loc .\ T \not\in S \implies 
\left( 
\exists M \in \Htrc_{\Loc}^f\ .\ 
\left( 
M \leq T\ \wedge\ 
\forall T' \in \Htrc_{\Loc} .\ M \leq T' \implies T' \not\in S
\right)\right).
$$ 
where $\Htrc_{\Loc}^f$ denotes the set of mappings from $\Loc$ to $\Act^*$ and $M \leq T$ means that, for every $\ell \in \Loc$, there exists a $\ell' \in \Loc$ such that $M(\ell)$ is a prefix of $T(\ell')$. 
\revOne{The formula above tells that a set of hypertraces $S$ is a hypersafety property if every hypertrace $T$ that does not belong to $S$ admits a ``finite justification'' $M$ for this; $M$ is a set of finte traces that can be extended to become $T$ (by possibly changing the traces-to-locations assignment --- as the `$\leq$' relation allows) and that, however we extend it, we obtain some hypertrace that is not in $S$.}

Notice that, in general, \cite{hyperproperties} do not associate traces to locations and, most importantly, they allow hypertraces to be possibly infinite sets of traces.
By contrast, our $T$'s are finite sets, since $|\Loc| = k$, which also implies an upper bound on any possible ``bad thing'' $M$. For such $M$'s it has been proven (see \cite{monitoring_hyperLTL}) that, for each $k>1$, these properties coincide with the so-called {\em $k$-safety hyperproperties}, which in turn correspond to HyperLTL sentences with at most $k$ universal quantifiers. 

However, our monitor synthesis also handles formulae with quantifier alternation; as an example, consider Goguen and Meseguer 's classic definition of the non-interference property given in \cite{hyperproperties}
(see $GMNI$ in \cref{ex:GMNI}).
%
Clearly, the hyperproperty {\em GMNI} can be expressed in the fragment of \HypermuHML for which we have synthesized centralized monitors in \cref{sec:centrSyn}
(see $\varphi_{GMNI}$ in \cref{ex:GMNI}): indeed, to detect violations for a fixed set of traces, it suffices to check all pairs for such a violation, which is what the formula expresses and the monitor synthesized from it does.
However, $GMNI$ is known not to be a hypersafety property, since it is not subset closed.

Hence, if we want to characterize all the properties monitorable for violations in our centralized monitoring set-up, where verdicts are irrevocable, what is needed is a variation on the notion of hypersafety properties that takes the fact that hypertraces are $\Loc$-indexed families of traces into account.
We believe that the following proposal better corresponds to our set-up, and hence to the power of the monitors we have discussed so far. 
Intuitively, by following the characterization of safety property given in \cite{AS85} (i.e., the well-known motto ``nothing bad happens''), the two key characteristics of the notion of the ``bad thing'' witnessing a violation of a hypersafety property are that it must be finitely observable and irremediable for a given \textit{fixed} set of locations $\Loc$.

\begin{definition}
A hyperproperty $S$ is {\em violation-monitorable} whenever
$$
\forall T \in \Htrc_\Loc .\ T \not\in S \implies 
\left( 
\exists M \in \Htrc_{\Loc}^f\ .\ 
\left( 
M \sqsubseteq T\ \wedge\
\forall T' \in \Htrc_{\Loc}.\  M \sqsubseteq T' \Rightarrow T' \not\in S
\right)\right).
$$
where $M \sqsubseteq T$ means that, for every $\ell \in \Loc$, it holds that $M(\ell)$ is a prefix of $T(\ell)$.
\end{definition}

Essentially, this definition allows for the ``bad thing'' $M$ to be infinitely extended (as $T'$), but, differently from hypersafety, no new trace can be added (i.e., $M$ and $T'$ are defined on the very same set of locations $\Loc$). 
Notice that $GMNI$ is a violation-monitorable property. Indeed, let $T \not\in GMNI$; 
then, there is some location $\ell$ such that, for every $\ell' \in \Loc$, it holds that either $T(\ell')$ 
contains some high action, or the low parts of $T(\ell)$ and $T(\ell')$ are different. For each $\ell' \in \Loc$, let $i_{\ell'}$ 
be the first index such that $T(\ell')(i_{\ell'})$ is high or $T(\ell)(i_{\ell'}) \neq T(\ell')(i_{\ell'})$, and let $i$ be the 
maximum of all the $i_{\ell'}$. Let $M$ be the $\Loc$-indexed set of finite traces mapping each $\ell' \in \Loc$ to the prefix 
of $T(\ell')$ of length $i$. 
Now, for every $T'$ such that $M \sqsubseteq T'$, we have that $T' \not\in GMNI$ because its prefix $M$ provides a witness breaking the definition of GMNI for $T'(\ell)$. 

We conjecture that violation-monitorable hypersafety corresponds to the properties that are monitorable for violations in the parallel monitor set-up. 
Moreover, such a definition should also allow us to prove maximality results for a specific monitoring set-up; for example, 
 any violation-monitorable property is monitorable for violations by the monitors produced by our synthesis, and vice versa. This is a challenging direction for future research.
 Furthermore, the classic definition of hypersafety property and that of violation-monitorable hyperproperty given above are parametrized on a relation between $\Loc$-index sets of finite traces and $\Loc$-indexed sets of traces: the relations $\leq$ and $\sqsubseteq$. Another variation on such a relation could be defined modulo a permutation of $\Loc$, so that the names of the locations do not really matter:
\begin{center}
$M \preceq T$ iff there is a permutation $\gamma$ over $\Loc$ such that $M(\ell)$ is a prefix of $T(\gamma(\ell))$, for each $\ell \in \Loc$. 
\end{center}
The study of the notion of safety that emerges by using $\preceq$ and how it compares with the above notions of safety is another direction for future work.

\section{Decentralized Monitoring}
\label{sec:decMon}

When verifying a distributed system, having a central authority that performs any type of runtime verification is a strong assumption, as it reduces the appeal of distribution, creates single points of failure during verification and can pose problems in storing all the traces locally, especially in light of the wide availability of multi-core systems.
Thus, we study to what extent hyperproperties can be monitored by decentralized monitors; these avoid high contentions (leading to vastly improved scalability~\cite{AcetoAFI24}) and also offer better privacy guarantees (whenever they are stationed locally at the nodes where the respective traces are generated~\cite{FraGP13,JiaGP16}). 

\subsection{Decentralized Monitors: Syntax and Operational Semantics}
\label{sec:monitors}
We associate monitors to locations, denoted by $\ell$, and monitors associated to $\ell$ monitor only actions required to happen at $\ell$, thus allowing the processing of events to happen locally. This imposes some form of coordination between monitors at different locations. For this reason, we introduce the possibility for monitors to communicate.



We define a communication alphabet $\Com$, ranged over by $c$, over some finite alphabet of communication constants $\Cvar$ (that contains $\Act$), ranged over by $\gamma$, as
\[
    \Com \ni c\ ::=\ (!G,\gamma) \ \mid\ (?G,\gamma),
\]
where $G\subseteq\Loc$ and $\gamma \in\Cvar$.
We have a communication action $(!G,\gamma)$ for sending $\gamma$ to group $G$ (multicast communication), and one $(?G,\gamma)$ for receiving $\gamma$ from any monitor from the set $G$.
Point-to-point communication can be represented by taking singleton sets for $G$.

The syntax of decentralized monitors is given by the following grammar:
	\begin{align*}
	\DMon \ni M &::=  [m]_\ell ~\mid~ M\vee M ~\mid M\wedge M  \\
	\LMon \ni m &::= \yes \,\mid\, \no \,\mid\, \vend \,\mid\, a.m \,\mid\, c.m \,\mid\,
	m+m \,\mid\, m \por m \,\mid\, m\pand m \,\mid\, \rec{x}.m \,\mid\, x
\end{align*}
Notationally, in what follows we shall sometimes use $\op$ to denote any of $\wedge$ and $\vee$.
Monitor $[m]_\ell$ denotes that $m$ monitors the trace located at location $\ell$, so, it is ‘localized' at $\ell$ (this justifies the name $\LMon$).
Monitors assigned to the same trace run in parallel and observe identical events;
contrary to \cite{circuit_monitors}, monitors assigned to different traces are no longer completely isolated from each other, but can now communicate, which is the main new feature of the decentralized set-up.  


The operational rules for $m \in \LMon$ are given in \Cref{tab:decentralized_locla_operational_semantics}. Notice that, when we have parallel monitors, only one of them at a time can send; by contrast,
all those that can receive from some location $\ell$ are forced to do so.
\begin{table}
    \hrule
\begin{mathpar}
a.m \xrightarrow a m
\and
\inferrule{\ell\in G}{
(?G,\gamma).m \xrightarrow {(?\ell,\gamma)} m}
\and
(!G,\gamma).m \xrightarrow {(!G,\gamma)} m
\and
v \xrightarrow a v
\and
\inferrule{m\{^{\rec x.m}/_x\} \xrightarrow\lambda m'}{\rec x.m \xrightarrow\lambda m'}
\and
\inferrule{m \xrightarrow a m' \\ n \xrightarrow a n'}{m \pandor n \xrightarrow a m' \pandor n'}
\and
\inferrule{m \xrightarrow{(?\ell,\gamma)} m' \\ n \xrightarrow{(?\ell,\gamma)} n' }{m \pandor n \xrightarrow{(?\ell,\gamma)} m' \pandor n'}
 \and 
 \inferrule{m \xrightarrow\lambda m'}{m+n \xrightarrow\lambda m'}
 \and
 \inferrule{m \xrightarrow{(!G,\gamma)} m'}{m \pandor n \xrightarrow{(!G,\gamma)} m' \pandor n} 
\and
\inferrule{m \xrightarrow{(?\ell,\gamma)} m' \\ n   \NOT{\xrightarrow{(?\ell,\gamma)}}}{m \pandor n \xrightarrow{(?\ell,\gamma)} m' \pandor n}
\end{mathpar}
\hrule\vspace{0.5em}
\caption{The operational semantics for decentralized local monitors (up to commutativity of +, $\pand$ and $\por$), where we let $\lambda$ denote either $a$, $(!G,\gamma)$ or $(?\ell,\gamma)$ for $\ell\in\Loc$, $G\subseteq\Loc$. }
\label{tab:decentralized_locla_operational_semantics}

\end{table}

\begin{table}[t]
\hrule
\begin{minipage}{0.55\textwidth}
\center

\[\begin{array}{lll}
\inferrule{m\xrightarrow {(!G,\gamma)} m'}
{[m]_\ell\xrightarrow {\ell:(!G,\gamma)} [m']_\ell}
&
    \inferrule{m\xrightarrow {(?\ell',\gamma)} m' \\ \ell\in G}
    {[m]_\ell \extoverset {$G:(?\ell',\gamma)$}\rightsquigarrow [m']_\ell}
\vspace*{.3cm}
\\
\inferrule{m\NOT{\xrightarrow {(?\ell',\gamma)}}}
{[m]_\ell \extoverset {$G:(?\ell',\gamma)$}\rightsquigarrow [m]_\ell}
&
\inferrule{\ell\notin G}
{[m]_\ell \extoverset {$G:(?\ell',\gamma)$}\rightsquigarrow [m]_\ell}
\vspace*{.3cm}
\\
\multicolumn{2}{l}{
\inferrule{M \extoverset{$G:(?\ell,\gamma)$}\rightsquigarrow M' \\ N \extoverset{$G:(?\ell,\gamma)$}\rightsquigarrow N' }{M \op N \extoverset{$G:(?\ell,\gamma)$}\rightsquigarrow M' \op N'}
}
\vspace*{.3cm}
\\
\multicolumn{2}{l}{
\inferrule{M \xrightarrow{\ell:(!G,\gamma)} M' \\ N \extoverset{$G:(?\ell,\gamma)$}\rightsquigarrow N'  }{M \op N \xrightarrow{\ell:(!G,\gamma)} M' \op N'}
}
\vspace{0.5em}
\end{array}\]

\caption{Operational semantics for communication of $M\in\DMon$ (up to commutativity of $\wedge$ and $\vee$ -- denoted by $\op$).}
\label{tab:communicating_operational_semantics_communicating_monitors}
\end{minipage}
\quad
\begin{minipage}{0.4\textwidth}
\center

\[\begin{array}{l}
    \inferrule{A(\ell)= a \\ m \xrightarrow{a} m'
    }{%
   [m]_\ell  \xrightarrow{A} [m']_\ell
    }
\vspace*{.3cm}
\\
\inferrule{A(\ell)= a \\ m \NOT{\xrightarrow{a}} \\ m \NOT{\xrightarrow{c}}}{%
[m]_\ell  \xrightarrow{A} [\vend]_\ell
}
\vspace*{.3cm}
\\
\inferrule{M \xrightarrow{A} M' \\ N \xrightarrow{A} N'}{
M \op N\xrightarrow{A} M' \op N'
}
\vspace{0.5em}
\end{array}\]

\caption{Operational semantics for actions of $M\in\DMon$ (up to commutativity of $\wedge$ and $\vee$  -- denoted by $\op$).}
\label{tab:actions_operational_semantics_communicating_monitors}
\end{minipage}
\hrule
\end{table}

For $M \in \DMon$, the operational semantics can be found in \Cref{tab:communicating_operational_semantics_communicating_monitors} (the rules concerning communication) and \Cref{tab:actions_operational_semantics_communicating_monitors} (the rules concerning action steps). The operational semantics in \Cref{tab:communicating_operational_semantics_communicating_monitors} defines multicast, where a monitor located at $\ell$ sends a message to group $G$ and every monitor at a location in $G$ that can receive from $\ell$ does so; every monitor that cannot, or that is not in $G$, does not change its state.
The first four rules capture the judgment for inferring when all components of a monitor which are able to receive a certain $\gamma$ sent from a location do so. Intuitively, $\ell$ is the location from which message $\gamma$ was sent to group $G$, and $M\extoverset{$G:(?\ell,\gamma)$} \rightsquigarrow N$ indicates that every monitor in $M$ located at a location in $G$ that can receive $\gamma$ from $\ell$ indeed has received $\gamma$ and transitioned appropriately in $N$.
The last two rules then actually define communication. In particular, the last rule in \Cref{tab:communicating_operational_semantics_communicating_monitors} implements multicast by stipulating that the outcome of the synchronization between a send action $\ell:(!G,\gamma)$ and a receive one of the form $G:(?\ell,\gamma)$ is the send action 
itself,
which can be received by other monitors at locations in $G$ in a larger monitor of which $M\op N$ is a sub-term. We note, in passing, that monitors $M\in\DMon$ are `input-enabled': for each $M, G,\ell$ and $\gamma$, there is always some $M'$ such that $M\extoverset{$G:(?\ell,\gamma)$} \rightsquigarrow M'$. So the last rule in \Cref{tab:communicating_operational_semantics_communicating_monitors} (and its symmetric version) can always be applied when the send transition in its premise is available.


Monitors can also locally observe an action, as prescribed by a location-to-action function $A$; the rules are given in \Cref{tab:actions_operational_semantics_communicating_monitors}.
Monitors at the same location
observe the same action.
If a monitor cannot take the action prescribed by $A$ at its location, the monitor becomes $\vend$, as stipulated by the second rule given in \Cref{tab:actions_operational_semantics_communicating_monitors}.
%
Note that it is not sufficient to trigger that rule when $m$ 
cannot 
exhibit action $A(\ell)$: 
we also require that $m$ cannot communicate. 
Note that the inability of $m$ to exhibit action $A(\ell)$ is not sufficient to trigger that rule: we also require that $m$ cannot communicate.
Intuitively, this is because  monitors exhibit an `alternating' behavior in which they observe the next action produced by a system hypertrace and then embark in a sequence of communications with other monitors to inform them of what they observed.  
\revOne{The order in which these communication messages are received is immaterial, since for receive actions a confluence property holds (at least for the monitors synthetized from formulae --- see Lemma~\ref{lemma:commutativereceiving}).}

As will be made clear in our definition 
of a weak bisimulation relation presented in~\Cref{def:weakbis}, 
such communications are interpreted as internal actions in monitor behavior. Therefore, the inability of some monitor $[m]_\ell$ to perform action $A(\ell)$ can only be gauged in `stable states'---that is, monitor states in which no communication is possible. This design choice is akin to that underlying the definition of refusal testing presented in~\cite{Phillips87} and of the stable-failures model for (Timed) CSP defined in~\cite{ReedR99,Roscoe1997}, where the inability of a process to perform some action can only be determined in states that afford no internal computation steps.

\begin{table}
\hrule
\begin{minipage}{0.45\textwidth}
\center

\[\begin{array}{ll}
\inferrule{m \valto v}{[m]_\ell \valto v}
&
\inferrule{M \valto \vend \quad N \valto \vend}{M \op N \valto \vend}
\vspace*{.3cm}
\\
\inferrule{M \valto \no}{M \wedge N \valto \no}
& 
\inferrule{M \valto \yes \quad N \valto v}{M \wedge N \valto v}
\vspace*{.3cm}
\\
\inferrule{M \valto \yes}{M \vee N \valto \yes}
& 
\inferrule{M \valto \no \quad N \valto v}{M \vee N \valto v}
\vspace{0.5em}
\end{array}\]

\caption{The verdict combination rules for decentralized monitors (up to commutativity of $\wedge$ and $\vee$  -- denoted by $\op$).}
\label{tab:verdict_combination_decentralized}
\end{minipage}
\quad
\begin{minipage}{0.45\textwidth}
\center

\[\begin{array}{l}
    \inferrule{M \xrightarrow{A} M' \\ T\xrightarrow{A}T'}
    {%
   M\triangleright T \rightarrowtail M'\triangleright T'}
\vspace*{.3cm}
   \\
   \inferrule{M\xrightarrow{\ell:(!G,\gamma)} M' }
   {%
  M\triangleright T \rightarrowtail M'\triangleright T}
\vspace*{.3cm}
   \\
   \inferrule{M \valto v }
   {%
  M\triangleright T \rightarrowtail v}
\vspace{0.5em}
\end{array}\]

\caption{The evolution of a decentralized monitor instrumented on a hypertrace.}
\label{tab:instrumentation_decentralized}
\end{minipage}
\hrule
\end{table}

Verdict evaluation for $M \in \DMon$ is defined in \Cref{tab:verdict_combination_decentralized} and relies on that for $m \in \CMon$ provided in \Cref{tab:verdict_evaluation_semantics}.
Finally, given a decentalized monitor $M$ and a hypertrace $T$, the instrumentation of the monitor on the trace is described by the rules of \Cref{tab:instrumentation_decentralized}.
 As before, we denote with $\rightarrowtail^* $ the reflexive transitive closure of $ \rightarrowtail$.

\revTwo{
To conclude, let us show how the operational semantics works, by considering the following toy example:
\begin{align*}
m_1 \triangleq a_1.m_1' & \qquad \text{where } \qquad m_1' \triangleq (!\{\ell_2,\ell_3\},a).m_1''
\\
m_2 \triangleq a_2.m_2' & \qquad \text{where } \qquad m_2' \triangleq (?\ell_1,a).m_2''
\\
m_3 \triangleq a_3.m_3' & \qquad \text{where } \qquad m_3' \triangleq (?\ell_1,a).m_3''
\\
m_4 \triangleq a_4.m'_4
\end{align*}
Assume that we have four locations  $\ell_1,\ldots,\ell_4$ and the 4-tuple of actions $A \triangleq [\ell_1 \mapsto a_1, \ell_2 \mapsto a_2, \ell_3 \mapsto a_3, \ell_4 \mapsto a]$. Then,
\[
m_i \xrightarrow{a_i} m_i'
\]
for all $i \in \{1,\ldots,4\}$ and so
\begin{align*}
[m_i]_{\ell_i} \xrightarrow A [m_i']_{\ell_i} & \qquad \text{for every } i \in \{1,2,3\} \text{ (since $a_i = A(\ell_i)$)}
\\
[m_4]_{\ell_4} \xrightarrow A [\vend]_{\ell_4} & \qquad \text{ (since $a_4 \neq A(\ell_4) = a$)}
\end{align*}
Thus,
\[
\bigwedge_{i=1}^4 [m_i]_{\ell_i} \xrightarrow A \bigwedge_{i=1}^3 [m_i']_{\ell_i} \ \ \wedge\ \  [\vend]_{\ell_4}
\]
Now, we can handle communication. First, we can infer that
\[
m_1' \xrightarrow{(!\{\ell_2,\ell_3\},a)} m_1''
\qquad \text{and} \qquad
m_i' \xrightarrow{(?\ell_1,a)} m_i'' \quad \text{for } i \in \{2,3\}
\]
and so 
\[
[m_1']_{\ell_1} \xrightarrow{\ell_1:(!\{\ell_2,\ell_3\},a)} [m_1'']_{\ell_1}
\qquad \text{and} \qquad
[m_i']_{\ell_i} \extoverset {$\{\ell_2,\ell_3\}:(?\ell_1,a)$}\rightsquigarrow  [m_i'']_{\ell_i} \quad \text{for } i \in \{2,3\}
\]
By observing that $\vend \NOT{\xrightarrow {(?\ell_1,a)}}$, we can also infer that
\[
[\vend]_{\ell_4} \extoverset {$\{\ell_2,\ell_3\}:(?\ell_1,a)$}\rightsquigarrow [\vend]_{\ell_4}
\]
and so
\[
[m_2']_{\ell_2} \wedge [m_3']_{\ell_3} \wedge [\vend]_{\ell_4} \extoverset {$\{\ell_2,\ell_3\}:(?\ell_1,a)$}\rightsquigarrow [m_2'']_{\ell_2} \wedge [m_3'']_{\ell_3} \wedge [\vend]_{\ell_4}
\]
Finally, we have that
\[
\bigwedge_{i=1}^3 [m_i']_{\ell_i} \ \ \wedge\ \  [\vend]_{\ell_4} \xrightarrow{\ell_1:(!\{\ell_2,\ell_3\},a)}  \bigwedge_{i=1}^3 [m_i'']_{\ell_i} \ \ \wedge\ \  [\vend]_{\ell_4}
\]
The final verdict of this monitor will be $\vend$ if all $m_i''$ (for $i \in \{1,2,3\}$) evaluate to either $\vend$ or $\yes$; in contrast, if any of them evaluates to $\no$, then the monitor evaluates to $\no$. In particular, because of the failure in observing $a$ in $\ell_4$ (as stipulated by $A$), there is no way to obtain a $\yes$ verdict.
}

\subsection{Synthesizing Decentralized Monitors Correctly}
\label{sec:princ}

In this section we describe how to synthesize decentralized monitors `correctly' from formulas, i.e. such that their behavior corresponds to that of the corresponding centralized monitors. The advantage of this approach is that it simplifies the proof that monitors synthesized via a `correct' decentralized synthesis function are sound and violation-complete, by utilizing the correspondence to centralized monitors. Moreover, it identifies desirable properties of a `correct' decentralized synthesis function that can guide the development of further automated decentralized-monitor synthesis algorithms. 

We first define the correspondence between centralized and decentralized monitors and show that this correspondence is sufficient to obtain soundness and violation-completeness in the decentralized setting from the corresponding results in the centralized setting (\Cref{thm:soundCentr,thm:complCentr}). In the remainder of the section, given a synthesis function which takes as inputs a formula $\varphi$ and a mapping $\sigma$ from location variables to locations, and outputs a monitor $\MdecVarDefault\in\DMon$, we specify criteria that allow us to derive this correspondence.

The correspondence between the centralized and the decentralized monitors is characterized as a weak bisimulation; this first requires defining sequences of communication actions, as the analogous of unobservable transitions in process calculi.

\begin{definition}
We write 
$M \boldarrow M'$ to denote the existence of an integer $h>0$, of 
$h$ monitors $M_1 , \ldots, M_h $ and of $h-1$ locations $\ell_1,\ldots , \ell_{h-1}$ and communication actions $c_1,\ldots , c_{h-1}$ such that $M_1 = M$, $M_h = M'$, and
$M_i \xrightarrow{\ell_i:c_i} M_{i+1} $ (for every $i = 1,\ldots,h-1$). 

Similarly, we write $m\boldarrow m'$ to denote the existence of an integer $h>0$, of $h$ local monitors $m_1,\dots, m_h$ and of $h-1$ actions $c_1,\cdots c_{h-1}\in \{(!G,\gamma),(?\ell,\gamma)\mid G\subseteq\Loc, \ell\in\Loc, \gamma\in\Cvar\}$ such that $m_1=m$, $m_h=m'$ and $m_i\xrightarrow{c_i}m_{i+1}$  (for every $i = 1,\ldots,h-1$).
\end{definition}
Notice that, by definition of $\rightarrow$ on communicating monitors,
each $c_i$ is $(!G_i,\gamma_i)$,
for some $G_i\subseteq\Loc$ and $\gamma_i\in\Cvar$. 

\begin{definition}\label{def:weakbis}
    A binary relation $\R$ over $\DMon\times\CMon$ is a weak bisimulation if and only if, whenever $M\R m$, it holds that:
    \begin{enumerate}
    \item $ \exists M'\in\DMon$ such that $M\boldarrow M'$ and $M'\valto v$ if and only if $m\valto v$.
    \item If $M\xrightarrow{A}M'$ then $\exists m'\in \CMon$ such that $m\xrightarrow{A}m'$ and $M'\R m'$.
    \item If $M\xrightarrow{c}M'$ then $M'\R m$, where $c=\ell:(!G,\gamma)$ for some $\ell\in\Loc$, $G\subseteq\Loc$, $\gamma\in\Cvar$.
    \item If $m\xrightarrow{A}m'$ then there exist $M_1,M_2, M'$ such that $M\boldarrow M_1\xrightarrow{A}M_2\boldarrow M'$ and $M'\R m'$.
    \end{enumerate}
    \end{definition}

    We now prove four facts concerning the instrumentation rules of weakly bisimilar monitors, which we use to prove \Cref{thm:centralvsdisstrong}, which shows that centralized and decentralized monitors behaviorally agree.

    \begin{lemma}\label{lemma:facts}
        Let $\R$ be a weak bisimulation relation such that $M\R m$. Then the following statements hold:
    \begin{enumerate}
        \item\label{itemone} If there exists $M'$ such that $M\triangleright T \rightarrowtail M'\triangleright T'$, then there exists $m'$ such that $m\triangleright T\rightarrowtail^* m' \triangleright T'$ and $M'\R m'$.
        \item\label{itemtwo} If there exists $m'$ such that $m\triangleright T \rightarrowtail m'\triangleright T'$, then there exists $M'$ such that $M\triangleright T\rightarrowtail^* M' \triangleright T'$ and $M'\R m'$.
        \item\label{itemthree} If $M\triangleright T\rightarrowtail v$ then $m\triangleright T\rightarrowtail v$.
        \item\label{itemfour} If $m\triangleright T\rightarrowtail v$ then $M\triangleright T\rightarrowtail^* v$.
    \end{enumerate}
    \end{lemma}
    \begin{proof}
    For the first item, we assume that $M\triangleright T \rightarrowtail M'\triangleright T'$. Hence we have that $M\xrightarrow{A}M'$ and $T\xrightarrow{A}T'$ or $M\xrightarrow{\ell:(!G,\gamma)}M'$ and $T=T'$. In the first case, we use the definition of weak bisimulation to obtain that $m\xrightarrow{A}m'$ for $M'\R m'$, from which we conclude that $m\triangleright T\rightarrowtail m' \triangleright T'$. In the second case, we use the definition of weak bisimulation to obtain that $M'\R m$ and we also know that $m\triangleright T\rightarrowtail^* m \triangleright T$.
    
    For the second item, we assume that $m\triangleright T\rightarrowtail m' \triangleright T'$. Hence we have that $m\xrightarrow{A}m'$ and $T\xrightarrow{A}T'$. From the definition of weak bisimulation, there exist $M_1,M_2, M'$ such that $M\boldarrow M_1\xrightarrow{A}M_2\boldarrow M'$ and $M'\R m'$. Thus, we can conclude that $M\triangleright T \rightarrowtail^* M_1\triangleright T$, $M_1\triangleright T \rightarrowtail M_2\triangleright T'$ and $M_2\triangleright T' \rightarrowtail^* M'\triangleright T'$, from which we obtain that $M\triangleright T \rightarrowtail^* M'\triangleright T'$.
    
    For the third item we assume that $M\triangleright T\rightarrowtail v$, from which we conclude that $M\valto v$. From the definition of weak bisimulation we obtain immediately that $m\valto v$ and thus $m\triangleright T\rightarrowtail v$.
    
    For the fourth item we assume that $m\triangleright T\rightarrowtail v$, from which we conclude that $m\valto v$. From the definition of weak bisimulation we obtain immediately that $M\boldarrow  M'$ and $M'\valto v$ and thus  $M\triangleright T \rightarrowtail^* v$.
    \end{proof}
    

    \begin{proposition}\label{thm:centralvsdisstrong}
        Let $\R$ be a weak bisimulation such that $M\R m$. Then, 
         $M\triangleright T \rightarrowtail^* v$ if and only if
        $m \triangleright T \rightarrowtail^* v$.
        \end{proposition}
    \begin{proof}
        For the left to right direction, assume that $M \triangleright T \rightarrowtail^* v$. By definition there exist an integer $h>0$ and $M_1 \triangleright T_1, \ldots, M_h \triangleright T_h$ such that $M_1 = M$, $T_1 = T$,
    $M_i \triangleright T_i \rightarrowtail M_{i+1} \triangleright T_{i+1}$ (for every $i = 1,\ldots,h-1$), and
    $M_h \triangleright T_h \rightarrowtail v$. We proceed by induction on $h$. The base case holds because of \Cref{lemma:facts} (Item \ref{itemthree}). In the inductive step we use \Cref{lemma:facts} (Item \ref{itemone}) to obtain $m\triangleright T\rightarrowtail^* m_2\triangleright T_2$ for some $m_2$ and $T_2$ such that $M_2\R m_2$. Then we can use the induction hypothesis to conclude that $m_2\triangleright T_2\rightarrowtail^* v$. Thus,  $m\triangleright T\rightarrowtail^* v$.
    
    For the right to left direction, the proof is identical but uses \Cref{lemma:facts} (Item \ref{itemfour}) and \Cref{lemma:facts} (Item \ref{itemtwo}) instead.
    \end{proof}

This allows us to obtain violation-completeness and soundness for decentralized monitors from the corresponding results for centralized monitors:

    
    \begin{corollary}[Soundness]\label{cor:soundDec}
        Let   $T \in \Htrc_\Loc$, $\varphi \in $ \WTFPHypermuHML be a closed formula such that $\MdecVar{\emptyset}{\varphi}$ is defined,
and $\mathcal{R}$ a weak bisimulation such that $(\MdecVar{\emptyset}{\varphi},\Mcv{\varphi}{\emptyset})\in \mathcal{R}$. 
        If $\MdecVar{\emptyset}{\varphi} \triangleright T \rightarrowtail^* \no$, then $T \not\in \sem\varphi$;
        if $\MdecVar{\emptyset}{\varphi} \triangleright T \rightarrowtail^* \yes$, then $T \in \sem\varphi$.
        \end{corollary}
        
        \begin{corollary}[Violation Completeness]\label{cor:compDec}
        Let   $T \in \Htrc_\Loc$, $\varphi \in $ \WTFPHypermuHML  be a closed formula such that $\MdecVar{\emptyset}{\varphi}$ is defined, 
        and $\mathcal{R}$ a weak bisimulation such that $(\MdecVar{\emptyset}{\varphi},\Mcv{\varphi}{\emptyset})\in \mathcal{R}$. 
        If $T \not\in \sem\varphi$, then $\MdecVar{\emptyset}{\varphi} \triangleright T \rightarrowtail^* \no$.
        \end{corollary}    
    
We now describe sufficient conditions for any decentralized synthesis function such that 
 there is a weak bisimulation between the centralized and the decentralized monitors synthesized from a formula $\varphi$ and a location environment $\sigma$. Whenever we write $M\xrightarrow{c}N$ for $M,N\in\DMon$, we assume that $c\in\{\ell:(!G,\gamma)\mid \ell\in\Loc, G\subseteq\Loc,\gamma\in\Cvar\}$, as per the labeling of the communication transitions of decentralized monitors. We write $[m]_\ell\in M$, for $M\in\DMon$, if $[m]_\ell$ is one of its constituents: formally, $[m]_\ell\in [m]_\ell$ and, if $[m]_\ell\in M$, then $[m]_\ell\in M\op N$ and $[m]_\ell\in N\op M$ (recall that $\op$ denotes either $\wedge$ or $\vee$).
We start 
by defining when $M\in\DMon$ can(not) communicate:

\begin{definition}\label{def:nocommunicate}
    Let $M\in\DMon$. We say that $M$ can communicate if there exists $[m]_\ell\in M$ such that $m\xrightarrow{c}n$, for some $c\in\Com$ and $n \in \LMon$. Otherwise, we say that $M$ cannot communicate.
\end{definition}


    

    \begin{definition}\label{def:principled-synthesis}
        We say that a monitor synthesis $\MdecVar{-}{-}$ is {\em principled}
        when it satisfies the following conditions, 
        for every formula $\varphi$ and environment $\sigma$ such that $\MdecVar\sigma\varphi$ is defined:    
    \begin{description}
        \item[Verdict Agreement:] for every verdict $v$, $\Mc{\varphi} \valto v$ if and only if $\MdecVarDefault \valto v$;
        \item[Verdict Irrevocability:] for every verdict $v$ and $\MdecVarDefSigma{\varphi} \xrightarrow{A} M_1 \boldarrow M_2 \boldarrow M$, if $M_2 \valto v$ then $M \valto v$; 
        \item[Reactivity:] for every $A$, if $\varphi$ has no free recursion variables, then there exists $M$ such that $\MdecVarDefault \xrightarrow{A} M$;
        \item[Bounded Communication:] 
        for every $\MdecVarDefault\xrightarrow{A}M\boldarrow M'$, there exists $M''$ such that $M'\boldarrow M''$ and $M''$ cannot communicate;
        \item[Processing-Communication Alternation:] for every $\MdecVarDefault \xrightarrow{A} M \boldarrow M_1$, 
        \begin{enumerate}
            \item $\MdecVarDefault$ cannot communicate, and 
            \item $M_1 \xrightarrow{c} M_2$ implies $M_1 \NOT{\xrightarrow{A'}}$, for every $A'$;
        \end{enumerate}
        \item[Formula Convergence:] 
        if $\MdecVarDefault \xrightarrow{A} M \boldarrow M'$, with $M'$ that cannot communicate, and $\Mc{\varphi} \xrightarrow{A} \Mcv{\varphi'}{\sigma'}$, for some formula $\varphi'$ and environment $\sigma'$, then $M' = \MdecVar{\sigma'}{\varphi'}$.
    \end{description}
    \end{definition}

    From the Formula Convergence property, we can immediately derive an auxiliary property:
    \begin{lemma}[Uniqueness]
    If $\MdecVarDefault \xrightarrow{A} M \boldarrow M_1$, with $M_1$ that cannot communicate, and $\MdecVarDefault \xrightarrow{A} M \boldarrow M_2$, with $M_2$ that cannot communicate, then $M_1=M_2$.
    \end{lemma}
    \begin{proof}
        Indeed, since $\Mc{\varphi}$ has a possible transition for every $A$, by \Cref{lemma:chain} we know that
        $\Mc{\varphi}\xrightarrow{A}\Mcv{\varphi'}{\sigma'}$. By applying Formula Convergence twice, we obtain $M_1=\MdecVar{\sigma'}{\varphi'}=M_2$.
    \end{proof}

    Let $\MdecVar{-}{-}$ be a decentralized synthesis function.
    We define relation $\R_{\MdecVarName}$ as follows:
    \begin{align*}\label{biscandidate}
        \R_{\MdecVarName} & \triangleq 
        \R_1 \cup \R_2 
\\
        \R_1 & \triangleq 
        \left\{\left(\MdecVarDefSigma{\varphi}, \Mc{\varphi}\right)
                \mid \FVloc(\varphi) \subseteq \dom(\sigma)
                \text{ and }
                \FVrec(\varphi) = \emptyset 
         \right\} 
        \\
        \R_2 & \triangleq 
        \left\{
        \left(M',\Mcv{\varphi'}{\sigma'}\right) \mid 
                \FVloc(\varphi) \subseteq \dom(\sigma),~
                \FVrec(\varphi) = \emptyset 
                 \text{ and }
	\MdecVarDefSigma{\varphi}\xrightarrow{A} M \boldarrow M'\boldarrow \MdecVar{\sigma'}{\varphi'}
        \right\}
        \end{align*}
The crucial property of any principled synthesis function is the following:

    \begin{theorem}
    \label{thm:principled-bisimulation}
        For every principled synthesis $\MdecVar{-}{-}$, 
        $\R_{\MdecVarName}$ is a weak bisimulation.
    \end{theorem}
\begin{proof}
    To prove that $\R_{\MdecVarName}$ is a weak bisimulation, we first consider a pair $(\MdecVarDefSigma{\varphi},\Mc{\varphi})\in \R_1$. The first condition of weak bisimulation holds because, from Processing-Communication Alternation and Reactivity, we conclude that $\MdecVarDefSigma{\varphi}\NOT{\xrightarrow{c}}$, and thus $M=M'$. Then the result follows from Verdict Agreement.
        
    For the second condition of weak bisimulation, we assume that $\MdecVarDefSigma{\varphi}\xrightarrow{A} M$. Since $\Mc{\varphi}$ has a possible transition for every $A$, by \Cref{lemma:chain} we know that $\Mc{\varphi}\xrightarrow{A} \Mcv{\varphi'}{\sigma'}$ for some $\sigma'$ and $\varphi'$. From Bounded Communication, we can use Formula Convergence to obtain that $M\boldarrow \MdecVar{\sigma'}{\varphi'}$. By definition, $M\R_2 \Mcv{\varphi'}{\sigma'}$ and this suffices to conclude.
    
    The third condition of weak bisimulation is not relevant, as $\MdecVarDefSigma{\varphi} \NOT{\xrightarrow{c}}$ via Processing-Communication Alternation and Reactivity.
    To check the fourth condition, we assume that $\Mc{\varphi}\xrightarrow{A}\Mcv{\varphi'}{\sigma'}$ (because of 
    \Cref{lemma:chain}, we know that $\Mc{\varphi}\xrightarrow{A}$ results in a monitor synthesized from a formula). Via Reactivity we conclude that $\MdecVarDefSigma{\varphi}\xrightarrow {A} M$ for some $M$. Then we can use Bounded Communication and Formula Convergence to derive that $\MdecVarDefSigma{\varphi}\xrightarrow {A} M\boldarrow \MdecVar{\sigma'}{\varphi'}$. As $\MdecVar{\sigma'}{\varphi'}\R_1 \Mcv{\varphi'}{\sigma'}$, we can conclude.
    
    Now consider a pair $(M',\Mcv{\varphi'}{\sigma'})\in \R_2$, with $\MdecVarDefSigma{\varphi}\xrightarrow{A}M\boldarrow M'\boldarrow \MdecVar{\sigma'}{\varphi'}$ for some $\sigma,\varphi$ and $M$. 
    To verify the first condition from left to right, we assume $M'\boldarrow M''$ and $M''\valto v$. Since $\Mc{\varphi}$ has a possible transition for every $A$ and \Cref{lemma:chain}, we know that $\Mc{\varphi}\xrightarrow{A} \Mcv{\varphi''}{\sigma''}$ for some $\sigma''$ and $\varphi''$. Hence, we can use Bounded Communication and Formula Convergence to conclude that $M'' \boldarrow \MdecVar{\sigma''}{\varphi''}$. Furthermore, we obtain from Uniqueness and Processing-Communication Alternation that $\MdecVar{\sigma'}{\varphi'}=\MdecVar{\sigma''}{\varphi''}$. Then we apply Verdict Irrevocability to conclude that $\MdecVar{\sigma'}{\varphi'}\valto v$. Finally, we can conclude using Verdict Agreement.
    For the right to left direction, we use Verdict Agreement to obtain that $\MdecVar{\sigma'}{\varphi'}\valto v$, which immediately satisfies the first condition of weak bisimulation. 
    For the second condition of weak bisimulation, we assume that $M'\xrightarrow{A'} M''$. From Processing-Communication Alternation, we conclude that $M' \NOT{\xrightarrow{c}}$, and thus that $M'=\MdecVar{\sigma'}{\varphi'}$. This sends back to the second condition for pairs in $\R_1$.
    
    For the third condition, we assume that $M'\xrightarrow{c}M''$. Since $\Mc{\varphi}$ has a possible transition for every $A$ and \Cref{lemma:chain}, we know that $\Mc{\varphi}\xrightarrow{A} \Mcv{\varphi''}{\sigma''}$ for some $\sigma''$ and $\varphi''$. Hence, we can use Bounded Communication and Formula Convergence to conclude that $M'' \boldarrow \MdecVar{\sigma''}{\varphi''}$. Furthermore, we obtain from Uniqueness and Processing-Communication Alternation that $\MdecVar{\sigma'}{\varphi'}=\MdecVar{\sigma''}{\varphi''}$. By definition, $(M'',\Mcv{\varphi'}{\sigma'})\in \R_2$ and this suffices to conclude.

    To check the fourth condition, we assume that 
    $\Mcv{\varphi'}{\sigma'}\xrightarrow{A'}\Mcv{\varphi''}{\sigma''}$. 
    We note that if $N \rightarrow N'$ and $N$ has no free recursion variables, then neither does $N'$ --- by straightforward induction on the derivation of $N \rightarrow N'$.
    Therefore, $\MdecVar{\sigma'}{\varphi'}$ has no free recursion variables, and neither does $\varphi'$.
    Via Reactivity we conclude that $\MdecVar{\sigma'}{\varphi'}\xrightarrow{A'}M_1$. Then we use Bounded Communication and Formula Convergence to derive that $M_1\boldarrow \MdecVar{\sigma''}{\varphi''}$, resulting in
    $M'\boldarrow \MdecVar{\sigma'}{\varphi'}\xrightarrow{A'}M_1 \boldarrow 
    \MdecVar{\sigma''}{\varphi''}$, 
    which concludes the proof.
\end{proof}

\subsection{From Formulas to Decentralized Monitors}
\label{sec:decSyn}

We now describe how to synthesize decentralized monitors for a fragment of \WTFPHypermuHML, and show that this synthesis function satisfies \Cref{def:principled-synthesis}. This allows us to apply \Cref{thm:principled-bisimulation} and obtain soundness and violation-completeness of these synthesized monitors. 

\label{sec:logic}
In what follows, we consider formulas from \PHypermuHML, the subset of \HypermuHML given by the following grammar (see \Cref{sec:concl} for a discussion on the fragment chosen):
 \begin{align*}
	\varphi &::= \forall \pi. \varphi ~\mid~ \exists \pi. \varphi ~\mid~ \varphi \land \varphi ~\mid~ \varphi \lor \varphi ~\mid~ \psi\\
	\psi &::=  \ttt \,\mid\, \ff \,\mid\, \psi \wedge \psi \,\mid\, \psi \lor \psi\,\mid\, \pi=\pi \,\mid\, \pi\neq\pi  \,\mid\, \max x. \psi \,\mid\, \min x. \psi \,\mid\, x \,\mid\, [a_\pi]\psi \,\mid\, \langle a_\pi\rangle\psi
\end{align*}
We denote the class of formulas of type $\psi$ with $\quantfree$ (quantifier free).
\PHypermuHML is a subset of \HypermuHML and thus its semantics over $\Htrc_\Loc$ is the 
one given in \Cref{tab:semantics_hyper_mu_HML}.

We synthesize decentralized monitors for the fragment of \PHypermuHML only containing formulas of type 
$\varphi$ 
without 
least fixed-points, 
which we call \FPHypermuHML. 
The synthesis for decentralized monitors is given in \Cref{tab:synthesis-decentralized}. 
To derive monitors belonging to $\DMon$
for formulas of type $\varphi$ (through function $\MdName$),
we need to derive a monitor belonging to $\LMon$ for formulas of type $\psi\in\quantfree$ (through function $\MdLocName$). 
The
synthesis functions are parametrized by a location $\ell \in \Loc$ and a partial function $\sigma$ from $\Pi$ to $\Loc$ 
that is defined for every free location variable in $\psi$ and $\varphi$.


\begin{table}
    \hrule
    \[\def\arraystretch{1.3}\begin{array}{r@{\,}lr@{\,}l}
            \Mone{\forall\pi.\varphi} &=  \bigwedge_{\ell \in \Loc} \Monev{\sigma[\pi \mapsto \ell]}{\varphi} 
            & \qquad 
            \Mone{\exists\pi.\varphi} &=  \bigvee_{\ell \in \Loc} \Monev{\sigma[\pi \mapsto \ell]}{\varphi}
                \vspace*{.3cm}
		\\
            \Mone{\varphi \wedge \varphi'} &= \Mone{\varphi} \wedge \Mone{\varphi'}
         &\qquad
        \Mone{\varphi \vee \varphi'} &= \Mone{\varphi} \vee \Mone{\varphi'}
            \vspace*{.3cm}
	\\
    \multicolumn{4}{c}{
    \Mone{\psi} =  \begin{cases}
        \bigvee_{\ell \in \range(\sigma)} [\synone{\psi}]_\ell & \mbox{if } \sigma\neq\emptyset
        \\
        [v]_{\ell_0} & \mbox{if } \sigma=\emptyset \wedge \Mc{\psi}\valto v
        \end{cases}}
            \end{array}\]
    \hrule
\[\def\arraystretch{1.3}\begin{array}{r@{\,}lr@{\,}l}
    \multicolumn{4}{c}{
        \synone{\ttt} = \yes \qquad
        \synone{\ff} = \no    \qquad
    \synone{x} = x \qquad
    \synone{\max x.\psi} = \rec x.\synone{\psi} 
    }
    \vspace*{.3cm}
    \\
    \synone{\psi \wedge \psi'} &= \synone{\psi} \pand \synone{\psi'}
 &
\synone{\psi \vee \psi'} &= \synone{\psi} \por \synone{\psi'}
        \vspace*{.3cm}
	\\
    \synone{\pi=\pi'} &= \begin{cases}
    \yes & \mbox{if } \sigma(\pi)=\sigma(\pi')
    \\
    \no & \mbox{otherwise }
    \end{cases}
    &
    \synone{\pi\neq\pi'}  &= \begin{cases}
    \yes & \mbox{if } \sigma(\pi)\neq\sigma(\pi') \\
    \no & \mbox{otherwise }
    \end{cases}
        \vspace*{.3cm}
	\\ 
    \multicolumn{4}{c}{
        \synone{[a_\pi]\psi} =
        \left\{
        \begin{array}{ll} \displaystyle
        a.(!
        (\range(\sigma){\setminus} \{ \ell \})
        , a).\synone{\psi} + 
        \displaystyle
        \sum_{b \neq a} b.(!
        (\range(\sigma){\setminus} \{ \ell \})
        , b).\yes
        & \mbox{if } \sigma(\pi) = \ell 
        \vspace*{.3cm}
        \\
        \displaystyle
        \sum_{b \in \Act} b.\Big((?\{\sigma(\pi)\} , a).\synone{\psi}\ +
        \displaystyle
        \sum_{b\neq a}(?\{\sigma(\pi)\} , b).\yes\Big) & \mbox{otherwise}
        \end{array}
        \right.
}
        \vspace*{.3cm}
	\\
    \multicolumn{4}{c}{
        \synone{\langle a_\pi \rangle \psi} =
        \left\{
        \begin{array}{ll} \displaystyle
        a.(!
        (\range(\sigma){\setminus} \{ \ell \})
        , a).\synone{\psi} + 
        \displaystyle
        \sum_{b \neq a} b.(!
        (\range(\sigma){\setminus} \{ \ell \})
        , b).\no
        & \mbox{if } \sigma(\pi) = \ell 
        \vspace*{.3cm}
        \\
        \displaystyle
        \sum_{b \in \Act} b.\Big((?\{\sigma(\pi)\} , a).\synone{\psi}\ +
        \displaystyle
        \sum_{b\neq a}(?\{\sigma(\pi)\} , b).\no\Big) & \mbox{otherwise}
        \end{array}
        \right.
}
    \end{array}\]
            \hrule\vspace{0.5em}
            \caption{Decentralized monitor synthesis, where $\ell_0$ is any fixed element of $\Loc$.}
            \label{tab:synthesis-decentralized}
        \end{table}

Note that, in the definition of $\Mone{\psi}$, $\Mc{\psi}$ is the monitor resulting from the centralized synthesis function defined in \Cref{tab:synthesis}. Intuitively $\Mone{\psi}$ synthesizes a local monitor at each location relevant to $\psi$, which are the locations associated by $\sigma$ to the free location variables in $\psi$. If $\sigma=\emptyset$, then $\psi$ does not have any free location variable and it is a boolean combination of $\ttt$ and $\ff$; thus, there is no need for communication between locations. In fact, a verdict can be obtained from $\psi$ immediately
 and it coincides with the one reached in the centralized synthesis. 

Thus, every closed formula $\varphi$ on which we apply our synthesis is
\begin{enumerate}
    \item either {\em trivial}, i.e., $\varphi$ is logically equivalent to $\ttt$ or $\ff$,  
    \item or is such that every subformula $\psi\in \quantfree$ of $\varphi$ is in the scope of a quantifier.
\end{enumerate}
For non-trivial formulas, 
the $\sigma = \emptyset$ case for $\Mone{\psi}$ never applies, and we can ignore it.
The decentralized monitor for a closed formula $\varphi$ is $\Monev{\emptyset}{\varphi}$.

\begin{textAtEnd}[allend, category=syn]
The monitor $\Mone{\psi}$ is well-defined for closed formulas because of the following lemma:
\begin{lemma}\label{lemma:emptysigma}
If $\psi$ has no free variables, then $\Mc{\psi}\valto v$ for some $v$.
\end{lemma}
\begin{proof}
    The proof follows from induction on $\psi$, making use of the fact that all formulas are assumed to be guarded.
\end{proof}
\end{textAtEnd}

\begin{remark}
In the first clause of the definition of the synthesis function for box formulas, it might seem superfluous to send a message also when the monitor observes some $b\neq a$. However, this is important to make sure monitors do not deadlock. To see this, consider a synthesis where that definition instead looks like 
\[
\synone{[a_\pi]\psi} =\left\{
        \begin{array}{ll} \displaystyle
        a.(!
        (\range(\sigma){\setminus} \{ \ell \})
        , a).\synone{\psi} + 
        \displaystyle
        \sum_{b \neq a} b.
        \yes
        & \mbox{if } \sigma(\pi) = \ell 
        \\
        \displaystyle
        \sum_{b \in \Act} b.(?\{\sigma(\pi)\} , a).\synone{\psi}\ 
        \displaystyle
         & \mbox{otherwise}
        \end{array}
\right.
\]

Consider $\Act=\{a,b\}$, $\Loc=\{\ell,\ell'\}$ and some hypertrace $T$ such that $T(\ell)=b.t_1$ and $T(\ell')=b.t_2$ for some traces $t_1$ and $t_2$.
Now consider $m\pand n$,
where $m=\synone{[a_\pi]\psi}$, $n=\synvone{\ell}{\sigma}{[a_{\pi'}]\psi'}$, 
$\sigma(\pi)=\ell$ and $\sigma(\pi')=\ell'$. For $A(\ell)=A(\ell')=b\neq a$, we then get $m\xrightarrow{A(\ell)}\yes$ and $n\xrightarrow{A(\ell')}(?\{\sigma(\pi')\} , a).\synone{\psi'}$, and monitor $\yes \pand (?\{\sigma(\pi')\} , a).\synone{\psi'} $ is stuck because the receive action of the monitor $(?\{\sigma(\pi')\} , a).\synone{\psi'} $ has no matching send. It is precisely to avoid these scenarios that we make sure that, for each sending transition, there is a corresponding receiving transition, and a monitor always sends the last action it read to all other locations in the range of the environment $\sigma$.

The same applies to the synthesis of diamonds, being the latter the same as the synthesis for box formulas in Table~\ref{tab:synthesis-decentralized} with $\no$ verdicts in place of $\yes$.
\exqed
\end{remark}

    

\begin{example}\label{exp:wolper_decentralized}
In order to highlight the inter-monitor communication, we consider the following formula 
$$
\varphi = \exists \pi .\exists \pi' .([a_{\pi}] \ff \wedge [b_{\pi'}] \ff)
$$ 
over  $\Loc = \{1,2 \}$ and $\Act = \{a,b\}$, which states that \changed{the two traces must start with different actions. Indeed, $[a_{\pi}] \ff$ requires that the trace at the location bound to $\pi$ does not start with an $a$ and the one at the location bound to $\pi'$ does not start with a $b$. Hence, the only way to satisfy this formula is to associate different locations to $\pi$ and $\pi'$, in such a way that the trace at the location bound to $\pi$ starts with a $b$ and the trace at the location bound to $\pi'$ starts with an $a$.}
By letting $\sigma = [\pi\mapsto\ell,\pi'\mapsto\ell']$, the synthesis function applied to this property gives: 
\[
\Monev{\emptyset}{\varphi} 
= \bigvee_{\ell,\ell' \in \Loc} \ \ \bigvee_{\ell'' \in \{\ell,\ell'\}} \left[\synvone{\ell''}\sigma{[a_{\pi}] \ff \wedge [b_{\pi'}] \ff} \right]_{\ell''}
\]
where
\[
\synvone{\ell''}{\sigma}{[a_{\pi}] \ff \wedge [b_{\pi'}] \ff} = 
\left\{
\begin{array}{ll}
(a. (!\emptyset,a).\no + b.(!\emptyset,b).\yes) \ \pand\
&
\text{if $ \ell = \ell' = \ell''$}
\\
\quad (b. (!\emptyset,b).\no + a.(!\emptyset,a).\yes)
\vspace*{.2cm}
\\
(a.(!\{\ell'\},a).\no + b.(!\{\ell'\},b).\yes) \ \pand 
&
\text{if $\ell \neq \ell'$ and $\ell'' = \ell$}
\\
\quad (a.((?\{\ell'\},b).\no + (?\{\ell'\},a).\yes)\ + \\ 
\quad \ \, b.((?\{\ell'\},b).\no + (?\{\ell'\},a).\yes))
\vspace*{.2cm}
\\
(a.((?\{\ell\},a).\no + (?\{\ell\},b).\yes)\ + 
&
\text{if $\ell \neq \ell'$ and $\ell'' = \ell'$}
\\
\ \, b.((?\{\ell\},a).\no + (?\{\ell\},b).\yes))
\\
\pand\  (b.(!\{\ell\},b).\no + a.(!\{\ell\},a).\yes)
& \phantom{\text{if $\ell \neq \ell'$ and $\ell'' = \ell'$}}\blacktriangleleft
\end{array}
\right.
\]
\end{example}

\begin{remark}
\label{rem:diamond}
We explicitly added in Table \ref{tab:synthesis-decentralized} a rule for formulas of the form $\langle a_\pi \rangle \psi$ even though such formulas are logically equivalent to 
\begin{equation}
\label{Eq:diam2box}
[a_\pi] \psi \textstyle \wedge \bigwedge_{b \neq a} [b_\pi] \ff . 
\end{equation}
However, working with this formulation of the diamond has drawbacks in terms of the size of the resulting monitor.
To showcase this, consider the decentralized synthesis applied on Wolper's property ($\varphi^h_e$ of  \eqref{ex:Whprop} from Example \ref{ex:Two}), 
%
%
expressed here as $\exists \pi. \psi$, with 
\[
\psi = \max x. (\psi_1 \wedge \psi_2)
\qquad
\psi_1 = [a_{\pi}]\langle a_{\pi} \rangle x
\qquad
\psi_2 =  [b_{\pi}]\langle a_{\pi} \rangle x
\]
Let $\Loc = \{1,2 \}$ and $\Act = \{a,b\}$. 
The synthesis is applied thus: 
\[
\Monev{\emptyset}{\varphi^h_e} 
=  
\textstyle
\bigvee_{\ell \in \Loc} \left[\rec x. \left(\synvone\ell{[\pi \mapsto \ell]}{\psi_1} \pand \synvone\ell{[\pi \mapsto \ell]}{\psi_2} \right) \right]_\ell
\]
with
\[
\begin{array}{rcl}
\synvone\ell{[\pi \mapsto \ell]}{\psi_1}  & = &a.(!\emptyset,a).\synvone\ell{[\pi \mapsto \ell]}{\langle a_{\pi} \rangle x} + b.(!\emptyset,b).\yes 
\vspace*{.2cm}
\\ 
\synvone\ell{[\pi \mapsto \ell]}{\psi_2} & = &b.(!\emptyset,b).\synvone\ell{[\pi \mapsto \ell]}{\langle a_{\pi} \rangle x} + a.(!\emptyset,a).\yes
\end{array}
\]
and, by using the decentralized monitor synthesis for diamond formulas based on (\ref{Eq:diam2box}),
\begin{equation}
\label{Eq:diam2boxAtWork}
\synvone\ell{[\pi \mapsto \ell]}{\langle a_{\pi} \rangle x} = 
(a. (!\emptyset,a).x + b.(!\emptyset,b).\yes)
\ \pand\ 
(b. (!\emptyset,b).\no + a.(!\emptyset,a).\yes)
\end{equation}

This indicates that the synthesis for diamonds based on (\ref{Eq:diam2box}) leads to monitors with a high degree of parallelism; for simplicity, we have just two parallel components in \eqref{Eq:diam2boxAtWork} because we assumed to have just two actions. However,  $|\Act|-1$ parallel conjunctions are required in general for every diamond (and may lead to an exponential blow up with several diamonds in sequence). 
By contrast, using the rule for diamonds given in Table \ref{tab:synthesis-decentralized} reduces \eqref{Eq:diam2boxAtWork} to:
\[
\synvone\ell{[\pi \mapsto \ell]}{\langle a_{\pi} \rangle x} = 
a. (!\emptyset,a).x + b.(!\emptyset,b).\no
\]
and the synthesized monitor now contains no occurrence of any parallel operator.
\exqed
\end{remark}

Soundness and violation completeness for the synthesis defined in \Cref{tab:synthesis-decentralized} follow from \Cref{cor:soundDec} and \ref{cor:compDec} by using \Cref{thm:principled-bisimulation}, once we prove the following key result:

\begin{theorem}
\label{thm:principled}
The synthesis function $\MdName$ defined in \Cref{tab:synthesis-decentralized} is principled.
\end{theorem}

The proof is long and technical. It is carried out by showing that $\MdName$ satisfies the six properties 
of \Cref{def:principled-synthesis}. The details are in Appendix~\ref{app:corr}.

\begin{textAtEnd}[allend, category=corr] 
We now prove that the synthesis defined in \Cref{tab:synthesis-decentralized} satisfies the properties in \Cref{def:principled-synthesis} by proving one property in each of the following subsections. 
We start with some preliminary definitions and results that will be used for proving different properties.

We start with a lemma that is the equivalent of \Cref{lemma:recursion-substitution} for $\synone{-}$; the proof is identical. Then, we show that $\Mone{\varphi}$ cannot communicate.
        
        \begin{lemma}\label{lemma:recursion-substitutiontwo}
            If $\synone{\psi_1} = m$, $\synone{\psi_2} = n$, and the free location variables of $\psi_2$ are not bound in $\psi_1$, then $\synone{\psi_1\{{}^{\psi_2}/{}_x\}} = m\{{}^n/{}_x\}$.
        \end{lemma}

\begin{lemma}
    \label{lemma:nocommunicate}
    Let $\Mone{\varphi}$ be defined; then, for all $[m]_\ell\in \Mone{\varphi}$ and  $c\in\{(!G,\gamma),(?\ell,\gamma)\mid G\subseteq\Loc, \ell\in\Loc, \gamma\in\Cvar\}$, we have that $m\NOT{\xrightarrow{c}}$.
\end{lemma}
\begin{proof}
The statement follows from a straightforward induction on $\varphi$ and the following claim: for all $\ell,\psi,\sigma$ and $c\in\{(!G,\gamma),(?\ell,\gamma)\mid G\subseteq\Loc, \ell\in\Loc, \gamma\in\Cvar\}$, it holds that $\synone{\psi}\NOT{\xrightarrow{c}}$. 
Such claim can be proved by induction on $\xrightarrow{c}$ to show that, if $m\xrightarrow{c}n$, then $m\neq \synone{\psi}$ for all $\sigma,\psi,\ell$.
\end{proof}

As most results in this work concern monitors that are reached in a run from a synthesized decentralized monitor, it is convenient to capture this in a definition. 

\begin{definition}\label{def:A-derived}
    We say that a monitor $M$ is \textit{$A$-derived} if, for all $[m]_\ell\in  M$, there exist $\psi, \sigma, m'$ such that $\synone{\psi}\xrightarrow{A(\ell)} m'\boldarrow m$. A monitor $M$ is \textit{action-derived} if there exists an $A$ such that $M$ is $A$-derived. We call a local monitor $m$ action-derived/A-derived for $\ell$ if $[m]_\ell$ is action-derived/A-derived.
\end{definition}

We define a syntactic characterization of the kinds of monitors that are action-derived. This provides us with a convenient inductive proof principle. 

\begin{definition}\label{def:relevant}
    The class of \emph{relevant local monitors} is defined for $\psi,\psi_a\in\quantfree$ by the grammar:
    \begin{align*}
        m ::=~ & \synone{\psi} &~\mid~ & (!G,a).\synone{\psi} &~\mid~ & 
        \sum_{a\in\Act}(?\{\ell'\},a).\synone{\psi_a}  ~
            \mid~ & m \pand m &~\mid~ & m \por m . &&\hfill 
    \end{align*}
    The class of relevant communicating monitors is then defined with the following grammar:
    \begin{align*}
        M ::=~~ 
         [m]_\ell ~~\mid~~  M \land M ~~ \mid~~  M \lor M , \hfill 
    \end{align*}
    where $m$ is a relevant local monitor and $\ell \in \Loc$.
    When it is clear from the context if a monitor is a relevant local monitor or a relevant communicating monitor, we will simply call it a relevant monitor, or say that the monitor is relevant.
\end{definition}

The classes of relevant communicating/local monitors are closed under transitions:

\begin{lemma}\label{lemma:stay_relevant}
    Let $M$ and $m$ be a relevant communicating monitor and a relevant local monitor, respectively. If $M \xrightarrow{\lambda} M'$ and $m \xrightarrow{\lambda} m'$, then $M'$ and $m'$  are a relevant communicating monitor and a relevant local monitor, respectively.
\end{lemma}

\begin{proof}
The proof for $M \xrightarrow{\lambda} M'$ follows directly from the result for $m \xrightarrow{\lambda} m'$;
this can be proved by induction on $\xrightarrow{\lambda}$ and the only interesting case is for $m=\synone{\psi}$. 
By \Cref{lemma:nocommunicate}, 
it must be that $\lambda=a\in\Act$.
The only base cases are $\synone{\psi}=a.m'$, but this cannot occur because of definition of the synthesis function and our assumption that $|\Act| \geq 2$, or $\synone{\psi}=m'=v$, that is trivial. 

For the inductive case, when $\synone{\psi}=n_1+n_2$ and $n_1\xrightarrow{a}m'$, $\psi$ must be of the form $[a'_\pi]\psi'$
or $\langle a'_\pi\rangle \psi'$;
so 
$$
\synone{\psi}=a'.(!\{\sigma(\pi')\mid\sigma(\pi' )\neq \ell\},a').\synone{\psi'}+\sum_{b\neq a'}b.(!\{\sigma(\pi')\mid\sigma(\pi' )\neq \ell\},b).v
$$ 
or 
$$
\synone{\psi}=\sum_{b\in \Act}b.\big((?\{\sigma(\pi)\},a').\synone{\psi'}+\sum_{b\neq a'}(?\{\sigma(\pi)\},b).v\big)
$$ 
where $v$ is $\yes$, if $\psi$ is a box, and is $\no$, if $\psi$ is a diamond.
In all cases, it is immediate that $m'$ is a relevant monitor, regardless of whether $a=a'$ or not.
    
    Next we consider the case for $\synone{\psi}=m=m_1\pandor m_2$ and $m'=m_1'\pandor m_2'$. Hence, $\psi=\psi_1\op \psi_2$ (for $\op = \wedge$ if $\pandor = \pand$ and $\op = \vee$ if $\pandor = \por$), $m_1=\synone{\psi_1}\xrightarrow{a}m_1'$ and $m_2=\synone{\psi_2}\xrightarrow{a}m_2'$. From the induction hypothesis it follows that $m_1'$ and $m_2'$ are relevant monitors, thus making $m'$ relevant as well.

    Last we consider $\synone{\psi}=\rec x.n\xrightarrow{a} m'$ because $n\{^{\rec x.n}/_x\}\xrightarrow{a} m'$. Thus we obtain that $\psi=\max x.\psi'$ and $n=\synone{\psi'}$. As $\synone{\psi'}\{^{\rec x.\synone{\psi'}}/_x\}=\synone{\psi'\{^{\maxx x.\psi'}/_x\}}$ (\Cref{lemma:recursion-substitutiontwo}), we use the induction hypothesis to obtain that $m'$ is relevant.
\end{proof}

Relevant monitors can be used to prove properties of action-derived monitors, because of the following:

\begin{lemma}\label{lemma:A-derived-iff-relevant}
    Every action-derived monitor $M$ is relevant.
\end{lemma}
\ifarxiversion 
\begin{proof}
    Straightforward from \Cref{lemma:stay_relevant}.
\end{proof}
\fi
\end{textAtEnd}

\begin{textAtEnd}[allend, category=va] 
\label{sec:va}
    
To relate verdicts of centralized and decentralized monitors, we need the following intermediary result:
\begin{lemma}\label{lemma:threefacts}
    Let $\psi\in\quantfree$ and $\varphi\in$~\PHypermuHML.
    \begin{enumerate}
   \item If there exists $\ell\in\Loc$ for which $\synvone{\ell}{\sigma}{\psi}= v $, then $\Mc{\psi}=v$.
   \item If $\Mc{\psi}=v$, then for all $\ell\in\Loc$ we have $\synvone{\ell}{\sigma}{\psi}= v $. 
    \end{enumerate}
\end{lemma}
\ifarxiversion 
\begin{proof}
For the first statement we only prove the case for $v=\yes$, the case for $v=\no$ is dual. We observe that $\synvone{\ell}{\sigma}{\psi}= \yes $ implies the following: $\psi=\ttt$, $\psi=(\pi=\pi')$ and $\sigma(\pi)=\sigma(\pi')$, or $\psi=(\pi\neq \pi')$ and $\sigma(\pi)\neq\sigma(\pi')$. In all these cases we obtain immediately that $\Mc{\psi}=\yes$. 

For the second statement, we prove only the case with $v=\no$. We observe that $\Mc{\psi}=\no$ implies that $\psi=\ff$, $\psi=(\pi=\pi')$ and $\sigma(\pi)\neq\sigma(\pi')$, or $\psi=(\pi\neq \pi')$ and $\sigma(\pi)=\sigma(\pi')$ (by definition, $\psi$ is quantifier-free). In all these cases we have  that $\synvone{\ell}{\sigma}{\psi}= \no $, for all $\ell$.
\end{proof}
\fi

We first relate the verdicts of centralized and decentralized synthesized monitors for $\psi$-formulas.
\begin{lemma}\label{lemma:verdictind}
    Let $\sigma$ be such that $\FVloc(\psi)\subseteq \dom(\sigma)$. 
     If there exists $\ell\in\Loc$ such that $\synvone{\ell}{\sigma}{\psi}\valto v $, then $\Mc{\psi}\valto v$. Conversely, if $\Mc{\psi}\valto v$, then $\synvone{\ell}{\sigma}{\psi}\valto v $, for all $\ell\in\Loc$. 
\end{lemma}
\begin{proof}
First, observe that $\psi$ is quantifier-free (by definition) and that the definition of $\Mc{\psi}$ is the same as the definition of $\synone{\psi}$ except for the box and diamond modalities.
Hence, we only prove the first statement by induction on $\valto$, as the other proof is identical. 
 


For the base case, we can skip the case for $\vend$ (as a synthesized monitor never takes that shape) and so the only possibility is $\synone{\psi}=v\valto v$: the result follows immediately from \Cref{lemma:threefacts}. 

For the inductive case, we can skip the case for $+$: if $\synone{\psi}=m+n$, then either $\psi=[a_\pi]\psi'$ or $\psi=\langle a_\pi\rangle \psi'$; from the synthesis function, it then follows that $\synone{\psi}$ cannot evaluate to a verdict. 

If $\synone{\psi}=m\pand n\valto \no$ because $m\valto \no$, then $\psi=\psi_1\wedge\psi_2$ and $\synvone{\ell}{\sigma}{\psi_1}=m$. 
We use the induction hypothesis to derive that $\Mc{\psi_1}\valto \no$ and thus $\Mc{\psi}\valto v$.
The case for $\synone{\psi}=m\por n\valto \yes$ because $m\valto \yes$ is dual.

Next we consider 
$\synone{\psi}=m\pand n\valto v$ because $m \valto \yes$ and $n\valto v$. Thus $\psi=\psi_1\wedge\psi_2$, $\synvone{\ell}{\sigma}{\psi_1}=m$ and $\synvone{\ell}{\sigma}{\psi_2}=n$. We use the induction hypothesis to derive that $\Mc{\psi_1}\valto \yes$ and $\Mc{\psi_2}\valto v$, which concludes the proof.
%
The case for $\synone{\psi}=m\por n\valto v$ because $m \valto \no$ and $n\valto v$ is dual.

Finally, if $\synone{\psi}=\rec x.m\valto v$ because $m\{^{\rec x.m}/_x\}\valto v$, then $\psi=\maxx x.\psi'$ and $m=\synone{\psi'}$. As $\synone{\psi'}\{^{\rec x.\synone{\psi'}}/_x\}=\synone{\psi'\{^{\maxx x.\psi'}/_x\}}$, from \Cref{lemma:recursion-substitutiontwo}, we get that $\synone{\psi'\{^{\maxx x.\psi'}/_x\}}\valto v$. By the induction hypothesis, we obtain that  $\Mc{\psi'\{^{\maxx x.\psi'}/_x\}}\valto v$. This, together with $\Mc{\psi'}\{^{\rec x.\Mc{\psi'}}/_x\}=\Mc{\psi'\{^{\maxx x.\psi'}/_x\}}$ (that holds by \Cref{lemma:recursion-substitution}), implies that $\Mc{\psi'}\{^{\rec x.\Mc{\psi'}}/_x\}\valto v$. The latter gives that $\Mc{\psi} = \rec x. \Mc{\psi'} \valto v$. 
\end{proof}

The next lemma states a few basic facts about verdict evaluation in a centralized and decentralized setting; the proof is immediate. Notationally, when $I=\emptyset$, we let $\bigvee_{i\in I} m_i=\no=\bigoplus_{i\in I} m_i$ and $\bigwedge_{i\in I} m_i=\yes=\bigotimes_{i\in I} m_i$.

\begin{lemma}\label{lemma:basicfactsverdict}
\ 
\begin{enumerate}
    \item $\bigoplus_{i\in I} m_i\valto \yes \ \ \Leftrightarrow\ \ \exists i\in I. m_i\valto \yes\ \ \Leftrightarrow\ \ \bigvee_{i\in I} m_i\valto \yes$;
    \item $\bigoplus_{i\in I} m_i\valto \no \ \ \Leftrightarrow\ \ \forall i\in I. m_i\valto \no  \ \ \Leftrightarrow\ \ \bigvee_{i\in I} m_i\valto \no$;
    \item $\bigotimes_{i\in I} m_i\valto \yes\ \ \Leftrightarrow\ \ \forall i\in I. m_i\valto \yes\ \ \Leftrightarrow\ \ \bigwedge_{i\in I} m_i\valto \yes$;
    \item $\bigotimes_{i\in I} m_i\valto \no\ \ \Leftrightarrow\ \ \exists i\in I. m_i\valto \no\ \ \Leftrightarrow\ \ \bigwedge_{i\in I} m_i\valto \no$.
\end{enumerate}
\end{lemma}

Now we can relate verdicts for monitors synthesized from any formula in the centralized and decentralized setting, thus proving Verdict Agreement.

\begin{lemma}\label{lemma:verdictbis}
Let $\varphi\in$\FPHypermuHML and $\sigma$ be such that $\FVloc(\varphi)\subseteq\dom(\sigma)$. Then, 
$\Mone{\varphi}\valto v \Leftrightarrow \Mc{\varphi}\valto v$.
\end{lemma}
\begin{proof}
    The proof proceeds by induction on $\varphi$. From the synthesis functions we know that $v\neq \vend$. For the base case we consider $\varphi=\psi$. We distinguish when 
$\sigma = \emptyset$ or not. If $\sigma=\emptyset$, 
    by assumption we know that $\FVloc(\psi)=\emptyset$ and so $\psi$ is trivial; 
    therefore, $\Mc\psi = v \valto v$ and $\Mone{\psi} = [v]_{\ell_0} \valto v$.
    Let us now consider the case for $\sigma\neq\emptyset$.
If $v=\no$, we derive:
\begin{align*}
    \Mone{\psi} &=  \bigvee_{\ell \in \range(\sigma)} [\synone{\psi}]_\ell \valto \no 
    \\
           &\Leftrightarrow \forall \ell\in \range(\sigma). [\synone{\psi}]_\ell \valto \no \tag{\Cref{lemma:basicfactsverdict}}\\
    &\Leftrightarrow \forall \ell\in \range(\sigma) . \synone{\psi} \valto \no \\
    &\Leftrightarrow \Mc{\psi}\valto \no \tag{\Cref{lemma:verdictind}}
    \end{align*}
If $v=\yes$, we have:
\begin{align*}
    \Mone{\psi}&=  \bigvee_{\ell \in \range(\sigma)} [\synone{\psi}]_\ell \valto \yes \\
       &\Leftrightarrow \exists \ell\in \range(\sigma). [\synone{\psi}]_\ell \valto \yes \tag{\Cref{lemma:basicfactsverdict}}\\
       &\Leftrightarrow \exists \ell\in \range(\sigma) . \synone{\psi} \valto \yes\\
&\Leftrightarrow \Mc{\psi}\valto \yes \tag{\Cref{lemma:verdictind}}
\end{align*}
For the inductive step, we first consider the case for $\exists \pi.\varphi$. 
    If $v=\yes$, we derive 
    \begin{align*}
        \Monev{\sigma}{\exists \pi.\varphi} &=  \bigvee_{\ell \in \Loc} \Monev{\sigma[\pi \mapsto \ell]}{\varphi}\valto \yes \\
        &\Leftrightarrow \exists \ell\in\Loc.\Monev{\sigma[\pi \mapsto \ell]}{\varphi}\valto \yes \tag{\Cref{lemma:basicfactsverdict}}\\
        &\stackrel{IH}\Leftrightarrow \exists \ell\in\Loc.\Mcv{\varphi}{\sigma[\pi \mapsto \ell]}\valto \yes \\
        &\Leftrightarrow \bigoplus_{\ell \in \Loc} \Mcv{\varphi}{\sigma[\pi \mapsto \ell]}=\Mc{\exists \pi.\varphi}\valto \yes \tag{\Cref{lemma:basicfactsverdict}}
    \end{align*}
    The case for $v=\no$ is dual.
    The cases for $\forall \pi.\varphi$, $\varphi_1\wedge\varphi_2$ and $\varphi_1\vee\varphi_2$ are similar.
\end{proof}
\end{textAtEnd}

\begin{textAtEnd}[allend, category=vi]
\label{sec:vi}

We start by proving Verdict Irrevocability at the level of local monitors.
\begin{lemma}\label{lemma:verdictproplow}
    Let $m$ be a relevant local monitor. If $m\xrightarrow{c}n$ and $m\valto v$, then $n\valto v$.
    \end{lemma}
    \begin{proof}
    We proceed by induction on $m$. 
    The base case for $m=\synone{\psi}$ can be concluded from \Cref{lemma:nocommunicate}; in the base cases for $m=(!G,a).\synone{\psi}$ and $m=\sum_{a\in\Act}(?\{\ell'\},a).\synone{\psi_a}$, the premise $m \valto v$ is not satisfied.   
    
    For the inductive step,
if $m=n_1\pand n_2$ with $n_1$ and $n_2$ relevant monitors, then $m\valto v$ yields two cases (by \Cref{lemma:valtoverdict}, $v\neq \vend$ and we can thus exclude that case):
        \begin{enumerate}
    \item $n_1\valto \no$ and $v=\no$. 
    From $m\xrightarrow{c}n$ we derive another three cases, where we let $n=o_1\pand o_2$:
    \begin{enumerate}
        \item $n_1\xrightarrow{c} o_1$ and $n_2=o_2$. We use the induction hypothesis to obtain that $o_1\valto \no$. Then, $n \valto \no$ follows immediately. 
        
        \item $n_2\xrightarrow{c} o_2$ and $n_1=o_1$. Like before, $o_2\valto\no$ and $n \valto \no$.
        
    
        
    
        
         \item $c=(?\ell',\gamma)$, $n_1\xrightarrow{c}o_1$ and $n_2\xrightarrow{c}o_2$. By induction, $o_1\valto\no$, $o_2\valto\no$ and $n \valto \no$.
        
    \end{enumerate}
    
    \item $n_1\valto \yes$ and $n_2\valto v$. 
    This case is symmetric.
    \end{enumerate}
        The case for $\por$ is similar to the one for $\pand$ and, therefore, we omit it.
    \end{proof}
    
    Next we show that the verdict is not $\vend$ and that verdicts carry over receiving transitions:
    \begin{lemma}\label{lemma:endfree}
        Let $M\in\DMon$ be action-derived. If $M\valto v$, then $v\neq \vend$.
    \end{lemma}
    \ifarxiversion
    \begin{proof}
    This follows immediately from a simple induction on $\valto$ and \Cref{lemma:valtoverdict}.
    \end{proof}
    \fi
    
    \begin{lemma}\label{lemma:verdictproponestepreceive}
        Let $M\in\DMon$ be action-derived.
        If $M\extoverset{$G:(?\ell, \gamma)$}\rightsquigarrow  N$ and $M\valto v$, then $N\valto v$.
    \end{lemma}
    \begin{proof}
        We proceed by induction on $M\extoverset{$G:(?\ell, \gamma)$}\rightsquigarrow  N$. For the first base case, we consider $M=[m]_{\ell'}$, $N=[n]_{\ell'}$ and $m\xrightarrow{(?\ell,\gamma)}n$: then, $M\valto v$ implies that $m\valto v$, and (via \Cref{lemma:verdictproplow}, that can be used thanks to \Cref{lemma:A-derived-iff-relevant}) we can conclude $n\valto v$, thus $N\valto v$. In the other two base cases, we have $M=N$ and the result follows immediately. 
    
    For the inductive step we consider $M=M_1\pand M_2$ (the case for $M=M_1\por M_2$ is similar), $N=N_1\pand N_2$, $M_1\extoverset{$G:(?\ell, \gamma)$}\rightsquigarrow N_1$ and $M_2\extoverset{$G:(?\ell, \gamma)$}\rightsquigarrow  N_2$. From $M_1\pand M_2\valto v$ we distinguish two cases, where we use \Cref{lemma:endfree} to exclude the case for $v=\vend$:
    \begin{enumerate}
    \item $M_1\valto \yes$ and $M_2\valto v$, From the induction hypothesis we get $N_1\valto \yes$ and $N_2\valto v$. We immediately obtain that $N_1\pand N_2\valto v$ according to the transition rules for $\valto$.
    \item $M_1\valto \no$ and $v=\no$. We use the induction hypothesis to obtain that $N_1\valto \no$, and the result follows immediately.\qedhere


    \end{enumerate}
    \end{proof}
    

We first prove Verdict Irrevocability for just one communication step;
Verdict Irrevocability then easily follows via straightforward induction.

    \begin{lemma}\label{lemma:verdictproponestep}
        Let $M\in\DMon$ be action-derived.
        If $M\xrightarrow{c}N$ and $M\valto v$, then $N\valto v$.
    \end{lemma}
    \begin{proof}
Let $c=\ell:(!G,\gamma)$; we proceed by induction on $M\xrightarrow{c}N$. For the base case, we consider $M=[m]_\ell$, $N=[n]_\ell$ and $m\xrightarrow{(!G,\gamma)}n$: then, $M\valto v$ implies that $m\valto v$; we conclude by \Cref{lemma:A-derived-iff-relevant} and \ref{lemma:verdictproplow}.
    
    For the inductive step, we consider $M=M_1\pand M_2$ (the case for $M=M_1\por M_2$ is similar), $N=N_1\pand N_2$, $M_1\xrightarrow{c}N_1$ and $M_2\extoverset{$G:(?\ell, \gamma)$}\rightsquigarrow N_2$. From $M_1\pand M_2\valto v$ we distinguish four cases, where we use \Cref{lemma:endfree} to exclude the case for $v=\vend$:
    \begin{enumerate}
\item $M_1\valto \no$, $v=\no$. We use the induction hypothesis to obtain that $N_1\valto \no$, and the result follows.
\item $M_2\valto \no$, $v=\no$. We use \Cref{lemma:verdictproponestepreceive} to derive that $N_2\valto \no$, which concludes the proof.
\item $M_1\valto \yes$, $M_2\valto v$. We use the induction hypothesis to obtain that $N_1\valto \yes$ and \Cref{lemma:verdictproponestepreceive} to derive that $N_2\valto v$, and the result follows.
\item $M_2\valto \yes$, $M_1\valto v$. Same as the previous case.
    \qedhere
    

    
    
    
    \end{enumerate}
    \end{proof}

\begin{corollary}\label{lemma:verdictprop}
    If $\Mone{\varphi}\xrightarrow{A} M_1\boldarrow  M_2\boldarrow M$ and $M_2\valto v$, then $M\valto v$.
\end{corollary}
\ifarxiversion
\begin{proof}
By induction on the number of transitions in $M_2\boldarrow M$. In the base case, $M_2=M$ and the claim is trivial. For the inductive step, we consider $M_2\xrightarrow{c}N\boldarrow M$: via \Cref{lemma:verdictproponestep}, we get that $N\valto v$ and the result then follows from the induction hypothesis.
\end{proof}
\fi
\end{textAtEnd}

\begin{textAtEnd}[allend, category=react]
\label{sec:react}

We start by defining an ordering on the formulae of kind $\psi$, based on the number of top-level fixed-point operators in them, where by ``top-level'' we mean ``not under the scope of a modality''. For example, the number of top-level fixed-points in 
\begin{equation}
\label{ex:order}
\max x.\max y.([a_\pi] x \wedge \langle b_{\pi'} \rangle y)\ \ \vee\ \ \max z.[a_{\pi}]\max t.(\langle b_{\pi'} \rangle t \wedge z)
\end{equation}
is 3,  $\max t$ being the only non-top-level fixed-point operator.

\begin{definition}
\label{def:order}
We let $\psi \prec \psi'$ iff
\begin{enumerate} 
\item the number of top-level fixed-point operators in $\psi$  is smaller than the number of top-level fixed-point operators in  $\psi'$, or 
\item the number is the same and $\psi$  is a strict sub-formula of  $\psi'$. 
\end{enumerate}
\end{definition}

\noindent
Notice that $\prec$ is well-founded and that the bottom elements are 
$\ttt$, $\ff$, $\pi=\pi$, $\pi\neq\pi$ and $x$. Coming back to \eqref{ex:order}, we have that
$$
\begin{array}{rcl}
\max z.[a_{\pi}]\max t.(\langle b_{\pi'} \rangle t \wedge z) &\prec& (\max x.[b_\pi] x \wedge \max y.\langle a_{\pi'} \rangle y)
\\
\max t.(\langle b_{\pi'} \rangle t \wedge z) &\prec& \max z.[a_{\pi}]\max t.(\langle b_{\pi'} \rangle t \wedge z)
\end{array}
$$
respectively by point 1 and 2 of \Cref{def:order}.

\begin{lemma}
    \label{lem:reactivity}
For all $\varphi\in$ \FPHypermuHML without free recursion variables, $\sigma$ such that $\FVloc(\varphi)\subseteq\dom(\sigma)$, and $A: \Loc \rightarrow \Act$, there exists $M\in\DMon$ such that $\Mone{\varphi}\xrightarrow{A} M$. 
\end{lemma}
\begin{proof}
   We proceed by induction on $\varphi$;
the proof is trivial for all the inductive cases. So we only treat the base case, where $\varphi=\psi\in \quantfree$ 
and has no free recursion variables.
We prove that $\synone{\psi}$ is reactive by induction on $\prec$.
The base cases are the bottom elements of $\prec$ and they all entail that $\synone{\psi}$ is a verdict; then, the claim follows immediately because $v \xrightarrow{a} v$ holds for each $v$ and $a$. 
For the inductive case:
\begin{itemize} 
\item If $\psi = [a_\pi] \psi'$ or $\psi = \langle a_\pi\rangle \psi'$, for some $a$, $\pi$ and $\psi'$, then $\synone{\psi}$ is reactive by construction. 
\item If $\psi = \psi_1 \op \psi_2$, then $\synone{\psi} = \synone{\psi_1} \pandor \synone{\psi_2}$ (with $\pandor = \pand$ if $\op = \wedge$ and $\pandor = \por$ if $\op = \vee$). By \Cref{def:order}, $\psi_1 \prec \psi$ and $\psi_2 \prec \psi$; hence, the inductive hypothesis yields that $\synone{\psi_1}$ and $\synone{\psi_2}$ are reactive. Reactivity of $\synone{\psi}$ immediately  follows from the operational rules for $\pandor$.
\item If $\psi = \maxx {x}.\psi'$, we show that, for every $a \in \Act$, there exists $m\in\LMon$ such that $\synone{\maxx {x}.\psi'}=\rec x.\synone{\psi'}\xrightarrow{a}m$. Since formulas are guarded, $\psi'\{^{\maxx{x}.\psi'}/_x\} \prec \maxx {x}.\psi'$; therefore, the induction hypothesis yields that $\synone{\psi'\{^{\maxx{x}.\psi'}/_x\}}\xrightarrow{a}m$ for some $m$. By \Cref{lemma:recursion-substitutiontwo}, we have that $\synone{\psi'}\{^{\rec x.\synone{\psi'}}/_x\}=\synone{\psi'\{^{\maxx{x}.\psi'}/_x\}}$, which concludes the proof, by using the operational rule for recursive monitors. 
\vspace*{-.5cm}
\end{itemize}
\end{proof}
\end{textAtEnd}

\begin{textAtEnd}[allend, category=bc]
\label{sec:bc}

We first show the existence of a finite communication path to a monitor that cannot communicate, and then argue that this monitor is unique. 
To this aim, we start with a few lemmas that are all stating basic facts regarding communication, and are used in many of the lemmas that follow.
    We start with three basic facts regarding receiving and sending of messages for action-derived local monitors: the first item states that local monitors can never receive the same message from the same sender in consecutive communications; the second one states that, if a monitor cannot receive a certain message from some sender, then it cannot receive that message in the future after taking only communication transitions; the last one states that monitors can only send the action they read last. 

    \begin{lemma}\label{lemma:singlesender}
        Let $m\in\LMon$ be $A$-derived for $\ell$. Then:
        \begin{enumerate}
            \item If $m\xrightarrow{(?\ell',\gamma)}n$, then $n\NOT{\xrightarrow{(?\ell',\gamma)}}$.
            \item If $m \NOT {\xrightarrow{(?\ell',\gamma)}}$ and $m\boldarrow n$, then $n \NOT {\xrightarrow{(?\ell',\gamma)}}$.
            \item If $m\xrightarrow{(!G,\gamma)} n$, then $\gamma=A(\ell)$.
        \end{enumerate}
        \end{lemma}
    \begin{proof}
        We prove the first two items for relevant local monitors, which is sufficient because $m$ is a relevant local monitor via \Cref{lemma:A-derived-iff-relevant}. We proceed by induction on the structure of $m$. 
    
    In the base case for $m=\synone{\psi}$, by \Cref{lemma:nocommunicate} $m$ cannot communicate; hence, the premise of the first item is not satisfied and in the second item we get $n=m$, making the statement trivially true. 
 For the base case with $m=(!G,a).\synone{\psi}$, the premise of the first item is not satisfied; for the second item, $n=m$ or $n=\synone{\psi}$, and the result follows from the premise of the lemma or \Cref{lemma:nocommunicate}. 
    If $m=\sum_{a\in\Act}(?\{\ell''\},a).\synone{\psi_a}$, both items follow immediately from \Cref{lemma:nocommunicate}. 
    
    For the inductive step, let $m=m_1\pandor m_2$, with $m_1$ and $m_2$ relevant monitors. For the first statement, we derive two cases from $m_1\pandor m_2\xrightarrow{(?\ell',\gamma)} n$.
    \begin{enumerate}
        \item $n=n_1\pandor n_2$ such that $m_1 \xrightarrow{(?\ell',\gamma)} n_1$ and $m_2 \xrightarrow{(?\ell',\gamma)} n_2$. The induction hypothesis gives that $n_1\NOT {\xrightarrow{(?\ell',\gamma)}}$ and $n_2\NOT {\xrightarrow{(?\ell',\gamma)}}$. Hence, $n\NOT {\xrightarrow{(?\ell',\gamma)}}$. 
        \item $n=n_1\pandor m_2$ such that  $m_1 \xrightarrow{(?\ell',\gamma)} n_1$ and  $m_2 \NOT {\xrightarrow{(?\ell',\gamma)}} $. From the induction hypothesis, we get that $n_1\NOT {\xrightarrow{(?\ell',\gamma)}}$. Hence, $n=n_1\pandor m_2\NOT {\xrightarrow{(?\ell',\gamma)}}$.
            \end{enumerate}
    For the second statement we derive from $m_1\pandor m_2 \NOT {\xrightarrow{(?\ell',\gamma)}}$ that $m_1 \NOT {\xrightarrow{(?\ell',\gamma)}}$ and $ m_2 \NOT {\xrightarrow{(?\ell',\gamma)}}$. Furthermore, we know that $n=n_1\pandor n_2$ such that $m_1\boldarrow n_1$ and $m_2\boldarrow n_2$. From the induction hypothesis, we obtain that $n_1 \NOT {\xrightarrow{(?\ell',\gamma)}}$ and $ n_2 \NOT {\xrightarrow{(?\ell',\gamma)}}$. Hence, $n\NOT {\xrightarrow{(?\ell',\gamma)}}$.

    Let us now prove the third item. 
   Because $m$ is $A$-derived for $\ell$, we know that there exist $\sigma$, $\psi$ and $m'$ such that $\synone{\psi}\xrightarrow{A(\ell)} m' \boldarrow m$; we proceed by induction on $\xrightarrow{A(\ell)}$.
   In the base case $\synone{\psi}=a.m'$; from the synthesis function and our assumption that $|\Act| \geq 2$, we observe that this never occurs.
   
    For the inductive case, if $\synone{\psi}=m_1+m_2$ and $m_1\xrightarrow{A(\ell)}m'$, $\psi$ must be of the form $[a_\pi]\psi'$ 
or $\langle a_\pi \rangle\psi'$;
so 
    \[\synone{\psi}=a.(!\{\sigma(\pi')\mid\sigma(\pi' )\neq \ell\},a).\synone{\psi'}+\sum_{b\neq a}b.(!\{\sigma(\pi')\mid\sigma(\pi' )\neq \ell\},b).v\] 
    or \[\synone{\psi}=\sum_{b\in \Act}b.\big((?\{\sigma(\pi)\},a).\synone{\psi'}+\sum_{b\neq a}(?\{\sigma(\pi)\},b).v\big)\] 
where $v = \yes/\no$ according to the modality (resp., box/diamond).
    It is immediate that $\gamma=A(\ell)$ in $m\xrightarrow{(!G,\gamma)}n$.
    
        
    Next we consider the case for $\synone{\psi}=m_1''\pandor m_2''$ and $m'=m_1'\pandor m_2'$. Hence, $\psi=\psi_1\op \psi_2$ (for $\op=\wedge$ if $\pandor = \pand$ and $\op=\vee$ if $\pandor = \por$), $m_1''=\synone{\psi_1}\xrightarrow{A(\ell)}m_1'$ and $m_2''=\synone{\psi_2}\xrightarrow{A(\ell)}m_2'$. Let $m=m_1\pandor m_2$, where $m_1'\boldarrow m_1$ and $m_2'\boldarrow m_2$. 
      We assume without loss of generality that $n=n_1\pandor m_2$ and $m_1\xrightarrow{(!G,\gamma)} n_1$. From the induction hypothesis, it follows that $\gamma=A(\ell)$.
        
    Last we consider $\synone{\psi}=\rec x.m''\xrightarrow{A(\ell)} m'$ because $m''\{^{\rec x.m''}/_x\}\xrightarrow{A(\ell)} m'$. Thus, $\psi=\max x.\psi'$ and $m''=\synone{\psi'}$. By \Cref{lemma:recursion-substitutiontwo}, $\synone{\psi'}\{^{\rec x.\synone{\psi'}}/_x\}=\synone{\psi'\{^{\maxx x.\psi'}/_x\}}$ and we use the induction hypothesis to obtain the desired result.
    \end{proof}

We generalize the third property from \Cref{lemma:singlesender} to communicating monitors: if an action-derived monitor $M$ can take a communicating transition, the message it sends is always the last action taken by the location of the sender:
    \begin{lemma}\label{lemma:sendright}
        Let $M\in\DMon$ be $A$-derived. If $M\xrightarrow{\ell:(!G,\gamma)}N$, then $\gamma=A(\ell)$.
    \end{lemma}
\ifarxiversion
    \begin{proof}
    The result follows by induction on $\xrightarrow{\ell:(!G,\gamma)}$. For the base case, we consider $M=[m]_{\ell}$, $N=[n]_\ell$ and $m\xrightarrow{(!G,\gamma)} n$; by \Cref{lemma:singlesender}, we derive that $\gamma=A(\ell)$. The inductive step follows immediately from the induction hypothesis.
    \end{proof}
\fi
    We now state a few other basic properties:     
    \Cref{lemma:itsallrelative} uses a `commutativity' property for communicating monitors (\Cref{lemma:commutecommunication}), and this is based on five similar results for local monitors (\Cref{lemma:sendreceiveorder}-\Cref{lemma:uniquenesstwo}), that are in addition used in a variety of other lemmas.

    The following lemma shows that, if action-derived local monitors can both send and receive a message, the order in which this happens does not matter. 

\begin{lemma}\label{lemma:sendreceiveorder}
Let $m$ be a relevant local monitor such that
$m\xrightarrow{(!G,\gamma)}m_1$ and $m\xrightarrow{(?\ell',\gamma')} m_2$, for some $G \subseteq \Loc$ and $\ell'\in \Loc$. Then, $m_1\xrightarrow{(?\ell',\gamma')} m_3$ and $m_2\xrightarrow{(!G,\gamma)} m_3$, for some $m_3$.
\end{lemma}
\begin{proof}
We proceed by induction on $m$. 
In the base case for $m=\synone{\psi}$, we conclude by \Cref{lemma:nocommunicate}; in the base cases for
$m=(!G,a).\synone{\psi}$ and $m=\sum_{a\in\Act}(?\{\ell''\},a).\synone{\psi_a}$, the premise of the lemma is not satisfied. 

For the inductive step, let $m=n_1\pandor n_2$, with $n_1$ and $n_2$ relevant monitors and $\pandor \in \{\pand,\por\}$. 
From $m\xrightarrow{(!G,\gamma)} m_1$ we derive w.l.o.g. that $m_1=n'_1\pandor n_2$ and $n_1\xrightarrow{(!G,\gamma)} n'_1$. From $m\xrightarrow{(?\ell',\gamma')} m_2$  we derive three cases.
\begin{enumerate}
    \item $m_2=m'_2\pandor m''_2$ such that $n_1 \xrightarrow{(?\ell',\gamma')} m'_2$ and $n_2 \xrightarrow{(?\ell',\gamma')} m''_2$.
    From the induction hypothesis, we derive that $n'_1\xrightarrow{(?\ell',\gamma')} m'_3$ and $m'_2\xrightarrow{(!G,\gamma)} m'_3$, for some $m'_3$. Thus, $m_1 = n'_1\pandor n_2\xrightarrow{(?\ell',\gamma')} m'_3\pandor m''_2$ and $m_2 = m'_2\pandor m''_2\xrightarrow{(!G,\gamma)} m'_3\pandor m''_2$, as desired.

    \item $m_2=m'_2\pandor n_2$ such that  $n_1 \xrightarrow{(?\ell',\gamma')} m'_2$ and  $n_2 \NOT {\xrightarrow{(?\ell',\gamma')}} $. From the induction hypothesis, we derive $n'_1\xrightarrow{(?\ell',\gamma')} m'_3$ and $m'_2\xrightarrow{(!G,\gamma)} m'_3$, for some $m'_3$. Thus, $m_1 = n'_1\pandor n_2\xrightarrow{(?\ell',\gamma')} m'_3\pandor n_2$ and $m_2 = m'_2\pandor n_2\xrightarrow{(!G,\gamma)} m'_3\pandor n_2$.
    
    \item $m_2=n_1\pandor m'_2$ such that  $n_1 \NOT {\xrightarrow{(?\ell',\gamma')}}$ and  $n_2 \xrightarrow{(?\ell',\gamma')} m'_2 $. We obtain that $m_2 = n_1\pandor m'_2\xrightarrow{(!G,\gamma)} n'_1\pandor m'_2 $ and $m_1 = n'_1\pandor n_2\xrightarrow{(?\ell',\gamma')} n'_1\pandor m'_2$, where we use \Cref{lemma:singlesender} to conclude that $n'_1 \NOT {\xrightarrow{(?\ell',\gamma')}}$ from $n_1 \NOT {\xrightarrow{(?\ell',\gamma')}}$. 
\qedhere
\end{enumerate}
\end{proof}

The next lemma proves uniqueness of reached states after a receiving transition.
\begin{lemma}\label{lemma:uniqueness}
    Let $m$ be a relevant local monitor such that
 $m\xrightarrow{(?\ell', \gamma)}n_1$ and $m\xrightarrow{(?\ell', \gamma)}n_2$. Then, $n_1=n_2$.
\end{lemma}
\begin{proof}
    We proceed by induction on $m$. In the base case for $m=\synone{\psi}$, we conclude by \Cref{lemma:nocommunicate}. For $m=(!G,a).\synone{\psi}$, the premise of the lemma is not satisfied; for $m=\sum_{a\in\Act}(?\{\ell''\},a).\synone{\psi_a}$, we can conclude because there is only one entry in the sum for each combination of ${\ell''}$ and $a$.
    
    For the inductive step, let $m=m_1\pandor m_2$, with $m_1$ and $m_2$ relevant monitors. From $m\xrightarrow{(?\ell',\gamma)} n_1$ we derive two cases.
    \begin{enumerate}
        \item $n_1=n'_1\pandor n''_1$ such that $m_1 \xrightarrow{(?\ell',\gamma)} n'_1$ and $m_2 \xrightarrow{(?\ell',\gamma)} n''_1$. From $m\xrightarrow{(?\ell',\gamma)} n_2$ we derive two more cases.
        \begin{enumerate}
            \item $n_2=n'_2\pand n''_2$ such that $m_1 \xrightarrow{(?\ell',\gamma)} n'_2$ and $m_2 \xrightarrow{(?\ell',\gamma)} n''_2$. From the induction hypothesis, we conclude that $n'_1=n'_2$ and $n''_1=n''_2$, and thus that $n_1=n_2$.
            \item $n_2=n'_2\pandor m_2$ such that  $m_1 \xrightarrow{(?\ell',\gamma)} n'_2$ and  $m_2 \NOT {\xrightarrow{(?\ell',\gamma)}} $. This is a contradiction with $m_2 \xrightarrow{(?\ell',\gamma)} n''_1$.
        \end{enumerate}
        \item $n_1=n'_1\pandor m_2$ such that $m_1 \xrightarrow{(?\ell',\gamma)} n'_1$ and  $m_2 \NOT {\xrightarrow{(?\ell',\gamma)}} $. We then conclude from $m\xrightarrow{(?\ell',\gamma)
        } n_2$ that $n_2=n'_2\pandor m_2$ such that $m_1 \xrightarrow{(?\ell',\gamma)} n'_2$. From the induction hypothesis, we conclude that $n'_1=n'_2$ and thus that $n_1 =n_2$. 
        \qedhere 
            \end{enumerate}
\end{proof}

The next lemma states that, if a monitor can receive a message from two different senders, the order in which this happens does not matter.
\begin{lemma}\label{lemma:commutativereceiving}
    Let $m$ be a relevant local monitor such that
    $m\xrightarrow{(?\ell_1,\gamma_1)}m_1$ and $m\xrightarrow{(?\ell_2,\gamma_2)} m_2$, where $\ell_1 \neq \ell_2$.Then, 
    $m_1\xrightarrow{(?\ell_2,\gamma_2)} m_3$ and $m_2\xrightarrow{(?\ell_1,\gamma_1)} m_3$, for some $m_3$.
\end{lemma}
\begin{proof}
    We proceed by induction on $m$. 
    In the base case, for $m=\synone{\psi}$, we conclude by \Cref{lemma:nocommunicate};
    for $m=(!G,a).\synone{\psi}$ and $m=\sum_{a\in\Act}(?\{\ell'\},a).\synone{\psi_a}$, the premise of the lemma is not satisfied. 
    For the inductive step, let $m=n_1\pandor n_2$, with $n_1$ and $n_2$ relevant monitors. From $m\xrightarrow{(?\ell_1,\gamma_1)} m_1$ and $m\xrightarrow{(?\ell_2,\gamma_2)} m_2$, wlog we derive the following four cases.

    \begin{enumerate}
        \item $m_1=n'_1\pandor n'_2$ and $m_2=n''_1\pandor n''_2$, where 
        \begin{align}
            n_1 &\xrightarrow{(?\ell_1,\gamma_1)} n'_1 \label{eq:yz1}
            \\
            n_2 &\xrightarrow{(?\ell_1,\gamma_1)} n'_2 \label{eq:yz2}
            \\
            n_1 &\xrightarrow{(?\ell_2,\gamma_2)} n''_1, \text{ and} \label{eq:yw1}
            \\
            n_2 &\xrightarrow{(?\ell_2,\gamma_2)} n''_2. \label{eq:yw2}
        \end{align}
        From the induction hypothesis on \Cref{eq:yz1,eq:yw1} and on \Cref{eq:yz2,eq:yw2},
        $n'_1\xrightarrow{(?\ell_2,\gamma_2)} r_1$ and $n''_1\xrightarrow{(?\ell_{1},\gamma_1)} r_1$, for some $r_1$,
        and 
        $n'_2\xrightarrow{(?\ell_2,\gamma_2)} r_2$ and $n''_2\xrightarrow{(?\ell_1,\gamma_1)} r_2$, for some $r_2$.
        We conclude that $m_1 \xrightarrow{(?\ell_2,\gamma_2)} r_1 \pandor r_2$ and $m_2 \xrightarrow{(?\ell_1,\gamma_1)} r_1 \pandor r_2$.

        \item $m_1=n'_1\pandor n'_2$ and $m_2=n''_1\pandor n_2$, where 
            \Cref{eq:yz1,eq:yz2,eq:yw1} hold, and
        \begin{align}
            n_2 &\NOT{\xrightarrow{(?\ell_2,\gamma_2)}} .
            \label{eq:ynot2}
        \end{align}
        From the induction hypothesis on \Cref{eq:yz1,eq:yw1},  we get that 
        \begin{align}
        n'_1&\xrightarrow{(?\ell_2,\gamma_2)} r_1 \text{ and} 
        \label{eq:zr1}
        \\
        n''_1&\xrightarrow{(?\ell_{1},\gamma_1)} r_1, 
        \label{eq:wr1}
        \end{align}
        for some $r_1$.
        \Cref{lemma:singlesender} and \Cref{eq:ynot2} yield that 
        $n'_2 \NOT{\xrightarrow{(?\ell_2,\gamma_2)}}$; therefore, from \Cref{eq:zr1}, we get that 
        $m_1 = n'_1 \pandor n'_2 \xrightarrow{(?\ell_2,\gamma_2)} r_1 \pandor z_2$; finally, 
        \Cref{eq:wr1,eq:yz2} yield 
        $m_2=n''_1\pandor n_2 \xrightarrow{(?\ell_1,\gamma_1)} r_1 \pandor n_2$.

        \item $m_1=n'_1\pandor n_2$ and $m_2=n''_1\pandor n_2$, where 
        \Cref{eq:yz1,eq:yw1,eq:ynot2} hold, and
    \begin{align}
        n_2 &\NOT{\xrightarrow{(?\ell_1,\gamma_1)}} .
        \label{eq:ynot21}
    \end{align}
    From the induction hypothesis on \Cref{eq:yz1,eq:yw1},  we get 
     \Cref{eq:zr1,eq:wr1},
    for some $r_1$; this, together with \Cref{eq:ynot2,eq:ynot21},
    yields
    $m_1 = n'_1 \pandor n_2 \xrightarrow{(?\ell_2,\gamma_2)} r_1 \pandor n_2$ and 
    $m_2=n''_1\pandor n_2 \xrightarrow{(?\ell_1,\gamma_1)} r_1 \pandor n_2$.

    \item $m_1=n'_1\pandor n_2$ and $m_2=n_1\pandor n''_2$, where 
        \Cref{eq:yz1,eq:yw2,eq:ynot21} hold, and
    \begin{align}
        n_1 &\NOT{\xrightarrow{(?\ell_2,\gamma_2)}} .
        \label{eq:ynot12}
    \end{align}
    \Cref{lemma:singlesender}, with \Cref{eq:ynot12,eq:yz1} and with \Cref{eq:ynot21,eq:yw2}, yields that 
    $n'_1 \NOT{\xrightarrow{(?\ell_2,\gamma_2)}}$ and $n''_2 \NOT{\xrightarrow{(?\ell_1,\gamma_1)}}$, respectively; therefore, from \Cref{eq:yz1,eq:yw2}, we get that 
    $m_1 = n'_1 \pandor n_2 \xrightarrow{(?\ell_2,\gamma_2)} n'_1 \pandor n''_2$ and 
    $m_2=n_1\pandor n''_2 \xrightarrow{(?\ell_1,\gamma_1)} n'_1 \pandor n''_2$.
    \qedhere
        \end{enumerate}
\end{proof}

The following lemma shows that if action-derived local monitors can send two different messages, the order in which this happens does not matter.
\begin{lemma}\label{lemma:sendreceiveordertwo}
Let $m$ be a relevant local monitor such that
$m\xrightarrow{(!G,\gamma)}m_1$ and $m\xrightarrow{(!G',\gamma')} m_2$, for some $G,G'\subseteq \Loc$ and $\gamma,\gamma' \in \Cvar$ such that either $G\neq G'$ or $\gamma\neq\gamma'$. Then, $m_1\xrightarrow{(!G',\gamma')} m_3$ and $m_2\xrightarrow{(!G,\gamma)} m_3$, for some $m_3$.
\end{lemma}
\begin{proof}
We proceed by induction on $m$. The base case for $m=\synone{\psi}$ follows from \Cref{lemma:nocommunicate}; for $m=(!G,a).\synone{\psi}$ and $m=\sum_{a\in\Act}(?\{\ell'\},a).\synone{\psi_a}$, the premise of the lemma is not satisfied.  
    
For the inductive step, let $m=n_1\pandor n_2$, with $n_1$ and $n_2$ relevant monitors. From $m\xrightarrow{(!G,\gamma)} m_1$ we derive w.l.o.g. that $m_1=n'_1\pandor n_2$ and $n_1\xrightarrow{(!G,\gamma)} n'_1$. From $m\xrightarrow{(!G',\gamma')} m_2$ we derive two cases.
\begin{enumerate}
    \item $m_2=n''_1\pandor n_2$ such that $n_1 \xrightarrow{(!G',\gamma')} n''_1$.
    From the induction hypothesis, we derive that $n'_1\xrightarrow{(!G',\gamma')} n'$ and $n''_1\xrightarrow{(!G,\gamma)} n'$, for some $n'$. Thus, $n'_1\pandor n_2\xrightarrow{(!G',\gamma')} n'\pandor n_2$ and $n''_1\pandor n_2\xrightarrow{(!G,\gamma)} n' \pandor n_2$, as desired.
    
    \item $m_2=n_1\pandor n''_2$ such that $n_2 \xrightarrow{(!G',\gamma')} n''_2 $. Then, $n_1\pandor n''_2\xrightarrow{(!G,\gamma)} n'_1\pandor n''_2 $ and $n'_1\pandor n_2\xrightarrow{(!G',\gamma')} n'_1\pandor n''_2$. 
    \qedhere
\end{enumerate}
\end{proof}

The following lemma is the version of \Cref{lemma:uniqueness} for sending.
\begin{lemma}\label{lemma:uniquenesstwo}
    Let $m$ be a relevant local monitor such that
 $m\xrightarrow{(!G, \gamma)}n_1$ and $m\xrightarrow{(!G, \gamma)}n_2$. Then, either $n_1=n_2$, or $n_1\xrightarrow{(!G,\gamma)}m'$ and $n_2\xrightarrow{(!G,\gamma)}m'$, for some $m'$.
\end{lemma}
\ifarxiversion
\begin{proof}
    We proceed by induction on $m$. If $m=(!G,a).\synone{\psi}$, it is immediate that $n_1=n_2$. If $m=\synone{\psi}$ or $m=\sum_{a\in\Act}(?\{\ell'\},a).\synone{\psi_a}$, the premise of the lemma is not satisfied. 
    
    Let $m=m_1\pandor m_2$, with $m_1$ and $m_2$ relevant monitors. From $m\xrightarrow{(!G,\gamma)} n_1$, we derive that $n_1=m'_1\pandor m'_2$ such that $m_1 \xrightarrow{(!G,\gamma)} m'_1$.
  From $m\xrightarrow{(!G,\gamma)} n_2$ we derive two cases.
        \begin{enumerate}
            \item $n_2=m''_1\pandor m_2$ such that $m_1 \xrightarrow{(!G,\gamma)} m''_1$. From the induction hypothesis, we conclude that $m'_1=m''_1$ (and thus that $n_1=n_2$), or $m''_1 \xrightarrow{(!G,\gamma)} r$ and $m'_1 \xrightarrow{(!G,\gamma)} r$, for some $r$. In the latter case, we derive that $n_1=m'_1\pandor m_2 \xrightarrow{(!G,\gamma)} r\pandor m_2$ and $n_2=m''_1\pandor m_2 \xrightarrow{(!G,\gamma)} r\pandor m_2$.
            \item $n_2=m_1\pandor m''_2$ such that $m_2 \xrightarrow{(!G,\gamma)} m''_2$. In this case, we derive that $n_1=m'_1\pandor m_2 \xrightarrow{(!G,\gamma)} m'_1\pandor m''_2$ and $n_2=m_1\pandor m''_2 \xrightarrow{(!G,\gamma)} m'_1\pandor m''_2$.
            \qedhere 
        \end{enumerate}
\end{proof}
\fi

Before we continue on the path to proving \Cref{lemma:itsallrelative}, we show how to formally capture the state of a decentralized monitor after taking a receiving transition, for which we only need \Cref{lemma:uniqueness}. To this aim, we use $M[[n]_\ell/[m]_\ell]$ as notation for replacing $[m]_\ell$ in $M$ with $[n]_\ell$ (the operation does not change $M$ in case $[m]_\ell$ does not occur in $M$).
\begin{lemma}\label{lemma:group}
    Let $M$ and $N$ be action-derived monitors such that $M\extoverset{$G:(?\ell,\gamma)$}\rightsquigarrow N$; let $\{[m_i]_{\ell_i} \mid 1 \leq i \leq k \}$ be the set of all the $[m]_{\ell} \in M$. 
    For each $1 \leq i \leq k$, let $n_i$ be $m_i$, when $\ell_i \notin G$ or $m_i \NOT{\xrightarrow{(?\ell,\gamma)}}$, and be such that $m_i \xrightarrow{(?\ell,\gamma)} n_i$, otherwise.
    Then, $M[[n_1]_{\ell_1}/[m_1]_{\ell_1}]\cdots [[n_k]_{\ell_k}/[m_k]_{\ell_k}] = N$.
\end{lemma}
\ifarxiversion
\begin{proof}
    \Cref{lemma:uniqueness} yields that each $n_i$ is uniquely defined.
Then, the lemma follows by an easy induction on $\extoverset{$G:(?\ell,\gamma)$}\rightsquigarrow$.
\end{proof}
\fi


 \Cref{lemma:receiveonce}-\Cref{lemma:commutereceivesend} lift the previous results from local to decentralized monitors and are used to prove \Cref{lemma:commutecommunication}.
This result, together with \Cref{lemma:relativestep}, leads to \Cref{lemma:itsallrelative}, which captures that the order of communication does not matter: regardless of the order in which the communication steps are executed, the ‘final' state (i.e., the state where no more communication can take place) is unique.

    \begin{lemma}\label{lemma:receiveonce}
        Let $M$ be action-derived. If $M  \extoverset{$G:(?\ell,\gamma)$}  \rightsquigarrow N$, then $N  \extoverset{$G:(?\ell,\gamma)$}  \rightsquigarrow N$.
    \end{lemma}
    \begin{proof}
    We proceed by induction on $M  \extoverset{$G:(?\ell,\gamma)$}  \rightsquigarrow N$. In the first base case we consider $M=[m]_{\ell'}$, $N=[n]_{\ell'}$, $\ell'\in G$, and $m\xrightarrow{(?\ell,\gamma)}n$: by \Cref{lemma:singlesender}, which we can apply because $M$ is action-derived, we obtain that $n\NOT{\xrightarrow{(?\ell, \gamma)}}$, and so $[n]_{\ell'}\extoverset{$G:(?\ell,\gamma)$}  \rightsquigarrow [n]_{\ell'}$.
    In the second and third base case we consider $M=[m]_{\ell'}=N$, and thus we are done immediately.
    In the inductive case we consider $M=M_1\op M_2$, $N=N_1\op N_2$, $M_1  \extoverset{$G:(?\ell,\gamma)$}  \rightsquigarrow N_1$ and $M_2  \extoverset{$G:(?\ell,\gamma)$}  \rightsquigarrow N_2$. From the induction hypothesis we get that $N_1  \extoverset{$G:(?\ell,\gamma)$}  \rightsquigarrow N_1$ and $N_2  \extoverset{$G:(?\ell,\gamma)$}  \rightsquigarrow N_2$, and thus that $N=N_1\op N_2  \extoverset{$G:(?\ell,\gamma)$}  \rightsquigarrow N_1\op N_2$.
    \end{proof}
    
    \begin{lemma}\label{lemma:commutereceive}
        Let $M$ be $A$-derived and such that $M  \extoverset{$G_1:(?\ell_1,A(\ell_1))$}  \rightsquigarrow N_1$ and $M  \extoverset{$G_2:(?\ell_2,A(\ell_2))$}  \rightsquigarrow N_2$. Then, either $N_1=N_2$, or $N_1  \extoverset{$G_2:(?\ell_2\,A(\ell_2))$}  \rightsquigarrow N_3$ and $N_2  \extoverset{$G_1:(?\ell_1,A(\ell_1))$}  \rightsquigarrow N_3$, for some $N_3$. 
    \end{lemma}
    \begin{proof}
        We proceed by induction on $M  \extoverset{$G_1:(?\ell_1,A(\ell_1))$}  \rightsquigarrow N_1$. In the first base case we consider $M=[m]_\ell$, $N_1=[n_1]_\ell$, $\ell\in G_1$, and $m\xrightarrow{(?\ell_1,A(\ell_1))}n_1$. We distinguish two cases:
        \begin{enumerate}
            \item $N_2=[n_2]_\ell$, $m\xrightarrow{(?\ell_2, A(\ell_2))}n_2$, $\ell\in G_2$. Since $M$ is action-derived, we can use \Cref{lemma:commutativereceiving} to obtain two further cases:
            \begin{enumerate}
                \item $\ell_1\neq\ell_2$ and $n_1\xrightarrow{(?\ell_2,A(\ell_2))}q$ and $n_2\xrightarrow{(?\ell_1,A(\ell_1))} q$. Hence, $N_1 =[n_1]_\ell \extoverset{$G_2:(?\ell_2,A(\ell_2))$}  \rightsquigarrow [q]_\ell$ and $N_2 =[n_2]_\ell \extoverset{$G_1:(?\ell_1,A(\ell_1))$}  \rightsquigarrow [q]_\ell$.
                
                \item $\ell_1=\ell_2$. Via \Cref{lemma:uniqueness} we conclude that $n_1=n_2$ and thus that $N_1=N_2$. 
            \end{enumerate}
            
            \item $N_2=[m]_\ell$ and either $\ell\notin G_2$ or $m\NOT{\xrightarrow{(?\ell_2,A(\ell_2))}}$. In the latter case, we know by \Cref{lemma:singlesender}(2) that $n_1\NOT{\xrightarrow{(?\ell_2,A(\ell_2))}}$. Hence, in both cases, $N_1 =[n_1]_\ell \extoverset{$G_2:(?\ell_2,A(\ell_2))$}  \rightsquigarrow [n_1]_\ell=N_1$. Since $N_2=M$, we know that $N_2 \extoverset{$G_1:(?\ell_1,A(\ell_1))$}  \rightsquigarrow N_1$, as desired.
        \end{enumerate}
        For the next two base cases, we consider $N_1=[m]_\ell$ and either $\ell\notin G_1$ or $m\NOT{\xrightarrow{(?\ell_1,A(\ell_1))}}$. We distinguish two cases:
        \begin{enumerate}
            \item $N_2=[n_2]_\ell$, $m\xrightarrow{(?\ell_2, A(\ell_2))}n_2$, $\ell\in G_2$. If $m\NOT{\xrightarrow{(?\ell_1,A(\ell_1))}}$, \Cref{lemma:singlesender}(2) entails that $n_2\NOT{\xrightarrow{(?\ell_1,A(\ell_1))}}$. Hence, in all cases, $N_2 =[n_2]_\ell \extoverset{$G_1:(?\ell_1,A(\ell_1))$}  \rightsquigarrow [n_2]_\ell=N_2$. Since $N_1=M$, we know that $N_1 \extoverset{$G_2:(?\ell_2,A(\ell_2))$}  \rightsquigarrow N_2$, as desired.
            
            \item $N_2=[m]_\ell$ and either $\ell\notin G_2$ or $m\NOT{\xrightarrow{(?\ell_2,A(\ell_2))}}$. Hence, $M=N_2=N_1$, and we are done immediately.
        \end{enumerate}

    For the inductive step, we let $M=M_1\op M_2$, $N_1=O_1\op O_2$, $M_1  \extoverset{$G_1:(?\ell_1,A(\ell_1))$}  \rightsquigarrow O_1$ and $M_2  \extoverset{$G_1:(?\ell_1,A(\ell_1))$}  \rightsquigarrow O_2$. This implies that $N_2=O_3\op O_4$, with $M_1  \extoverset{$G_2:(?\ell_2,A(\ell_2))$}  \rightsquigarrow O_3$ and $M_2  \extoverset{$G_2:(?\ell_2,A(\ell_2))$}  \rightsquigarrow O_4$. We apply the induction hypothesis to $M_1  \extoverset{$G_1:(?\ell_1,A(\ell_1))$}  \rightsquigarrow O_1$ to obtain two cases:
    \begin{enumerate}
    \item $O_1  \extoverset{$G_2:(?\ell_2,A(\ell_2))$}  \rightsquigarrow Q_1$ and $O_3  \extoverset{$G_1:(?\ell_1,A(\ell_1))$}  \rightsquigarrow Q_1$. We apply the induction hypothesis to $M_2  \extoverset{$G_1:(?\ell_1,A(\ell_1))$}  \rightsquigarrow O_2$ to obtain two more cases:
    \begin{enumerate}
    \item $O_2  \extoverset{$G_2:(?\ell_2,A(\ell_2))$}  \rightsquigarrow Q_2$ and $O_4  \extoverset{$G_1:(?\ell_1,A(\ell_1))$}  \rightsquigarrow Q_2$.
    Hence, $O_1\op O_2  \extoverset{$G_2:(?\ell_2,A(\ell_2))$}  \rightsquigarrow Q_1\op Q_2$ and $O_3\op O_4  \extoverset{$G_1:(?\ell_1,A(\ell_1))$}  \rightsquigarrow Q_1\op Q_2$.
    \item $O_2=O_4$. We apply \Cref{lemma:receiveonce} to $M_2 \extoverset{$G_2:(?\ell_2,A(\ell_2))$}\rightsquigarrow O_4 $ to conclude that $O_4 \extoverset{$G_2:(?\ell_2,A(\ell_2))$}\rightsquigarrow O_4 $. 
    Being $O_2=O_4$, we obtain that $O_1\op O_2 = O_1\op O_4 \extoverset{$G_2:(?\ell_2,A(\ell_2))$}  \rightsquigarrow Q_1\op O_4$ and $O_3\op O_4  \extoverset{$G_1:(?\ell_1,A(\ell_1))$}  \rightsquigarrow Q_1\op O_4$.
    \end{enumerate}
    \item $O_1=O_3$. We apply the induction hypothesis to $M_2  \extoverset{$G_1:(?\ell_1,A(\ell_1))$}  \rightsquigarrow O_2$ to obtain two more cases:
    \begin{enumerate}
    \item $O_2  \extoverset{$G_2:(?\ell_2,A(\ell_2))$}  \rightsquigarrow Q_2$ and $O_4  \extoverset{$G_1:(?\ell_1,A(\ell_1))$}  \rightsquigarrow Q_2$.
    We apply \Cref{lemma:receiveonce} to $M_1 \extoverset{$G_2:(?\ell_2,A(\ell_2))$}\rightsquigarrow O_3 $ to conclude that $O_3 \extoverset{$G_2:(?\ell_2,A(\ell_2))$}\rightsquigarrow O_3 $. 
Being $O_1=O_3$, we have that $O_1\op O_2 = O_3\op O_2 \extoverset{$G_2:(?\ell_2,A(\ell_2))$}  \rightsquigarrow O_3\op Q_2$ and $O_3\op O_4  \extoverset{$G_1:(?\ell_1,A(\ell_1))$}  \rightsquigarrow O_3\op Q_2$.
    \item $O_2=O_4$. Then we immediately have $N_1=N_2$.\qedhere
    \end{enumerate}
    \end{enumerate}
    \end{proof}
    
    \begin{lemma}\label{lemma:uniquereceive}
        Let $M$ be action-derived and such that $M  \extoverset{$G:(?\ell,\gamma)$}  \rightsquigarrow N_1$ and $M  \extoverset{$G:(?\ell,\gamma)$}  \rightsquigarrow N_2$. Then, $N_1=N_2$.
    \end{lemma}
    \ifarxiversion
    \begin{proof}
        We proceed by induction on $M  \extoverset{$G:(?\ell,\gamma)$}  \rightsquigarrow N_1$. In the first base case we consider $M=[m]_{\ell'}$, $N_1=[n_1]_{\ell'}$, $\ell'\in G$, and $m\xrightarrow{(?\ell,\gamma)}n_1$. We distinguish two cases:
        \begin{enumerate}
            \item $N_2=[n_2]_{\ell'}$, $m\xrightarrow{(?\ell, \gamma)}n_2$. We apply \Cref{lemma:uniqueness} ($M$ is action-derived) to conclude that $n_2=n_1$, and thus $N_1=N_2$. 
            
            \item $N_2=[m]_{\ell'}$ and either $m\NOT{\xrightarrow{(?\ell,\gamma)}}$ or $\ell'\notin G$. This would contradict  $m\xrightarrow{(?\ell,\gamma)}n_1$ and $\ell' \in G$, thus we can exclude this case.
        \end{enumerate}
        In the second and third base case we consider $M=[m]_{\ell'}=N_1$, and either $\ell'\notin G$ or $m\NOT{\xrightarrow{(?\ell,\gamma)}}$. We distinguish two cases:
        \begin{enumerate}
            \item $N_2=[n_2]_{\ell'}$, $m\xrightarrow{(?\ell, \gamma)}n_2$, $\ell'\in G$. This would contradict  $m\NOT{\xrightarrow{(?\ell,\gamma)}}$ or $\ell' \notin G$. Hence we can exclude this case.
            
            \item $N_2=[m]_{\ell'}$ and either $\ell'\notin G$ or $m\NOT{\xrightarrow{(?\ell,\gamma)}}$. Hence, $M=N_2=N_1$, and we are done immediately.
        \end{enumerate}
        
For the inductive case, let $M=M_1\op M_2$, $N_1=O_1\op O_2$, $M_1  \extoverset{$G:(?\ell,\gamma)$}  \rightsquigarrow O_1$ and $M_2  \extoverset{$G:(?\ell,\gamma)$}  \rightsquigarrow O_2$. This implies that $N_2=O_3\op O_4$, with $M_1  \extoverset{$G:(?\ell,\gamma)$}  \rightsquigarrow O_3$ and $M_2  \extoverset{$G:(?\ell,\gamma)$}  \rightsquigarrow O_4$. We apply the induction hypothesis twice to obtain that $O_1=O_3$ and $O_2=O_4$, and thus $N_1=N_2$.
        \end{proof}
    \fi
    
        \begin{lemma}\label{lemma:commutereceivesend}
            Let $M$ be action-derived and such that $M  \xrightarrow{\ell_1:(!G_1,\gamma_1)} N_1$ and 
            $M  \extoverset{$G_2:(?\ell_2,A(\ell_2))$}  \rightsquigarrow N_2$. Then, $N_1  \extoverset{$G_2:(?\ell_2,A(\ell_2))$}  \rightsquigarrow N_3$ and $N_2  \xrightarrow{\ell_1:(!G_1,\gamma_1)} N_3$, for some $N_3$.
        \end{lemma}
        \begin{proof}
        We proceed by induction on $M  \xrightarrow{\ell_1:(!G_1,\gamma_1)} N_1$. In the base case we consider $M=[m]_{\ell_1}$, $N_1=[n_1]_{\ell_1}$ and $m\xrightarrow{(!G_1,\gamma_1)} n_1$. We distinguish two cases based on $M  \extoverset{$G_2:(?\ell_2,A(\ell_2))$}  \rightsquigarrow N_2$:
        \begin{enumerate}
            \item $N_2=[n_2]_{\ell_1}$, $m\xrightarrow{(?\ell_2, A(\ell_2))}n_2$, and $\ell_1\in G_2$. Hence, we can apply \Cref{lemma:sendreceiveorder} to obtain that there exists $m_1$ such that $n_1\xrightarrow{(?\ell_2,A(\ell_2))}m_1$ and $n_2\xrightarrow{(!G_1,\gamma_1)} m_1$. Hence, $N_2\xrightarrow{\ell_1:(!G_1,\gamma_1)} [m_1]_{\ell_1}$ and $N_1  \extoverset{$G_2:(?\ell_2,A(\ell_2))$}  \rightsquigarrow [m_1]_{\ell_1}$.
            
            \item $N_2=[m]_{\ell_1}$ and either $\ell_1\notin G_2$ or $m\NOT{\xrightarrow{(?\ell_2,A(\ell_2))}}$. 
 The latter case implies that $n_1\NOT{\xrightarrow{(?\ell_2,A(\ell_2))}}$, by applying \Cref{lemma:singlesender} to $m\xrightarrow{(!G_1,\gamma_1)} n_1$. As $M=N_2$, in both cases we have $N_2\xrightarrow{\ell_1:(!G_1,\gamma_1)} N_1$ and $N_1  \extoverset{$G_2:(?\ell_2,A(\ell_2))$}  \rightsquigarrow N_1$, as desired.
        \end{enumerate}
        
        For the inductive case, we consider $M=M_1\op M_2$, $N_1=O_1\op O_2$, 
        $M_1\xrightarrow{\ell_1:(!G_1,\gamma_1)}O_1$ and $M_2 \extoverset{$G_1:(?\ell_1,\gamma_1)$}  \rightsquigarrow O_2$. We deduce that $N_2=O_3\op O_4$, with $M_1  \extoverset{$G_2:(?\ell_2,A(\ell_2))$}  \rightsquigarrow O_3$ and $M_2  \extoverset{$G_2:(?\ell_2,A(\ell_2))$}  \rightsquigarrow O_4$. We use the induction hypothesis to conclude that there exists $Q$ such that $O_1  \extoverset{$G_2:(?\ell_2,A(\ell_2))$}  \rightsquigarrow Q$ and $O_3  \xrightarrow{\ell_1:(!G_1,\gamma_1)} Q$. Since $M\xrightarrow{\ell_1:(!G_1,\gamma_1)} N_1$, we know that $\gamma_1=A(\ell_1)$ (\Cref{lemma:sendright}). Hence, we can apply \Cref{lemma:commutereceive} on $M_2 \extoverset{$G_1:(?\ell_1,\gamma_1)$}  \rightsquigarrow O_2$ and $M_2  \extoverset{$G_2:(?\ell_2,A(\ell_2))$}  \rightsquigarrow O_4$ and have two cases:
        \begin{enumerate}
            \item $O_2  \extoverset{$G_2:(?\ell_2,A(\ell_2))$}  \rightsquigarrow Q'$ and $O_4 \extoverset{$G_1:(?\ell_1,\gamma_1)$}  \rightsquigarrow Q'$. Hence, $N_1=O_1\op O_2\extoverset{$G_2:(?\ell_2,A(\ell_2))$}  \rightsquigarrow Q\op Q'$ and $N_2=O_3\op O_4  \xrightarrow{\ell_1:(!G_1,\gamma_1)} Q\op Q'$. 
            
            \item $O_2=O_4$. We apply \Cref{lemma:receiveonce} to $M_2 \extoverset{$G_1:(?\ell_1,\gamma_1)$}  \rightsquigarrow O_2$ 
to conclude that $O_2 \extoverset{$G_1:(?\ell_1,\gamma_1)$}  \rightsquigarrow O_2$.
Being $O_2=O_4$, we have
$N_1=O_1\op O_2=O_1\op O_4\extoverset{$G_2:(?\ell_2,A(\ell_2))$}  \rightsquigarrow Q\op O_4$ and $N_2=O_3\op O_4 \xrightarrow{\ell_1:(!G_1,\gamma_1)} Q\op O_4$.
\qedhere
        \end{enumerate}
        \end{proof}

    \begin{lemma}\label{lemma:commutecommunication}
        Let $M$ be action-derived and such that $M\xrightarrow{c_1}M_1$ and $M\xrightarrow{c_2}M_2$. Then, either 
        $c_1=c_2$ and $M_1=M_2$, or $M_1\xrightarrow{c_2}N$ and $M_2\xrightarrow{c_1}N$, for some $N$.
    \end{lemma}
    \begin{proof}
    We proceed by induction on $M\xrightarrow{c_1} M_1$. Let $c_1=\ell_1:(!G_1,\gamma_1)$ and $c_2=\ell_2:(!G_2,\gamma_2)$.
In the base case we consider $M=[m]_{\ell_1}$, $M_1=[m']_{\ell_1}$ and $m\xrightarrow{(!G_1,\gamma_1)}m'$. Thus we know that $M_2=[m'']_{\ell_1}$, $\ell_1=\ell_2$ and $m\xrightarrow{(!G_2,\gamma_2)}m''$. We distinguish two cases:
    \begin{enumerate}
    \item $G_1\neq G_2$ or $\gamma_1\neq \gamma_2$. Then we apply \Cref{lemma:sendreceiveordertwo} to obtain that $m'\xrightarrow{(!G_2,\gamma_2)}n$ and $m'' \xrightarrow{(!G_1,\gamma_1)}n$. Hence, $M_1\xrightarrow{c_2}[n]_{\ell_1}$ and $M_2\xrightarrow{c_1}[n]_{\ell_1}$.
    \item $G_1=G_2$ and $\gamma_1=\gamma_2$. We apply \Cref{lemma:uniquenesstwo} to obtain that $m'\xrightarrow{(!G_2,\gamma_2)}n$ and $m'' \xrightarrow{(!G_1,\gamma_1)}n$, or that $m'=m''$. In the former case, we conclude like in the first item; in the latter case, we get immediately that $M_1=M_2$ and $c_1=c_2$.
    \end{enumerate}
    
    For the inductive case, let $M=N_1\op N_2$, $M_1=N_1'\op N_2'$, 
    $N_1\xrightarrow{c_1}N_1'$ and $N_2 \extoverset{$G_1:(?\ell_1,\gamma_1)$}  \rightsquigarrow N_2'$. We distinguish two cases.
    \begin{enumerate}
    \item $M_2=O_1\op O_2$, with
    $N_1\xrightarrow{c_2}O_1$ and $N_2 \extoverset{$G_2:(?\ell_2,\gamma_2)$}  \rightsquigarrow O_2$. We obtain two further cases from the induction hypothesis applied to $N_1$:
    \begin{enumerate}
    \item $N_1'\xrightarrow{c_2} Q $ and $O_1\xrightarrow{c_1} Q$. Since $M\xrightarrow{c_1} M_1$ and $M\xrightarrow{c_2}M_2$, we obtain via \Cref{lemma:sendright} that $\gamma_1=A(\ell_1)$ and $\gamma_2=A(\ell_2)$. Thus we can apply \Cref{lemma:commutereceive} to $N_2 \extoverset{$G_1:(?\ell_1,\gamma_1)$}  \rightsquigarrow N_2'$ and $N_2 \extoverset{$G_2:(?\ell_2,\gamma_2)$}  \rightsquigarrow O_2$ to obtain two further cases:
    \begin{enumerate}
        \item $N_2'  \extoverset{$G_2:(?\ell_2,\gamma_2)$}  \rightsquigarrow Q'$ and $O_2  \extoverset{$G_1:(?\ell_1,\gamma_1)$}  \rightsquigarrow Q'$. Hence $N_1'\op N_2'  \xrightarrow{c_2} Q\op Q'$ and $O_1\op O_2 \xrightarrow{c_1} Q\op Q'$, as desired.
        
        \item $N_2'=O_2$. By \Cref{lemma:receiveonce}, $O_2  \extoverset{$G_2:(?\ell_2,\gamma_2)$}  \rightsquigarrow O_2$;
hence, $N_1'\op N_2'=N_1'\op O_2  \xrightarrow{c_2} Q\op O_2$ and  $O_1\op O_2 \xrightarrow{c_1} Q\op O_2$. 
    \end{enumerate}
    \item $N_1'=O_1$ and $c_1=c_2$. Since $c_1=c_2$, we obtain from \Cref{lemma:uniquereceive} that $N_2'=O_2$. Hence, $M_1=M_2$ and we can conclude. 
    \end{enumerate}
    \item $M_2=O_1\op O_2$, with
    $N_1\extoverset{$G_2:(?\ell_2,\gamma_2)$}  \rightsquigarrow O_1$ and $N_2 \xrightarrow{c_2} O_2$. By \Cref{lemma:commutereceivesend} (applied twice, to $N_1$ and $N_2$, which we can do because, like before, $\gamma_1=A(\ell_1)$ and $\gamma_2=A(\ell_2)$), we obtain $N_1' \extoverset{$G_2:(?\ell_2,\gamma_2)$}  \rightsquigarrow Q$ and $O_1\xrightarrow{c_1} Q$, and $N_2' \xrightarrow{c_2} Q'$ and $O_2 \extoverset{$G_1:(?\ell_1,\gamma_1)$}  \rightsquigarrow Q'$. Hence, $N_1'\op N_2'  \xrightarrow{c_2} Q\op Q'$ and $O_1\op O_2 \xrightarrow{c_1} Q\op Q'$, which concludes the proof.
    \qedhere
    \end{enumerate}
    \end{proof}

\begin{lemma}\label{lemma:relativestep}
    If $\Mone{\varphi}\xrightarrow{A} M\boldarrow M' \boldarrow N\NOT{\xrightarrow{c}}$ for all $c$ and $M'\xrightarrow{c} M''$ for some $c$, then $M''\boldarrow N$.
 \end{lemma}
 \begin{proof}
     We proceed by induction on the length of $M'\boldarrow N$. The base case is for $M'=N$: as $N$ cannot send messages, we then also have $M'\NOT{\xrightarrow{c}}$, and the desired result trivially holds. 
     In the inductive step, we consider $M'\xrightarrow{c'} M_1\boldarrow N$. Let $c=\ell:(!G,\gamma)$ and $c'=\ell':(!G',\gamma')$; 
     by \Cref{lemma:commutecommunication}, either there exists $N'$ such that $M_1\xrightarrow{c} N'$ and $M''\xrightarrow{c'}N'$, or $M''=M_1$. In the latter case we can conclude immediately, because $M_1\boldarrow N$. In the former case we use the induction hypothesis on $\Mone{\varphi}\xrightarrow{A} M \boldarrow M_1 \boldarrow N$ and $M_1\xrightarrow{c} N'$ to obtain that $N'\boldarrow N$ from which we can conclude because $M'\xrightarrow{c} M''\xrightarrow{c'} N'$. 
 \end{proof}

 \begin{lemma}\label{lemma:itsallrelative}
    If $\Mone{\varphi}\xrightarrow{A} M\boldarrow M' \boldarrow N\NOT{\xrightarrow{c}}$ for all $c$ and $M'\boldarrow  M''$, then $M''\boldarrow N$.
\end{lemma}
\begin{proof}
    We proceed by induction on the length of $M'\boldarrow M''$. If the length is $0$, then $M'=M''$ and we are done immediately. Otherwise we consider $M'\xrightarrow{c}M_1\boldarrow M''$ with $c=\ell:(!G,\gamma)$;
    by \Cref{lemma:relativestep}, $M_1\boldarrow N$. From the induction hypothesis on $\Mone{\varphi}\xrightarrow{A} M\boldarrow M_1 \boldarrow N$ and $M_1\boldarrow  M''$,
    we obtain that $M'' \boldarrow N$ and conclude.
\end{proof}


The next three lemmas are used to prove \Cref{lemma:communicationexists}, \Cref{lemma:unioncom}, \Cref{lemma:psijoincom} and \Cref{lemma:crux}. They describe `correct' communication paths of local monitors, where `correct' means that the local monitors only take receiving transitions that correspond to messages containing the last action of the sender. \Cref{lemma:side} shows that in correct communication paths, the order of communication does not matter for the final destination (when no sending and receiving can take place).  \Cref{lemma:sidegeneral} is a direct generalization of \Cref{lemma:side}, and is the equivalent of \Cref{lemma:itsallrelative} on the level of local monitors. \Cref{lemma:merge} is used to show that those correct communication paths can be combined when multiple local monitors are put in parallel, which is used in \Cref{lemma:communicationexists} and \Cref{lemma:crux}.

We first define formally what is such a correct communication path.

       \begin{definition}
    Let $m\boldarrowt m'$ 
    denote the existence of an integer $h>0$, of 
    $h$ monitors $m_1 , \ldots, m_h $ and of $h-1$ communication actions $c_1,\ldots, c_{h-1}$
    such that $m_1 = m$, $m_h = m'$,
    $m_i \xrightarrow{c_i} m_{i+1} $ for every $i = 1,\ldots,h-1$ and, if $c_i=(?\ell,\gamma)$, 
    then 
    $\gamma=A(\ell)$.
\end{definition}
    Hence, $\boldarrowt$ is the same as  $\boldarrow$ but with the extra constraint on actions of kind $(?\ell,\gamma)$.
    
    \begin{lemma}\label{lemma:side}
        Let $m$ be a relevant local monitor such that $ m\boldarrowt m'\NOT{\xrightarrow{c}}$, for all $c$, and either $m\xrightarrow{(?\ell',A(\ell'))}n$, for some $\ell'\in\Loc$, or $m\xrightarrow{(!G,\gamma)}n$, for some $G\subseteq \Loc$. Then $n\boldarrowt m'$.
    \end{lemma}
    \begin{proof}
        Let $m\boldarrowt m'$ in $h$ steps;
    we proceed by induction on $h$.
    In the base case, $m=m'$ and, as $m'$ cannot communicate, the premise of the lemma is not satisfied. 
    Now suppose that $m\xrightarrow{c_1}m_1\boldarrowt m'$, with $c_1=(!G,\gamma)$ or $c_1=(?\ell'',A(\ell''))$ for some $\ell''\in\Loc$. 
    First, let us consider the case in which $m\xrightarrow{(?\ell',A(\ell'))}n$, for some $\ell'\in\Loc$;
    depending on the form of $c_1$, we have two subcases:
    \begin{enumerate}
    \item $c_1=(!G,\gamma)$. We use \Cref{lemma:sendreceiveorder} to obtain that $m_1\xrightarrow{(?\ell',A(\ell'))} n'$ and $n\xrightarrow{c_1} n'$. Hence, $m_1\boldarrowt m'$ in $h-1$ steps, and $m_1\xrightarrow{(?\ell',A(\ell'))} n'$. We apply the induction hypothesis to obtain that $n' \boldarrowt m'$, from which we can conclude because $n\xrightarrow{c_1} n'$.
    
    \item $c_1=(?\ell'',A(\ell''))$ for some $\ell''\in \Loc$. Via \Cref{lemma:commutativereceiving}, we get that either $m_1\xrightarrow{(?\ell',A(\ell'))} n'$ and $n\xrightarrow{c_1} n'$ (in which case we proceed as before), or $\ell'=\ell''$. In the latter case, we use \Cref{lemma:uniqueness} to derive that $m_1=n$ and conclude, because $m_1\boldarrowt m'$.
    \end{enumerate}
    Then, we consider the case in which $m\xrightarrow{(!G,\gamma)}n$, for some $G\subseteq \Loc$, and we again distinguish two subcases:
    \begin{enumerate}
    \item $c_1=(!G',\gamma')$. 
    If $G \neq G'$ or $\gamma \neq \gamma'$, then 
    we use \Cref{lemma:sendreceiveordertwo} to obtain that $m_1\xrightarrow{(!G,\gamma)} n'$ and $n\xrightarrow{c_1} n'$ and proceed like in the first subcase above.
If $G=G'$ and $\gamma=\gamma'$, we obtain from \Cref{lemma:uniquenesstwo} that $m_1=n$ and proceed like the second subcase above.
    
    \item $c_1=(?\ell'',A(\ell''))$ for some $\ell''\in \Loc$. Via \Cref{lemma:sendreceiveorder}, we get that $m_1\xrightarrow{(!G,\gamma)} n'$ and $n\xrightarrow{c_1} n'$; we then proceed as in the first subcase above.
    \qedhere 
    \end{enumerate} 
    \end{proof}
    
    We can extend the previous lemma to a more general statement:
    \begin{lemma}\label{lemma:sidegeneral}
        Let $m$ be a relevant local monitor such that $ m\boldarrowt m'\NOT{\xrightarrow{c}}$, for all $c$, and $m\boldarrowt n$. Then $n\boldarrowt m'$.
    \end{lemma}
\ifarxiversion
    \begin{proof}
        Let $m\boldarrowt n$ in $h$ steps;
        we prove the lemma by induction on $h$.
        In the base case, $n=m$ and we are done immediately. For the inductive step, we consider $m\xrightarrow{c_1} n' \boldarrowt n$, for $c_1=(!G,\gamma)$ or $c_1=(?\ell', A(\ell'))$, where $n' \boldarrowt n$ in $h-1$ steps. By \Cref{lemma:side}, we get $n' \boldarrowt m'$; this, together with $n' \boldarrowt n$, allows us to apply the induction hypothesis and conclude.
    \end{proof}
\fi
    
    \begin{lemma}\label{lemma:merge}
    Let $m_1^1$ and $m_1^2$ be relevant local monitors such that
    $m_1^1\boldarrowt n_1\NOT{\xrightarrow{c}}$ and $m_1^2\boldarrowt n_2\NOT{\xrightarrow{c}}$, for all $c$. 
    Then, $m_1^1\pandor m_1^2\boldarrowt n_1\pandor n_2$, for $\pandor\in\{\pand,\por\}$.
    \end{lemma}
\ifarxiversion
    \begin{proof}
Let $m_1^1\boldarrowt n_1$ in $h$ steps; we proceed by induction on $h$. In the base case, $m_1^1=n_1$; since $n_1$ cannot communicate, we get that $m_1^1\pandor m_1^2\boldarrowt m_1^1\pandor n_2=n_1\pandor n_2$.
    
For the inductive step, we consider $m_1^1\xrightarrow{c_1}n'_1\boldarrowt n_1$, with $c_1=(!G,\gamma)$ or $c_1=(?\ell',A(\ell'))$.
We distinguish three cases.
    \begin{enumerate}
    \item If $c_1=(!G,\gamma)$, then $m_1^1\pandor m_1^2\xrightarrow{c_1} n'_1\pandor m_1^2$. From the induction hypothesis, we get that $n'_1\pandor m_1^2\boldarrowt n_1\pandor n_2$, as desired.
    
    \item $c_1=(?\ell',A(\ell'))$ and $m_1^2\NOT{\xrightarrow{(?\ell',A(\ell'))}}$. In this case $m_1^1\pandor m_1^2\xrightarrow{c_1} n'_1\pandor m_1^2$, and we continue as in the previous case.
    
    \item $c_1=(?\ell',A(\ell'))$ and $m_1^2\xrightarrow{c_1} n'_2$. Thus, $m_1^2\boldarrowt n_2$ and $m_1^2\xrightarrow{(?\ell',A(\ell'))} n'_2$; by \Cref{lemma:side}, we conclude that $n'_2\boldarrowt n_2$. We can then apply the induction hypothesis to conclude that $n'_1\pandor n'_2\boldarrowt n_1\pandor n_2$. Together with  $m_1^1\pandor m_1^2\xrightarrow{c_1} n'_1\pandor n'_2$, this concludes the proof.
    \qedhere 
    \end{enumerate}
    \end{proof}
\fi

    The following lemma shows that, for every receiving transition in a monitor $m$ that is reached from a monitor synthesized from some $\psi\in \quantfree$, there exists a corresponding sending transition in another local monitor reached from a synthesized local monitor for $\psi$ located at a different location. 
    \begin{lemma}\label{lemma:independentcommunication}
       Let $\synvone{\ell}{\sigma}{\psi}\xrightarrow{A(\ell)}m\boldarrow m_1$ and $m_1\xrightarrow{(?\ell',\gamma)}m'$. Then $\ell'\in \range(\sigma)$, $\ell'\neq\ell$, and there are $n,n'$ and $G'$ such that $\synvone{\ell'}{\sigma}{\psi}\xrightarrow{A(\ell')}n\xrightarrow{(!G',
        A(\ell'))} n'$ and $\ell\in G'$.
    \end{lemma}
    \begin{proof}
   We proceed by induction on $\xrightarrow{A(\ell)}$. If $\synone{\psi}=v$, the premise of the lemma is not satisfied; moreover, we do not consider the case where $\synone{\psi}=a.m$, as this cannot occur in the synthesis function (because of our assumption that $|\Act| \geq 2$).
   
    For the first inductive case, we consider $\synone{\psi}=n_1+n_2\xrightarrow{A(\ell)} m$, because $n_1\xrightarrow{A(\ell)} m$ or $n_2\xrightarrow{A(\ell)} m$. Hence, $\psi=[a_\pi]\psi'$ 
or $\psi=\langle a_\pi\rangle\psi'$.
As $m_1\xrightarrow{(?\ell',\gamma)}m'$, we know that $\sigma(\pi)=\ell'\neq\ell$. Thus, from the synthesis function, we  immediately have that  $\ell'\in \range (\sigma)$ and $\synvone{\ell'}{\sigma}{\psi}\xrightarrow{A(\ell')}n$ such that $n\xrightarrow{(!G', A(\ell'))} n'$ with $\ell\in G'$ (being $\ell \in \range(\sigma)$). 
   
    Next we consider $\synone{\psi}=n_1\pandor n_2\xrightarrow{A(\ell)} n'_1\pandor n'_2 (=m) \boldarrow n''_1\pandor n''_2 (=m_1)$, with $n_1\xrightarrow{A(\ell)}n'_1$, $n_2\xrightarrow{A(\ell)} n'_2$, $n'_1\boldarrow n''_1$, and $n'_2\boldarrow n''_2$. From the synthesis function, we know that $\psi=\psi_1\op \psi_2$ (with $\op=\wedge$ if $\pandor=\pand$ and $\op=\vee$ if $\pandor=\por$), $n_1=\synone{\psi_1}$ and $n_2=\synone{\psi_2}$.
The transition $m_1 = n''_1\pandor n''_2 \xrightarrow{(?\ell',\gamma)} m'$ can happen in two ways:
\begin{enumerate}
\item $m' = m'_1 \pandor m'_2$, for $n''_1\xrightarrow{(?\ell' ,\gamma)} m'_1$ and $n''_2\xrightarrow{(?\ell',\gamma)} m'_2$.
From the induction hypothesis, we get that $\ell'\in \range(\sigma)$, $\synvone{\ell'}{\sigma}{\psi_1}\xrightarrow{A(\ell')}o_1 \xrightarrow{(!G_1', A(\ell'))} o_1'$ and $\ell\in G_1'$; similarly, $\synvone{\ell'}{\sigma}{\psi_2}\xrightarrow{A(\ell')} o_2 \xrightarrow{(!G_2', A(\ell'))} o_2'$ and $\ell\in G_2'$. Hence, $\synvone{\ell'}{\sigma}{\psi_1\op \psi_2} =\synvone{\ell'}{\sigma}{\psi_1}\pandor \synvone{\ell'}{\sigma}{\psi_2} \xrightarrow{A(\ell')} o_1\pandor o_2 \xrightarrow{(!G_1', A(\ell'))} o_1'\pandor o_2$ with $\ell\in G_1'$, as desired. 
   
\item $m' =m'_1 \pandor n''_2$, for $n''_1\xrightarrow{(?\ell',\gamma)}m'_1$ and $n''_2\NOT{\xrightarrow{(?\ell',\gamma)}} $ (the symmetric case is identical). 
From the induction hypothesis, we get that $\ell'\in \range(\sigma)$, $\synvone{\ell'}{\sigma}{\psi_1}\xrightarrow{A(\ell')}o \xrightarrow{(!G', A(\ell'))} o'$ and $\ell\in G'$. It is immediate from the synthesis function that $\synvone{\ell'}{\sigma}{\psi_2} \xrightarrow{A(\ell')} o_2$, for some monitor $o_2$.  
    Thus, $\synvone{\ell'}{\sigma}{\psi_1\op \psi_2} =\synvone{\ell'}{\sigma}{\psi_1}\pandor \synvone{\ell'}{\sigma}{\psi_2} \xrightarrow{A(\ell')} o\pandor o_2 \xrightarrow{(!G', A(\ell'))} o'\pandor o_2$ with $\ell\in G'$, as desired. 
\end{enumerate}
   
    Last, we consider $\synone{\psi}=\rec x. o\xrightarrow{A(\ell)}m$ because $o\{^{\rec x.o}/_x\}\xrightarrow{A(\ell)} m$. Thus, $\psi=\max x.\psi'$ and $o=\synone{\psi'}$. By \Cref{lemma:recursion-substitutiontwo}, $\synone{\psi'}\{^{\rec x.\synone{\psi'}}/_x\}=\synone{\psi'\{^{\maxx x.\psi'}/_x\}}$; from induction hypothesis, we obtain that $\ell'\in \range(\sigma)$, $\synvone{\ell'}{\sigma}{\psi'\{^{\maxx x.\psi'}/_x\}}\xrightarrow{A(\ell')}n \xrightarrow{(!G', A(\ell'))} n'$ and $\ell\in G'$. Since $\synvone{\ell'}{\sigma}{\psi}\{^{\rec x.\synvone{\ell'}{\sigma}{\psi}}/_x\}=\synvone{\ell'}{\sigma}{\psi'\{^{\maxx x.\psi'}/_x\}}$, we get that $\synvone{\ell'}{\sigma}{\maxx x.\psi'}=\rec x. \synvone{\ell'}{\sigma}{\psi}\xrightarrow{A(\ell')}n$, which concludes the proof.
    \end{proof}
   
    \Cref{lemma:squiggly_is_not_straight} is used to prove \Cref{lemma:squiggly_on_the_right} and it states that a communicating monitor never makes a choice between a send and a receive, so 
if $M_1$ can send (with some transition labelled with $c$) and reach $M_2$, it cannot have reached $M_2$ by receiving some message (intuitively because the send from $c$ is still present if $M_1$ takes a receiving transition, and it is not present in $M_2$). Instead, if a monitor can both send and receive, this must be because multiple monitors are connected in parallel, either locally with $\por/\pand$ or as communicating monitors via $\vee/\wedge$. 

To this aim, we first define a function $cc$ from relevant (local) monitors to natural numbers that counts the number of top-level send prefixes and prove that it decreases after every send action and remains the same after every receive action. 

\begin{definition}
For arbitrary $\sigma,\ell$ and $\psi$, we define: 
    \begin{itemize}
        \item $cc(\synone{\psi}) = cc(\sum_{a\in \Act}(?\{\ell'\},a).\synone{\psi_a}) = 0$; 
        \item $cc((!G,a).\synone{\psi}) = 1$;
        \item $cc(m_1 \pandor m_2) = cc(m_1) + cc(m_2)$, where $\pandor$ denotes either $\por$ or $\pand$;
        \item $cc([m]_\ell) = cc(m)$; and
        \item $cc(M \op N) = cc(M) + cc(N)$, where $\op$ denotes either $\vee$ or $\wedge$.
    \end{itemize}
\end{definition}
    
    \begin{lemma}\label{lemma:cc}
    Let $m_1$ be a relevant local monitor.
    \begin{enumerate}
    \item If $m_1 \xrightarrow{(!G,\gamma)} m_2$, then $cc(m_1) > cc(m_2)$.
    \item If $m_1 \xrightarrow{(?G,\gamma)} m_3$, then $cc(m_1) = cc(m_3)$.
    \end{enumerate}
    \end{lemma}
    \begin{proof}
    For the first claim, we proceed by induction on $m_1$. In the base cases for $m_1=\synone{\psi}$ and  $m_1=\sum_{a\in \Act}(?\{\ell'\},a).\synone{\psi_a}$, we have that $m_1 \xrightarrow{(!G,\gamma)} m_2$ does not hold (by using \Cref{lemma:nocommunicate} for $m_1=\synone{\psi}$);
for $m_1 = (!G,\gamma).\synone{\psi}$, we have that $m_2=\synone{\psi}$ and, thus, $cc(m_2) = 0<1= cc ((!G,\gamma).\synone{\psi})$. 
For the inductive case, let $m_1=m\pandor m'$, with $m$ and $m'$ relevant monitors; the result follows immediately from the definition of $cc$ and the inductive hypothesis.

For the second claim, we again proceed by induction on $m_1$. 
    The base cases for $m_1=\synone{\psi}$ and $m_1 = (!G,\gamma).\synone{\psi}$ entail that $m_1 \xrightarrow{(?G,\gamma)} m_3$ does not hold; 
    for $m_1 = \sum_{a\in \Act}(?\{\ell'\},a).\synone{\psi_a}$, then $m_3 = \synone{\psi_a}$, for some $\psi_a$, and therefore $cc(m_1) = cc(m_3) = 0$.
    The case for $m_1 = m \pandor m'$ is straightforward from the inductive hypothesis.
    \end{proof}

    \begin{lemma}\label{lemma:squiggly_is_not_straight}
        Let $m_1$ and $M_1$ be action-derived monitors.
        \begin{enumerate}
            \item If $m_1 \xrightarrow{(!G,\gamma)} m_2$ and 
            $m_1 \xrightarrow{(?G,\gamma)} m_3$, then $m_2 \neq m_3$.
            \item If $M_1 \xrightarrow{\ell:(!G,\gamma)} M_2$ and 
            $M_1 \xrightsquigarrow{G:(?\ell,\gamma)} M_3$, then $M_2 \neq M_3$.
        \end{enumerate}
    \end{lemma}
    \begin{proof}    
Notice that $m_1$ is a relevant local monitor by \Cref{lemma:A-derived-iff-relevant}.
For the first claim, we use \Cref{lemma:cc} to obtain that $cc(m_3) = cc(m_1) < cc(m_2)$ and so $m_2\neq m_3$.
    The second statement is proved similarly, by observing that
$M_1\xrightarrow{\ell:(!G,\gamma)} M_2$ implies $cc(M_1) > cc(M_2)$ 
and that
$M_1 \xrightsquigarrow{G:(?\ell,\gamma)} M_3$ implies $cc(M_1)= cc(M_3)$ 
(both claims can be proved by a straightforward induction, by using \Cref{lemma:cc}).
    \end{proof}
    
 From the previous lemma we can then conclude that, if $M_1$ reaches $M_2$ by a sending transition, and we combine $M_1$ with another monitor $N_1$ that transitions to $M_2$ composed with $N_2$, then $N_1$ must have taken a receiving step. This is useful for \Cref{lemma:unioncom} and \Cref{lemma:many}, which are in turn used to prove \Cref{lemma:joincommunication}, \Cref{lemma:communicationexists} and \Cref{lemma:actionbispsi}.

    \begin{lemma}\label{lemma:squiggly_on_the_right}
        Let $M_1$ be action-derived, $M_1 \xrightarrow{\ell:(!G,\gamma)} M_2$ and 
        $M_1 \op N_1 \xrightarrow{\ell:(!G,\gamma)} M_2 \op N_2$.
        Then, $N_1 \xrightsquigarrow{G:(?\ell,\gamma)} N_2$.
    \end{lemma}
\ifarxiversion
    \begin{proof}
        From the semantics and $M_1 \op N_1 \xrightarrow{\ell:(!G,\gamma)} M_2 \op N_2$, we obtain that $M_1 \xrightarrow{\ell:(!G,\gamma)} M_2$ and $N_1 \xrightsquigarrow{G:(?\ell,\gamma)} N_2$, or $N_1 \xrightarrow{\ell:(!G,\gamma)} N_2$ and $M_1 \xrightsquigarrow{G:(?\ell,\gamma)} M_2$. In the first case we are done; in the second case, 
$M_1 \xrightarrow{\ell:(!G,\gamma)} M_2$ and $M_1 \xrightsquigarrow{G:(?\ell,\gamma)} M_2$ would contradict  \Cref{lemma:squiggly_is_not_straight}, hence $N_1 \xrightsquigarrow{G:(?\ell,\gamma)} N_2$.
    \end{proof}
\fi

   The next lemma states that action transitions are deterministic for synthesized monitors, and is used in the proof of \Cref{lemma:unioncom}.
    \begin{lemma}\label{lemma:deterministicstep}
       If $\synone{\psi}\xrightarrow{A(\ell)} m$ and $\synone{\psi}\xrightarrow{A(\ell)}n$, then $m=n$.
       \end{lemma}
       \begin{proof}
           We proceed by induction on $\xrightarrow{A(\ell)}$. The interesting cases are the inductive ones. 
           If $\synone{\psi}=m_1+m_2\xrightarrow{A(\ell)} m$, then either $\psi=[a_\pi]\psi'$ 
           or $\psi=\langle a_\pi\rangle\psi'$, for some $a$, $\pi$ and $\psi'$. As there is only one possible transition for each $a\in \Act$, we have obtained the required result.
       
           If $\synone{\psi}=\rec {x}.m_1\xrightarrow{A(\ell)} m$ because $m_1\{^{\rec x.m_1}/_x\} \xrightarrow {A(\ell)} m$, we know that $\psi=\maxx x.\psi'$ and $m_1=\synone{\psi'}$.
By \Cref{lemma:recursion-substitutiontwo}, $m_1\{^{\rec x.m_1}/_x\} = \synone{\psi'\{^{\psi}/_x\}}$ 
(note that the free location-variables of $\psi$ are not bound in $\psi'$); thus, we can use the induction hypothesis to get $m=n$.
           
           If $\synone{\psi}=m_1\pandor m_2\xrightarrow{A(\ell)} m_3\pandor m_4=m$, we know from the synthesis definition that $\psi= \psi_1\op\psi_2$ (for $\op = \vee$ if $\pandor = \por$ and $\op = \wedge$ if $\pandor = \pand$), for some $\psi_1$ and $\psi_2$ such that $m_1=\synone{\psi_1}$ and $m_2=\synone{\psi_2}$. Hence, $\synone{\psi_1}\xrightarrow{A(\ell) } m_3$ and $\synone{\psi_2}\xrightarrow{A(\ell) }m_4$. Similarly, we know that $n=n_1\pandor n_2$ for some $n_1$ and $n_2$ such that $\synone{\psi_1}\xrightarrow{A(\ell) } n_1$ and $\synone{\psi_2}\xrightarrow{A(\ell) }n_2$. We use the induction hypothesis to conclude that $m_3=n_1$ and $m_4=n_2$, from which we derive that $m=n$.
       \end{proof}
   
       In what follows we need to specify a condition on a group of local monitors who all take some `correct' set of communication transitions (correct as in following the restrictions in $\boldarrowt$). This condition states that the reached states in the group of local monitors will no longer take receiving transitions that match a sending transition that occurred on one of the local communication paths.
       
       \begin{definition}
    Let $\ell_1,\ldots,\ell_k \in \Loc$ and $m_1,\ldots,m_k , n_1,\ldots,n_k$ be monitors. 
   We write 
   $m_i \xybRightarrow{\ell_i}{A}{k}{i=1} n_i$ 
   if, for every $i\in\{1,\ldots, k\}$, there exist an integer $h_i>0$,
    $h_i$ monitors $m_i^1 , \ldots, m_i^{h_i} $, and $h_i-1$ actions $c_i^1,\ldots,c_i^{h_i-1}$
    such that $m_i = m_i^1$, $m_i^{h_i} = n_i$,
    $m_i^j \xrightarrow{c_i^j} m_{i}^{j+1} $ for every $j = 1,\ldots,{h_i}-1$ and 
    \begin{enumerate}
   \item if $c_i^j=(?\ell,\gamma)$ for some $\gamma$ and $\ell$, then 
    $\gamma=A(\ell)$; 
    \item if $c_i^j = (!G,\gamma)$ for some $G$ and $\gamma$, then, for every $i'\neq i$, 
   $\ell_{i'} \in G$ implies that 
   $n_{i'}\NOT{\xrightarrow{(?\ell_i,\gamma)}}$.
   \end{enumerate}  
\end{definition}
   Essentially, $m_i \xybRightarrow{\ell_i}{A}{k}{i=1} n_i$  amounts to saying that, for every $i \leq k$, 
   $m_i\boldarrowt n_i$ plus condition 2.
   

   The next lemma is used in the proof of \Cref{lemma:psijoincom}. It proves that local monitors that `accidentally' receive messages from another local monitor's communication transitions are such that they cannot take any more receiving transitions matching those accidental messages. This captures the intuition that a local monitor always takes all receiving transitions that match a message at once, and no new receiving transitions appear for a local monitor when it is taking a chain of consecutive communication transitions. 

Important in the proof is \Cref{lemma:independentcommunication}: in $\boldarrowt$ local monitors can send and receive freely, but in $\boldarrow$ between decentralized monitors, communication is always initiated by a sender. Therefore, all receiving transitions that need to happen need to be initiated by some sending transition.
   \begin{lemma}\label{lemma:unioncom} 
Let $\range (\sigma)=\{\ell_1,\cdots, \ell_k\} \neq \emptyset$
 and, for every $i \in \{1,\ldots,k\}$, 
 $\synvone{\ell_i}{\sigma}{\psi}\xrightarrow{A(\ell_i)}m_{i} \boldarrowt n_{i} \boldarrowt o_{i}$, with $o_{i}$ that cannot communicate. 
If $m_{i} \xybRightarrow{\ell_i}{A}{k}{i = 1} n_{i}$, 
then for every $j \in \{1,\cdots,k\}$ there exist some
$q_{1}, \ldots , q_k$ such that 
   $\bigvee_{i=1}^k[n_{i}]_{\ell_i} \boldarrow \bigvee_{i=1}^k[q_{i}]_{\ell_i}$, 
   $n_{i}\xybRightarrow{\ell_i}{A}{k}{i=1}  q_{i}$,
   and $q_{j}=o_j$.
   \end{lemma}
   \begin{proof}
Fix any $j \in \{1,\ldots,k\}$ and 
let $n_j \boldarrowt o_j$ in $h_j$ steps; we proceed by induction on $h_j$. In the base case, we have that $n_j=o_j$ and so it suffices to consider $q_i = n_i$, for every $i$.
   
   For the inductive step, let $n_j \xrightarrow{c} n'_j \boldarrowt o_j$ with $c=(!G,\gamma)$ or $c=(?\ell,A(\ell))$ for some $\ell\in\Loc$:
   \begin{enumerate} 
   \item Case $c=(!G,\gamma)$: 
by \Cref{lemma:singlesender}, we know that $\gamma=A(\ell_j)$; moreover,
   $\bigvee_{i =1}^k[n_{i}]_{\ell_i} \xrightarrow{\ell_j:(!G,\gamma)}
   \bigvee_{i=1}^{k}
   [n'_{i}]_{\ell_i}$, where $\bigvee_{i \in \{1..k\} \setminus \{j\}}[n_{i}]_{\ell_i} \extoverset{$G:(?\ell_j,\gamma)$}\rightsquigarrow \bigvee_{i \in \{1..k\} \setminus \{j\}}[n'_{i}]_{\ell_i}$ by \Cref{lemma:squiggly_on_the_right}.
   
   We now argue that $m_i \xybRightarrow{\ell_i}{A}{k}{i=1}  n'_{i}$. For each $\ell_i \in \range(\sigma)$, we have that $m_{i}\boldarrowt n_{i}$, say in $g_{i}$ steps. 
   By \Cref{lemma:group}, we obtain that either $n_{i}\xrightarrow{(?\ell_j,A(\ell_j))} n'_{i}$ or $n_{i}=n'_{i}$, for all $i \neq j$; thus, $m_{i} \boldarrowt n'_{i}$ for each $\ell_i \in\range(\sigma)$ in $f_{i} = g_{i}+1$ or $f_{i} = g_{i}$ steps, respectively.
   Furthermore, if $f_{i} = g_{i} +1$,
   then $c_{i}^{g_{i}+1}=(?\ell_j,A(\ell_j))$, for all $i \neq j$, and $c_{j}^{g_{j}+1}=c$.
   For the second condition of $m_{i} \xybRightarrow{\ell_i}{A}{k}{i=1}  n'_{i}$, 
   let
   $i \neq j$ and $m_{i} \boldarrowt n'_i$ in 
   $f_{i}$ steps;
   let $f \leq f_{i}$ and $c_{i}^f=(!G',\gamma')$. 
   Now take some $\ell_h \in G'\cap\range(\sigma)$ with $h \neq i$, and we have to show that $n'_{h} \NOT{\xrightarrow{(?\ell_i,\gamma')}}$. 
   If $f \leq g_i$, then $m_{i} \xybRightarrow{\ell_i}{A}{k}{i =1}  n_{i}$ implies that $n_{h} \NOT{\xrightarrow{(?\ell_i,\gamma')}}$
   %
and, by \Cref{lemma:singlesender}, we can conclude that $n'_{h} \NOT{\xrightarrow{(?\ell_i,\gamma')}}$. 
If $f=g_{i}+1$, then we have reasoned above that $i=j$, $G'=G$ and $\gamma'=\gamma$. 
   Furthermore, we know that either $n_{h}\xrightarrow{(?\ell_i,\gamma)} n'_{h}$ or $n'_{h}= n_{h}\NOT{\xrightarrow{(?\ell_i,\gamma)}}$; in both cases, \Cref{lemma:singlesender} entails that $n'_{h}\NOT{\xrightarrow{(?\ell_i,\gamma)}}$ and so $m_{i} \xybRightarrow{\ell_i}{A}{k}{i =1}  n'_{i}$.

By \Cref{lemma:side}, for all $\ell_i \in \range(\sigma)$, we have that 
   $\synvone{\ell_i}{\sigma}{\psi}\xrightarrow{A(\ell_i)}m_{i}\boldarrowt n'_{i} \boldarrowt o_{i}$. Since $m_{i}\xybRightarrow{\ell_i}{A}{k}{i=1}  n'_{i}$, we can apply the induction hypothesis to obtain that 
   $\bigvee_{i=1}^{k}[n'_{i}]_{\ell_i}\boldarrow \bigvee_{i=1}^{k}[q_{i}]_{\ell_i}$, for some $q_i$s such that $n'_{i}\xybRightarrow{\ell_i}{A}{k}{i=1}  q_{i}$ and $q_j = o_j$. 
It remains to show that $n_{i} \xybRightarrow{\ell_i}{A}{k}{i=1}  q_{i}$: 
   $m_{i} \xybRightarrow{\ell_i}{A}{k}{i=1}  n'_{i}$
   and 
   $n'_{i} \xybRightarrow{\ell_i}{A}{k}{i=1}  q_{i}$ 
   immediately yield 
   the first condition for $n_{i} \xybRightarrow{\ell_i}{A}{k}{i=1}  q_{i}$
   and they yield the second condition by using 
   \Cref{lemma:singlesender}.
   
   
   \item Case $c=(?\ell,A(\ell))$: by \Cref{lemma:independentcommunication}, $\ell \in \range(\sigma)$ (say, $\ell = \ell_h$), $j \neq h$, and 
   $\synvone{\ell_h}{\sigma}{\psi}\xrightarrow{A(\ell_h)}p_{h} \xrightarrow{(!G',A(\ell_h))}p'_{h}$, for some $G' \subseteq \Loc$ such that $\ell_j \in G'$. 
   From \Cref{lemma:deterministicstep}, $p_{h} = m_{h}$.
   Since $m_{i} \xybRightarrow{\ell_i}{A}{k}{i=1}  n_{i}$, then in particular
   $m_{h} \boldarrowt n_{h}$ with the same transitions;
   by
   \Cref{lemma:sendreceiveordertwo} and \Cref{lemma:sendreceiveorder}, $n_{h} \xrightarrow{(!G',A(\ell_h))} n'_{h}$, if 
   $m_{h}\boldarrowt n_{h}$ 
   does not include any transition labeled by $(!G',A(\ell_h))$. 
   If such a transition does exist, from  $m_{i} \xybRightarrow{\ell_i}{A}{k}{i=1}  n_{i}$, we get that $n_{i}\NOT{\xrightarrow{(?\ell_h,A(\ell_h))}}$, for all $i \neq h$ and $\ell_i \in \range(\sigma)\cap G'$; but 
   this would contradict $n_j \xrightarrow{c} n'_j$, since $\ell_j \in G'$. 
   Hence, we have $n_{h} \xrightarrow{(!G',A(\ell_h))} n'_{h}$ and so
   $\bigvee_{i=1}^k [n_{i}]_{\ell_i} \xrightarrow{\ell_h:(!G',A(\ell_h)) }\bigvee_{i=1}^k[n'_{i}]_{\ell_i}$, 
   where, by  \Cref{lemma:squiggly_on_the_right}, 
   $\bigvee_{i \in \{1..k\} \setminus\{h\}} [n_{i}]_{\ell_i} \extoverset{$G':(?\ell_h,A(\ell_h))$} \rightsquigarrow \bigvee_{i \in \{1..k\} \setminus\{h\}} [n'_{i}]_{\ell_i}$.
   We can now argue that $m_{i} \xybRightarrow{\ell_i}{A}{k}{i=1} n'_{i}$ as in case 1 above. 
By \Cref{lemma:sidegeneral}, $n'_{i} \boldarrowt q_{i}$ for all $i \in \{1,\ldots,k\}$ and, by the induction hypothesis, 
$\bigvee_{i =1}^k [n'_{i}]_{\ell_i} \boldarrow \bigvee_{i =1}^k [q_{i}]_{\ell_i}$ 
with $n'_{i} \xybRightarrow{\ell_i}{A}{k}{i=1} q_{i}$ and $q_j=o_j$. 
We can then conclude, by arguing that $n_{i}\xybRightarrow{\ell_i}{A}{k}{i=1}  q_{i}$ as in case 1.
   \qedhere
   \end{enumerate}
   \end{proof}
   
   The next lemma is important for \Cref{lemma:communicationexists} and \Cref{lemma:actionbispsi} (part of the lemmas for Formula Convergence). It is actually stronger than what is needed for \Cref{lemma:communicationexists}, as therein we shall just need that $\bigvee_{i=1}^k [n_i]_{\ell_i}$ transitions to some monitor that cannot comunicate, and not exactly  to $\bigvee_{i=1}^k [o_i]_{\ell_i}$. However, the latter is needed for \Cref{lemma:actionbispsi}, so we present the stronger statement already here. The lemma shows that local monitors can be combined and, even if their communications `accidentally' influence one another, the same final state (where no communication can take place) is reached. It also shows another important property: in $\boldarrowt$ local monitors can send and receive freely, but in $\boldarrow$ between decentralized monitors communication is always initiated by a sender. This lemma captures that all the receiving transitions a local monitor can do before it reaches a state that cannot communicate are triggered by sending actions present in other local monitors located at locations found in the range of $\sigma$.

   \begin{lemma}\label{lemma:psijoincom}
   Let $\range(\sigma)=\{\ell_1,\dots,\ell_k\}$ and, for all $i \in \{1,\ldots,k\}$, 
   $\synvone{\ell_i}\sigma{\psi}\xrightarrow{A(\ell_i)}m_i \boldarrowt n_i \boldarrowt o_i$, with $o_i$ that cannot communicate. 
If $m_i\xybRightarrow{\ell_i}{A}{k}{i=1} n_i$, then $\bigvee_{i=1}^k [n_i]_{\ell_i} \boldarrow \bigvee_{i=1}^k [o_i]_{\ell_i}$. 
   \end{lemma}
   \begin{proof}
       We proceed by induction on $k$. If $k=0$, the result is trivial.
       In the inductive step, 
by \Cref{lemma:unioncom}, there exist $q_1,\ldots,q_k$ and $j \in \{1,\ldots,k\}$ such that $\bigvee_{i=1}^k[n_{i}]_{\ell_i} \boldarrow \bigvee_{i =1}^k[q_{i}]_{\ell_i}$, $n_{i} \xybRightarrow{\ell_i}{A}{k}{i=1} q_{i}$, and $q_j=o_j$. 
For all $i \neq j$, we have $\synvone{\ell_i}\sigma{\psi}\xrightarrow{A(\ell_i)}m_{i} \boldarrowt n_{i} \boldarrowt o_i$ and $n_{i}\boldarrowt q_{i}$; by \Cref{lemma:sidegeneral}, $q_{i} \boldarrowt o_i$. 
From this, we derive that $\synvone{\ell_i}\sigma{\psi}\xrightarrow{A(\ell_i)}m_{i}\boldarrowt q_{i}\boldarrowt o_i$, for all $i \in \{1,\ldots,k\}$. 
   
We now show that $m_{i} \xybRightarrow{\ell_i}{A}{k}{i=1}  q_{i}$.
   For each $i \in \{1,\ldots,k\}$, we have $m_{i}\boldarrowt n_{i}$ (say, in $f_i$ steps) and $n_i \boldarrowt q_i$ (say, in $g_i$ steps); thus, $m_i\boldarrowt q_i$ in $h_i=f_i+g_i$ steps. 
   For the second condition of $m_{i}\xybRightarrow{i}{A}{k}{i=1}  q_{i}$, take any $i$, any $h \leq h_i$ (with $c_i^h = (!G,\gamma)$) and any $\ell_{i'} \in G\cap\range(\sigma)$ with $i'\neq i$, and we have to show that $q_{i'}\NOT{\xrightarrow{(?\ell_i,\gamma)}}$. 
   If $h \leq f_i$, then $m_{i}\xybRightarrow{\ell_i}{A}{k}{i=1} n_{i}$ entails that $n_{i'}\NOT{\xrightarrow{(?\ell_i,\gamma)}}$; by \Cref{lemma:singlesender}, $q_{i'}\NOT{\xrightarrow{(?i\ell_,\gamma)}}$. 
   On the other hand, if $h >f_i$, the result follows from $n_{i}\xybRightarrow{\ell_i}{A}{k}{i=1} q_{i}$.
   
Finally, the induction hypothesis entails that $\bigvee_{i \in \{1..k\}\setminus\{j\}}[q_{i}]_{\ell_i} \boldarrow \bigvee_{i \in \{1..k\}\setminus\{j\}}[o_i]_{\ell_i}$. Hence, $\bigvee_{i \in \{1..k\}} [n_i]_{\ell_i} \boldarrow \bigvee_{i \in \{1..k\}\setminus\{j\}}[q_{i}]_{\ell_i} \vee [o_j]_{\ell_j} \boldarrow \bigvee_{i \in \{1..k\}}[o_i]_{\ell_i}$. In the last step we use that $o_j$ cannot communicate and therefore does not change anymore.
   \end{proof}   

Another important feature that we need for Bounded Communication and Formula Convergence is the following: if different monitors influence each other via communication when combined (so a decentralized monitor is combined with another monitor and all of a sudden a local monitor in the decentralized system receives messages that it did not receive before), this makes no difference on the states the monitors reach when no more communication can take place. 

We first set up a definition that helps us
with two things. First, it allows us to
describe what a decentralized monitor $M$ looks like if it is combined via $\vee$ or $\wedge$ with another decentralized monitor, and then transitions according to a communication path of that other monitor. In other words, $M$ will change state when the other monitor sends messages that local monitors in $M$ can receive, and we formally capture what $M$ looks like after these changes. Second, this definition will allow us to prove that influences from another monitor captured via \Cref{def:completeset} do not affect the final outcome (when no more communication takes place). These two elements then lead to \Cref{lemma:joincommunication}, where the different monitors in a join will send messages that `accidentally' reach the others and it is shown this is not a problem. \Cref{lemma:joincommunication} is used to prove \Cref{lemma:communicationexists} (Bounded Communication) and \Cref{lemma:actionbis} (Formula Convergence).

To this end, we define a set $L$ of monitors that can all take the same receiving transition. \Cref{def:completeset} defines this set properly, and \Cref{lemma:preservecompleteness} proves three general facts about this set.
\begin{definition}\label{def:completeset}
    Let $H\subseteq \Loc$ and $M\in \DMon$. Let $L=\{[m_1]_
    {\ell_1}, \dots ,[m_k]_{\ell_k}\}$. 
    We say that $L$ is a \emph{set of recipients for $H:(?\ell', \gamma)$} if, for all $j\in \{1,\dots,k\}$:
    \begin{enumerate}
        \item $\ell_j\in H$, and 
        \item $m_j\xrightarrow{(?\ell',\gamma)} n_j$, for some $n_j\in\LMon$.
    \end{enumerate}
    We call a set $L$ an \emph{$M$-cover} for $H:(?\ell', \gamma)$ if 
    \begin{enumerate}
        \item $L$ is a set of recipients for $H:(?\ell', \gamma)$, 
        \item every $[m]_\ell \in L$ 
        is action-derived, and 
        \item if $[m]_\ell\in M$, $\ell\in H$ and 
    $m \boldarrow m'\xrightarrow     {(?\ell',\gamma)} n $, 
    then $[m']_\ell\in L$. 
    \end{enumerate}
\end{definition}

\begin{lemma}\label{lemma:preservecompleteness}
    Let $H\subseteq \Loc$, $M$ be action-derived, and $L=\{[m_1]_
    {\ell_1}, \dots ,[m_k]_{\ell_k}\}$ be an $M$-cover for $H:(?\ell, \gamma)$. 
    \begin{enumerate}
        \item If $M=[m]_{\ell'}\notin L$, and $M\xrightarrow{\ell':(!G,\gamma)} N$ or $M\extoverset{$G:(?\ell'',\gamma)$}\rightsquigarrow N$, then $N=[n]_{\ell'}\notin L$.
        \item If $M\extoverset{$G:(?\ell',\gamma)$}\rightsquigarrow N$, then $L$ is also an $N$-cover for $H:(?\ell, \gamma)$.
        \item If $M\xrightarrow{\ell':(!G,\gamma)} N$, then $L$ is also an $N$-cover for $H:(?\ell, \gamma)$.
    \end{enumerate}
\end{lemma}
\begin{proof}
    For the first item, we observe that $N=[n]_{\ell'}$, where 
    \begin{itemize}
        \item $m=n$; or 
        \item $m\xrightarrow{(?\ell'',\gamma)}n$, or;
        \item $m\xrightarrow{(!G,\gamma)}n$. 
    \end{itemize}
    In the first case, we have $[n]_{\ell'}=[m]_{\ell'}\notin L$. In the other two cases, because $L$ is an $M$-cover, either $\ell'\notin H$, and so $[n]_{\ell'}\notin L$, or $m\NOT{\xrightarrow{(?\ell, \gamma)}}$, and by \Cref{lemma:singlesender} this implies that $n\NOT{\xrightarrow{(?\ell, \gamma)}}$ and again $[n]_{\ell'}\notin L$.
    
    For the second item, we proceed by induction on $\extoverset{$G:(?\ell',\gamma)$}\rightsquigarrow$.
    For the base case, $M=[m]_{\ell''}$, $N=[n]_{\ell''}$, $\ell''\in G$, and $m=n$ or $m\xrightarrow{(?\ell',\gamma)}n$. Suppose that $n\boldarrow n_1\xrightarrow{(?\ell,\gamma)} n_2$ and $\ell''\in H$. Then, we also know that $m\boldarrow n_1\xrightarrow{(?\ell,\gamma)} n_2$; as $[m]_{\ell''}=M$, because $L$ is an $M$-cover, we obtain that $[n_1]_{\ell''}\in L$, yielding that $L$ is an $N$-cover. 
    For the inductive case, $M=M_1\op M_2$, $N=N_1\op N_2$, $M_1 \extoverset{$G:(?\ell',\gamma)$}\rightsquigarrow N_1$ and $M_2 \extoverset{$G:(?\ell',\gamma)$}\rightsquigarrow N_2$. Since $L$ is an $M$-cover, it is also an $M_1$-cover and an $M_2$-cover. From the induction hypothesis, $L$ is an $N_1$-cover and an $N_2$-cover, and thus an $N$-cover.

For the third item, we proceed by induction on $\xrightarrow{\ell' :(!G,\gamma)}$. For the base case, $M=[m]_{\ell'}$, $N=[n]_{\ell'}$, and $m\xrightarrow{(!G,\gamma)}n$. Suppose that $n\boldarrow n_1\xrightarrow{(?\ell,\gamma)} n_2$. Then, we also know that $m\boldarrow n_1\xrightarrow{(?\ell,\gamma)} n_2$; as $[m]_{\ell'}=M$, because $L$ is an $M$-cover, we obtain that $[n_1]_{\ell'}\in L$, yielding that $L$ is an $N$-cover. 
For the inductive case, $M=M_1\op M_2$, $N=N_1\op N_2$, $M_1 \xrightarrow{\ell':(!G,\gamma)} N_1$ and $M_2 \extoverset{$G:(?\ell',\gamma)$}\rightsquigarrow N_2$. Since $L$ is an $M$-cover, it is also an $M_1$-cover and an $M_2$-cover. From the induction hypothesis and the previous item, $L$ is an $N_1$-cover and an $N_2$-cover, and thus an $N$-cover.
\end{proof}

In what follows it is useful to describe the state of a monitor $M$ after receiving the message corresponding to $L$. To do this formally, we observe the following:

\begin{lemma}\label{lemma:unique-transitions-from-covers}
    Let $L$ be an $M$-cover for $H:(?\ell,\gamma)$. 
    For every $[m]_{\ell'}\in L$, if 
    $m \xrightarrow{(?\ell,\gamma)} m_1$ and $m \xrightarrow{(?\ell,\gamma)} m_2$, then $m_1 = m_2$ and $[m_1]_\ell \notin L$.
\end{lemma}
\ifarxiversion
\begin{proof}
    The lemma follows immediately from \Cref{lemma:singlesender,lemma:uniqueness}.
\end{proof}
\fi

With this lemma, we can then define what $M$ after $L$ is:
\begin{definition}\label{def:M_L_}
    Let $L= \{[m_1]_
{\ell_1},\dots ,[m_k]_{\ell_k}\}$ 
be an $M$-cover for $H:c$ and, for each $i \leq k$, let $n_i$ be the unique monitor (by \Cref{lemma:unique-transitions-from-covers}) such that 
$m_i \xrightarrow{c} n_i$.
We define $M[L] = M[[n_1]_{\ell_1}/[m_1]_{\ell_1}]\dots [[n_k]_{\ell_k}/
[m_k]_{\ell_k}]$.
\end{definition}

We can now use \Cref{def:completeset,def:M_L_} and \Cref{lemma:preservecompleteness} to describe a monitor that `accidentally' receives messages from another monitor's communication transitions.
\begin{lemma}
    \label{lemma:many}
    Let 
    $M_1$ and $M_2$ be $A$-derived and such that
    $M_1^1\xrightarrow{c_1}M_1^2\xrightarrow{c_2}\dots 
    \xrightarrow{c_{n-1}} M_1^n$
and
    $M_1^1\op
    M_2^1
    \xrightarrow{c_1} 
    M_1^2\op
    M_2^2
    \xrightarrow {c_2}
    \dots\xrightarrow{c_{n-1}}M^n_1
    \op
    M^n_2
    $, 
    where 
    $M_1^1=M_1$, 
    $M_2^1=M_2$, 
    $n> 0$,
    and  
    $c_i=\ell_i:(!G_i,A
    (\ell_i))$, for every $i\in\{1,\dots,n-1\}$. 
    Then, for every $i\in\{1,\dots,n-1\}$, there exists 
    an $M_2$-cover 
    $L_i$ for 
    $G_i:(?\ell_i,A(\ell_i))$,
    such that
    \(M_2[L_1]\dots [L_i]= M_2^{i+1}\). 
\end{lemma}
\begin{proof}
    We proceed by induction on $n$. If $n=1$, then 
    the lemma holds vacuously. 
%
      Otherwise, we have 
      \[M_1^1\op
        M_2^1
        \xrightarrow{c_1} 
        M_1^2\op 
        M_2^2
        \xrightarrow {c_2}
        \dots
        \xrightarrow{c_{n-2}}M^{n-1}_1
        \op
        M^{n-1}_2
        \xrightarrow{c_{n-1}}M^n_1
        \op
        M^n_2.
      \] 
    The induction hypothesis tells us that there exist 
    $L_1,\dots,L_{n-2}$ such that, for each $i\in\{1,\dots,n-2\}$, $L_i$ is an 
    $M_2$-cover
    for $G_i:(?\ell_i,A(\ell_i))$, 
    and 
    $M_2[L_1]\cdots[L_{i}] = M_2^{i+1}$. 
    Let
    \[L_{n-1} = \{ [m']_\ell \mid \ell \in G_{n-1} \text{ and } \exists\,  [m]_\ell \in M_2.~m \boldarrow m'\xrightarrow  {(?\ell_{n-1},A(\ell_{n-1}))} m''  \}.\] 
    Then, $L_{n-1}$ is an 
    $M_2$-cover
    for $G_{n-1}:(?\ell_{n-1},A(\ell_{n-1}))$ by \Cref{def:completeset},
    and also an 
    $M_2^{i}$-cover
    for $G_{n-1}:(?\ell_{n-1},A(\ell_{n-1}))$ by \Cref{lemma:squiggly_on_the_right,lemma:preservecompleteness}.

    What is left to show is that $M_2[L_1]\cdots[L_{n-2}][L_{n-1}] = M_2^{n}$. We can conclude this from the induction hypothesis, if we show that $M_2^{n-1}[L_{n-1}]=M_2^n$. By \Cref{lemma:squiggly_on_the_right}, $M_2^{n-1}\xrightsquigarrow{G_{n-1}:(\ell_{n-1},\gamma_{n-1})}M_2^n$. 
    Let $\{ [m_i]_{\ell_i} \mid 1 \leq i \leq k \} = \{ [m]_\ell \in M_2^{n-1} \}$.
    For each $1 \leq i \leq k$, let $m'_i$ be $m_i$, when $\ell_i \notin G_{n-1}$ or $m_i \NOT{\xrightarrow{(?\ell_{n-1},\gamma_{n-1})}}$, or be such that $m_i \xrightarrow{(?\ell_{n-1},\gamma_{n-1})} m'_i$, otherwise.
    Then, $M[L_{n-1}] = M_2^{n-1}[[m'_1]_{\ell_1}/[m_1]_{\ell_1}]\cdots [[m'_k]_{\ell_k}/[m_k]_{\ell_k}] = M_2^{n}$, by \Cref{lemma:group}.
\end{proof}

\Cref{lemma:stillreceive} and \Cref{lemma:onestep} are used to prove \Cref{lemma:receivealwayshappens}, and they state some basic facts of behavior of $M[L]$ w.r.t. the behavior of $M$.

\begin{lemma}\label{lemma:stillreceive}
    Let $H\subseteq \Loc$, $M$ be $A$-derived, and $L$ be an $M$-cover for $H:(?\ell', A(\ell'))$. If $M\extoverset{$G:(?\ell'',A(\ell''))$}\rightsquigarrow N$, for some $\ell''\in \Loc$, then $M[L]\extoverset{$G:(?\ell'',A(\ell''))$}\rightsquigarrow N[L]$.
\end{lemma}
\begin{proof}
We proceed by induction on $\extoverset{$G:(?\ell'',A(\ell''))$}\rightsquigarrow$. In the first base case, $M=[m]_\ell$, $N=[n]_\ell$, $\ell\in G$ and $m\xrightarrow{(?\ell'',A(\ell''))}n$. If $[m]_\ell\notin L$, then $[n]_\ell\notin L$ via \Cref{lemma:preservecompleteness}. Hence, $M=M[L]$ and $N=N[L]$, which concludes the proof. Otherwise, $M[L]=[m_1]_\ell$ with $m\xrightarrow{(?\ell',A(\ell'))}m_1$ and $\ell\in H$. We distinguish two cases:
\begin{enumerate}
    \item $\ell''\neq \ell'$. We then apply \Cref{lemma:commutativereceiving} on $m$:
    $n\xrightarrow{(?\ell',A(\ell'))} n_1$ and $m_1\xrightarrow{(?\ell'',A(\ell''))} n_1$. From this, we derive that $M[L]=[m_1]_\ell\extoverset{$G:(?\ell'',A(\ell''))$}\rightsquigarrow [n_1]_\ell$. Since $L$ is an $M$-cover, it is also an $N$-cover via \Cref{lemma:preservecompleteness}. 
    Thus, from $n\xrightarrow{(?\ell',A(\ell'))} n_1$ we obtain that 
    $N = [n]_\ell
    \in L$. Hence,  from 
    $n \xrightarrow{(?\ell',A(\ell'))} n_1$ and 
    via \Cref{lemma:unique-transitions-from-covers},  we get 
    $N[L] = [n_1]_\ell$
    and $M[L]=[m_1]_\ell\extoverset{$G:(?\ell'',A(\ell''))$}\rightsquigarrow [n_1]_\ell=N[L]$.

    \item $\ell''=\ell'$. Then, $A(\ell'') = A(\ell')$, and, via \Cref{lemma:uniqueness}, we obtain that $m_1=n$; so, $M[L]=[m_1]_\ell=[n]_\ell=N$. From $m\xrightarrow{(?\ell',A(\ell'))}m_1$ and \Cref{lemma:singlesender}, we conclude that $m_1\NOT{\xrightarrow{(?\ell',A(\ell'))}}$. Since $\ell''=\ell'$ and $m_1=n$, it holds that $n\NOT{\xrightarrow{(?\ell'',A(\ell''))}}$. 
    and $n\NOT{\xrightarrow{(?\ell',A(\ell'))}}$, which implies that $[n]_\ell \notin L$, and therefore $N=N[L]$. 
    From $n\NOT{\xrightarrow{(?\ell'',A(\ell''))}}$ we derive that $M[L]=[n]_\ell\extoverset{$G:(?\ell'',A(\ell''))$}\rightsquigarrow [n]_\ell=N=N[L]$, which is what we wanted to show.
\end{enumerate}

In the second base case, we consider $M=[m]_\ell$, $N=[m]_\ell$, and $m
\NOT {\xrightarrow{(?\ell'',A(\ell''))}}$. If $M=M[L]$, then $[m]_\ell\notin L$ (by \Cref{lemma:unique-transitions-from-covers}), and thus also $N=N[L]$, which concludes the proof. Otherwise, we have 
that $M[L]=[m_1]_\ell$ with $m\xrightarrow{(?\ell',A(\ell'))}m_1$ and $\ell\in H$. We know that $M=N$, so $M[L]=N[L]$. We know that $m
\NOT {\xrightarrow{(?\ell'',A(\ell''))}}$ and, via \Cref{lemma:singlesender}, that $m_1
\NOT {\xrightarrow{(?\ell'',A(\ell''))}}$. We can conclude that $M[L]=[m_1]_\ell\extoverset{$G:(?\ell'',A(\ell''))$}\rightsquigarrow [m_1]_\ell=N[L]$.

In the third base case, we consider $M=[m]_\ell$, $N=[m]_\ell$ and $\ell\notin G$. If $M=M[L]$, then $[m]_\ell\notin L$, and thus also $N=N[L]$, which concludes the proof. Otherwise, we have 
that $M[L]=[m_1]_\ell$ with $m\xrightarrow{(?\ell',A(\ell'))}m_1$ and $\ell\in H$. We know that $M=N$, so $M[L]=N[L]$. Since $\ell\notin G$, we can conclude that $M[L]=[m_1]_\ell\extoverset{$G:(?\ell'',A(\ell''))$}\rightsquigarrow [m_1]_\ell=N[L]$.

For the inductive step, $M=M_1\op M_2$, $N=N_1\op N_2$, $M_1\extoverset{$G:(?\ell'',A(\ell''))$}\rightsquigarrow N_1$, and $M_2\extoverset{$G:(?\ell'',A(\ell''))$}\rightsquigarrow N_2$. Clearly, $M_1$ and $M_2$ are $A$-derived and $L$ is an $M_1$-cover and an $M_2$-cover; thus, we can apply the induction hypothesis to obtain that $M_1[L]\extoverset{$G:(?\ell'',A(\ell''))$}\rightsquigarrow N_1[L]$, and $M_2[L]\extoverset{$G:(?\ell'',A(\ell''))$}\rightsquigarrow N_2[L]$. From this, we derive $M[L]=M_1[L]\op M_2[L]\extoverset{$G:(?\ell'',A(\ell''))$}\rightsquigarrow N_1[L]\op N_2[L]=N[L]$.
\end{proof}

\begin{lemma}\label{lemma:onestep}
    Let $H\subseteq \Loc$, $M$ be $A$-derived and $L$ be an $M$-cover for $H:(?\ell', A(\ell'))$. 
    If $M\xrightarrow{\ell:(!G,\gamma)} N$, then $M[L]\xrightarrow{\ell:(!G,\gamma)} N[L]$.
\end{lemma}
\begin{proof}
The proof proceeds by induction on $\xrightarrow{\ell:(!G,\gamma)}$. For the base case, $M=[m]_\ell$, $N=[n]_\ell$ and $m\xrightarrow{(!G,\gamma)}n$. If $M=M[L]$, then also $N=N[L]$ (by \Cref{lemma:preservecompleteness}) and the proof is done. Otherwise, we must have that $M[L]=[m]_\ell[[m_1]_\ell/[m]_\ell]=[m_1]_\ell$ for $m\xrightarrow{(?\ell', A(\ell'))} m_1$ and $\ell\in H$. By \Cref{lemma:sendreceiveorder}, $n\xrightarrow{(?\ell', A(\ell'))} n_1$ and $m_1\xrightarrow{(!G,\gamma)}n_1$. Since $L$ is an $M$-cover, we know it is also an $N$-cover (\Cref{lemma:preservecompleteness}) and therefore $[n]_\ell \in L$. Then we obtain from $n\xrightarrow{(?\ell', A(\ell'))} n_1$ 
and \Cref{lemma:unique-transitions-from-covers} 
that 
$N[L] = [n_1]_\ell$;
hence, $M[L]=[m_1]_\ell\xrightarrow{\ell:(!G,\gamma)}[n_1]_\ell=N[L]$.

For the inductive step, $M=M_1\op M_2$, $N=N_1\op N_2$, $M_1\xrightarrow{\ell:(!G,\gamma)}N_1$, and $M_2\extoverset{$G:(?\ell,\gamma)$}\rightsquigarrow N_2$. Clearly,  $M_1$ and $M_2$ are $A$-derived, and $L$ is an $M_1$-cover and an $M_2$-cover; thus, we can apply the induction hypothesis to obtain that $M_1[L]\xrightarrow{\ell:(!G,\gamma)}N_1[L]$. By \Cref{lemma:sendright}, $\gamma=A(\ell)$. Since $M_2\extoverset{$G:(?\ell,A(\ell))$}\rightsquigarrow N_2$, we can apply \Cref{lemma:stillreceive} to get $M_2[L]\extoverset{$G:(?\ell,A(\ell))$}\rightsquigarrow N_2[L]$. Thus $M[L]=M_1[L]\op M_2[L]\xrightarrow{\ell:(!G,\gamma)}N_1[L]\op N_2[L]=N[L]$. 
\end{proof}


The next lemma shows that $M_i[L]$ is the state of $M_i$ after receiving the message corresponding to $L$. This can be understood as $M_i$ receiving messages from some other decentralized monitor (as what happens in \Cref{lemma:joincommunication}). So, even if $M_i$ receives at any point some messages from another decentralized monitor, the same state is reached when no more communication can take place.

\begin{lemma}\label{lemma:receivealwayshappens}
    Let $H\subseteq \Loc$, $M_1$ be action-derived, and
    $M_1\xrightarrow{c_1} M_2\xrightarrow{c_2}\dots \xrightarrow{c_{n-1}} M_n \xrightarrow{c_n} 
  N$, where $n\geq 1$ and $N$ cannot communicate. If $s\leq n$ and $L$ is an $M_s$-cover for $H:(?\ell, A(\ell))$,
then, for all $i \in\{s,\ldots,n\}$, it holds that
$M_i[L]\xrightarrow{c_i}M_{i+1}[L]\xrightarrow{c_{i+1}}
\dots
\xrightarrow{c_{n-1}} M_n[L] 
\xrightarrow{c_n}N$.
\end{lemma}
\ifarxiversion
\begin{proof}
    We prove this by induction on 
    $n-i$.
    For the base case, $i=n\geq 1$; thus, we know that $M_n\xrightarrow{c_n}N$, and we need to prove that $M_n[L]\xrightarrow{c_n} N[L]$. By \Cref{lemma:preservecompleteness}, $L$ is an $M_n$-cover; 
    hence, we can apply \Cref{lemma:onestep} 
    to conclude that $M_n[L]\xrightarrow{c_n} N[L]$. Since $N$ cannot communicate by assumption, we know that $N[L] = N$.

    In the inductive step, we obtain from the induction hypothesis that $M_{i+1}[L]\xrightarrow{c_{i+1}}
    \dots \xrightarrow{c_{n-1}} M_n[L] \xrightarrow{c_n}N$. Then, we apply \Cref{lemma:onestep} to conclude that $M_i[L]\xrightarrow{c_i}M_{i+1}[L]$, which completes the proof.
\end{proof}
\fi

We can generalize the previous result to include a union of $L$'s.
Notationally, if $L = L_1 \cup \ldots \cup L_k$, we denote $M[L_1]\ldots[L_k]$ as $M[L]$.
This lemma is a cornerstone for proving \Cref{lemma:joincommunication} (together with \Cref{lemma:many}). 
\begin{lemma}\label{lemma:joincommunicationone}
    Let $H\subseteq \Loc$, $M_1 $ be action-derived, and
    $M_1\xrightarrow{c_1} M_2\xrightarrow{c_2}\dots \xrightarrow{c_{n-1}} M_n \xrightarrow{c_n} N$,
    where $n\geq 1$ and $N$ cannot communicate. 
    Let $s\leq n$ and
    $L= \bigcup_{j=1}^k L_j$ 
such that $L_j$ is an $M_s$-cover for $H_j:(?\ell_j, A(\ell_i))$, for all $j$. 
Then, for all $i\in\{s,\ldots,n\}$, it holds that
$M_i[L]\xrightarrow{c_i}M_{i+1}[L]\xrightarrow{c_{i+1}}
\dots
\xrightarrow{c_{n-1}} M_n[L] 
\xrightarrow{c_n}
N$.
%
\end{lemma}
\begin{proof}
We proceed by induction on $k$. If $k=0$, the result is trivial. 
In the inductive step, we obtain from the induction hypothesis that 
\[M_i[L_1]\dots [L_{k-1}]\xrightarrow{c_i}M_{i+1}[L_1]\dots[L_{k-1}]\xrightarrow{c_{i+1}}\dots\xrightarrow{c_{n-1}}M_{n}[L_1]\dots[L_{k-1}]\xrightarrow{c_n} N.\] 
If we show that $L_k$ is a cover for $M_i[L_1]\dots [L_{k-1}]$, the result follows immediately from \Cref{lemma:receivealwayshappens}. Take $[m]_\ell\in M_i[L_1]\dots[L_{k-1}]$ with $\ell \in H_k$ such that $m\boldarrow m'\xrightarrow{(?\ell^k,A(\ell^k))}n$. Hence, we know that there exists $[m_1]_\ell\in M_i$ such that $m_1\boldarrow m$, and thus also $m_1\boldarrow m'$. As $L_k$ is an $M_s$-cover, it is also an $M_i$-cover via \Cref{lemma:preservecompleteness}. Hence, we obtain that $[m']_\ell\in L_k$, as desired.
\end{proof}

The next lemma is actually stronger than what is needed for Bounded Communication (for Bounded Communication it is sufficient that $\bigop_{\ell\in\Loc} M_\ell$ in the lemma below can transition to some decentralized monitor that cannot communicate, and it does not exactly need to be $\bigop_{\ell\in\Loc} N_\ell$). However, for Formula Convergence we need the lemma in its strong version, so we prove only that one.


\begin{lemma}\label{lemma:joincommunication}
    Let $\Loc=\{\ell_1,\dots,\ell_k\}$ and, for all $i \in \{1,\dots,k\}$, we have $
    \Monev{\sigma_i}{\varphi_i}
    %
     \xrightarrow{A}M_i\boldarrow 
     N_i$, where $N_i$ cannot communicate.
     Then,
    $\bigop_{i=1}^k M_i\boldarrow \bigop_{i=1}^k
    N_i
    $.
\end{lemma}
\begin{proof}
    By induction on $k$.
    If $k=1$, we are done immediately. If $\Loc=\{\ell_1,\dots \ell_{k+1}\}$, let $M_1\xrightarrow{c_1}M_1^2\xrightarrow{c_2}\dots\xrightarrow{c_n} 
    N_1$ for $n\geq 0$;
by \Cref{lemma:sendright}, $c_j=\ell_j:(!G_j,A(\ell_j))$, for all $j\in\{1,\dots,n\}$.   
By \Cref{lemma:many}, there exist $L_1,\dots,L_n$ such that, for each $j\in\{1,\dots,n\}$, $L_j$ is a $(\bigop_{i=2}^{k}M_i)$-cover for $G_j:(?\ell_j,A(\ell_j))$.
Let 
\[N = \bigop_{\ell=2}^{k}M_i[L_1]\dots [L_n];\]
 then,
    $\bigop_{i=1}^kM_i = M_1 \op (\bigop_{i=2}^kM_i)\xrightarrow{c_1}\dots\xrightarrow{c_n} 
    N_1
    \op N$. 
    From the 
    induction hypothesis, we know that $\bigop_{i=2}^{k}M_i\boldarrow \bigop_{i=2}^{k}
    N_i
    $. 
    Thus, we can apply \Cref{lemma:joincommunicationone} to conclude that $N\boldarrow \bigop_{i=2}^{k}
    N_i
    $. 
     From $\bigop_{i=1}^kM_i\boldarrow 
     N_1
     \op N$ we can conclude, since $
     N_1
     $ cannot take any communication steps and therefore does not change anymore. 
\end{proof}

We can now finally prove Bounded Communication:

\begin{lemma}\label{lemma:communicationexists}
    For every $\Mone{\varphi}\xrightarrow{A}M\boldarrow M'$, there exists $M''$ such that $M'\boldarrow M''$ and $M''$ cannot communicate.
\end{lemma}
\begin{proof}
We proceed by induction on $\varphi$. The base case is for $\varphi\in\quantfree$
and we have two possible subcases. If $\sigma = \emptyset$, then $\Mone{\varphi}=[v]_{\ell_0} \xrightarrow{A} [v]_{\ell_0} = M=M'$ and $[v]_{\ell_0}$ cannot communicate, since $v \not\xrightarrow c$, for all $c$.
If $\sigma \neq \emptyset$, then
$\Mone{\varphi}=\bigvee_{\ell\in\range(\sigma)}[\synone{\varphi}]_\ell\xrightarrow{A}\bigvee_{\ell\in\range(\sigma)}[m_\ell]_\ell=M$ and $M\boldarrow\bigvee_{\ell\in\range(\sigma)}[n_\ell]_\ell=M'$, where $m_\ell\boldarrow n_\ell$. Actually, we have that $m_\ell\boldarrowt n_\ell$ because we know from the semantics that, for any receive $\xrightarrow{(?\ell',\gamma)}$ in the path $m_\ell\boldarrow n_\ell$, there exists $\ell'$ such that $m_{\ell'}\boldarrow m'_{\ell'} \xrightarrow{(!G,\gamma)}$ (with $\ell\in G$) and, by \Cref{lemma:singlesender}, that $\gamma=A(\ell')$.

As $n_\ell$ is a relevant local monitor, we prove by induction on the structure of relevant local monitors that there exists $n'_\ell$ such that $n_\ell\boldarrowt n'_\ell \NOT{\xrightarrow{c}}$, for all $c$. The first two base cases ($n_\ell=\synone{\psi}$ and $n_\ell=(!G,a).\synone{\psi}$) are trivial because of \Cref{lemma:nocommunicate}. For the third base case, $n_\ell=\sum_{a\in\Act}(?\{\ell'\},a).\synone{\psi_a}$, we observe that there is some $a\in\Act$ such that $a=A(\ell')$; so, we take that transition and conclude again by \Cref{lemma:nocommunicate}. For the inductive case,  $n_\ell=n_1\pandor n_2$ and, from the inductive hypothesis, $n_1\boldarrowt m_1$ and $n_2\boldarrowt m_2$ such that $m_1\NOT{\xrightarrow{c}}$ and $m_2\NOT{\xrightarrow{c}}$, for all $c$. By \Cref{lemma:merge}, $n_1\pandor n_2\boldarrowt m_1\pandor m_2 \NOT{\xrightarrow{c}}$, for all $c$. 
 
Hence, $n_\ell\boldarrowt n'_\ell$ such that $n'_\ell\NOT{\xrightarrow{c}}$, for all $c$. From this, we derive that $\synone{\psi}\xrightarrow{A(\ell)}m_\ell\boldarrowt n'_\ell$, for all $\ell\in\range(\sigma)$. We can then apply \Cref{lemma:psijoincom} (with $m_\ell=n_\ell$ in the statement of the lemma) to conclude that $\bigvee_{\ell\in\range(\sigma)}[m_\ell]_\ell\boldarrow \bigvee_{\ell\in\range(\sigma)}[n'_\ell]_\ell$. It is immediate that $\bigvee_{\ell\in\range(\sigma)}[n'_\ell]_\ell\NOT{\xrightarrow{c}}$, for all $c$. We apply \Cref{lemma:itsallrelative} to derive that $M'\boldarrow \bigvee_{\ell\in\range(\sigma)}[n'_\ell]_\ell$, as desired.

We only treat one of the inductive cases, the others are similar. Let $\varphi=\exists \pi.\varphi'$. Then $\Mone{\varphi}=\bigvee_{\ell\in\Loc}\Monev{\sigma[\pi\mapsto\ell]}{\varphi'}\xrightarrow{A} \bigvee_{\ell\in\Loc}M_\ell=M$. 
From the induction hypothesis we obtain that $M_\ell\boldarrow 
N_\ell$ such that $N_\ell$ cannot communicate. We apply \Cref{lemma:joincommunication} and conclude that $M\boldarrow
\bigvee_{\ell\in\Loc}N_\ell$. Then we can use \Cref{lemma:itsallrelative} to derive that $M'\boldarrow
\bigvee_{\ell\in\Loc}N_\ell$, which concludes the proof.
\end{proof}

\end{textAtEnd}

\begin{textAtEnd}[allend, category=pca]
\label{sec:pca}

We showed the first item of Processing-Communication Alternation in \Cref{lemma:nocommunicate}.
For the second item, we actually prove the contrapositive, because this simplifies the inductive proofs; we first show the result for local monitors and then we conclude by an induction based on \Cref{lemma:onlyone1}.

\begin{lemma}\label{lemma:onlyone1}
    Let $\synone{\psi}\xrightarrow{A(\ell)}m'\boldarrow m\xrightarrow{A'(\ell)}n$. Then $m\NOT{\xrightarrow{c}}$, for all $c$. 
    \end{lemma}
    \begin{proof}
        We focus on relevant local monitors, because $m$ is a relevant local monitor via \Cref{lemma:A-derived-iff-relevant}.
    We proceed by induction on $m$. 
    In the base case, for $m=\synone{\psi}$ the result follows from \Cref{lemma:nocommunicate};
    for $m=(!G,a).\synone{\psi}$ and $m=\sum_{a\in\Act}(?\{\ell''\},a).\synone{\psi_a}$, the premise of the lemma is not satisfied. 
    
    For the inductive step, let $m=m_1\pandor m_2$ with $m_1$ and $m_2$ relevant monitors. We derive from $m_1\pandor m_2 \xrightarrow{A'(\ell)} n$ that $n=n_1\pandor n_2$, $m_1\xrightarrow{A'(\ell)}n_1$ and $m_2\xrightarrow{A'(\ell)}n_2$. Hence, we can apply the induction hypothesis to obtain that  $m_1\NOT{\xrightarrow{c}}$ and $m_2\NOT{\xrightarrow{c}}$, from which we conclude.
    \end{proof}

    
    \begin{lemma}\label{lemma:onlyone}
        Let $M\in\DMon$ be action-derived. If $M\xrightarrow{A} N$, then $M\NOT{\xrightarrow{c}}$, for all $c$. 
    \end{lemma}
\ifarxiversion
    \begin{proof}
    We proceed by induction on $\xrightarrow{A}$. 
    For the first base case, we consider $M=[m]_{\ell}$, $N=[n]_\ell$, $m\xrightarrow{a}n$ and $A(\ell)=a$; by \Cref{lemma:onlyone1}, $m\NOT{\xrightarrow{c}}$, from which we obtain that $M\NOT{\xrightarrow{c}}$. 
    For the second base case, we consider $M=[m]_\ell$, and by assumption $m\NOT{\xrightarrow{c}}$. 
    For the inductive step, we consider $M=M_1\op M_2$, $N=N_1\op N_2$, $M_1\xrightarrow{A}N_1$ and $M_2\xrightarrow{A} N_2$. From the induction hypothesis, it follows that $M_1\NOT{\xrightarrow{c}}$ and $M_2\NOT{\xrightarrow{c}}$; so, $M\NOT{\xrightarrow{c}}$.
    \end{proof}
    \fi
\end{textAtEnd}

\begin{textAtEnd}[allend, category=fc]
\label{sec:fc}

We start with a lemma that 
captures a basic equality between synthesized centralized and decentralized local monitors, used in the proofs of \Cref{lemma:crux} and \Cref{lemma:imaginary}. 
    \begin{lemma}\label{lemma:equal}
        If $\Mc{\psi}=\Mcv{\psi'}{\sigma'}$, for either $\sigma = \sigma'$ or $|\Loc| = 1$,
        then $\synvone{\ell}{\sigma}{\psi}=\synvone{\ell}{\sigma'}{\psi'}$, for all $\ell\in \Loc$.
        \end{lemma}
        \begin{proof}
        We proceed by induction on $\psi$. 
        As a first base case, we consider $\psi=\ttt$: we have $\Mc{\psi}=\yes=\Mcv{\psi'}{\sigma'}$ and, by \Cref{lemma:threefacts}, we conclude that, for all $\ell\in \Loc$, we have $\synvone{\ell}{\sigma}{\psi}=\yes=\synvone{\ell}{\sigma'}{\psi'}$. The base cases for $\psi=\ff$, $\psi=(\pi=\pi')$ and $\psi=(\pi\neq\pi')$ are similar. 
        The last base case is for $\psi=x$: then, $\Mc{\psi}=x=\Mcv{\psi'}{\sigma'}$, from which we can conclude that $\psi'=x$. Immediately it follows that, for all $\ell\in \Loc$, $\synvone{\ell}{\sigma}{\psi}=x=\synvone{\ell}{\sigma'}{\psi'}$.
        
        We now proceed with the inductive cases. We first consider $\psi=[a_\pi]\psi_1$
(the case for $\psi=\langle  a_\pi\rangle \psi_1$ is the same, with $\no$ in place of $\yes$);
thus $\Mc{\psi}=a_{\sigma(\pi)}.\Mc{\psi_1} + \sum_{b \neq a} b_{\sigma(\pi)}.\yes=\Mcv{\psi'}{\sigma'}$ and so $\psi'=[a_{\pi'}]\psi_2$, $\sigma(\pi)=\sigma'(\pi')$ and $\Mc{\psi_1}=\Mcv{\psi_2}{\sigma'}$. We use the induction hypothesis to obtain that, for all $\ell\in\Loc$, we have $\synvone{\ell}{\sigma}{\psi_1}=\synvone{\ell}{\sigma'}{\psi_2}$. We have two possibilities:
\begin{enumerate}
\item $\sigma(\pi)\ (= \sigma'(\pi')) = \ell$. In this case,
$\synone{\psi}=a.(!(\range(\sigma) \setminus \{\ell\}),a).\synone{\psi_1}+\sum_{b\neq a}b.(!(\range(\sigma) \setminus \{\ell\}),b).\yes$ and $\synvone{\ell}{\sigma'}{\psi'}=a.(!(\range(\sigma') \setminus \{\ell\}),a).\synvone{\ell}{\sigma'}{\psi_2}+\sum_{b\neq a}b.(!(\range(\sigma') \setminus \{\ell\}),b).\yes$. Using the induction hypothesis, we obtain the required result, since $\range(\sigma) \setminus \{\ell\} = \range(\sigma') \setminus \{\ell\}$, both if $\sigma=\sigma'$ and if $|\Loc|=1$ (and so $\Loc = \{\ell\}$).

\item $\sigma(\pi) \ (= \sigma'(\pi')) \neq \ell$ (notice that this case is not possible if $|\Loc|=1$). Then, $\synone{\psi}= \sum_{b \in \Act}b.\left ( ?(\{\sigma(\pi)\},a).\synone{\psi_1}+\sum_{b\neq a}(?(\{\sigma(\pi)\}),b).\yes \right)$ 
and $\synvone{\ell}{\sigma'}{\psi'}= \sum_{b \in \Act}b.\left ( ?(\{\sigma'(\pi)\},a).\synvone{\ell}{\sigma'}{\psi_2}+\sum_{b\neq a}(?(\{\sigma'(\pi')\}),b).\yes \right)$. Using the induction hypothesis, we obtain the required result.
\end{enumerate}

        Next we consider $\psi=\maxx {x}.\psi_1$. Hence, $\Mc{\psi}=\rec x.\Mc{\psi_1}=\Mcv{\psi'}{\sigma'}$. From this we conclude that $\psi'=\maxx{x}.\psi_2$ and $\Mc{\psi_1}=\Mcv{\psi_2}{\sigma'}$. We use the induction hypothesis to obtain that, for all $\ell\in\Loc$, we have $\synvone{\ell}{\sigma}{\psi_1}=\synvone{\ell}{\sigma'}{\psi_2}$; thus, for all $\ell\in\Loc$, $\synone{\maxx {x}.\psi_1}=\rec x.\synone{\psi_1}=\rec x.\synvone{\ell}{\sigma'}{\psi_2}=\synvone{\ell}{\sigma'}{\maxx{x}.\psi_2}=\synvone{\ell}{\sigma'}{\psi'}$.
        
        Finally, let $\psi=\psi_1\op \psi_2$. Hence, $\Mc{\psi}=\Mc{\psi_1}\pandor \Mc{\psi_2}=\Mcv{\psi'}{\sigma'}$ (for $\pandor=\pand$ if $\op=\wedge$ and $\pandor=\por$ if $\op=\vee$)), with $\psi'=\psi_1'\op \psi_2'$, $\Mc{\psi_1}=\Mcv{\psi_1'}{\sigma'}$ and $\Mc{\psi_2}=\Mcv{\psi_2'}{\sigma'}$. The result immediately follows from the induction hypothesis.
        \end{proof}

        To prove Formula Convergence for quantifier free formulas, we need the following result. Informally, it captures part of the weak bisimulation at the level of local monitors.
        
        \begin{lemma}\label{lemma:crux}
            Let $\psi\in\quantfree$. If
        $\Mc{\psi}\xrightarrow{A}\Mc{\psi'}$, then $\synone{\psi}\xrightarrow{A(\ell)}m_\ell\boldarrowt \synone{\psi'}$, for all $\ell\in \Loc$. 
        \end{lemma}
        \begin{proof}
        We proceed by induction on $\xrightarrow{A}$ for central monitors. 
The base case with $\Mc{\psi}=a_\ell.m$ never applies, because of the assumption that $|\Act| \geq 2$. Hence, the only base case is for $\Mc{\psi}=v=\Mc{\psi'}$; then,  \Cref{lemma:threefacts} entails that $\synvone{\ell}{\sigma}{\psi}= v =\synvone{\ell}{\sigma}{\psi'}$, for all $\ell\in\Loc$, and we can conclude, as $v\xrightarrow{A(\ell)}v\boldarrowt v$ for any $A(\ell)$. 
        
For the inductive case, let $\Mc{\psi}=m+n\xrightarrow{A} \Mc{\psi'}$;
hence, either $\psi=[a_\pi]\psi''$ 
or $\psi=\langle a_\pi\rangle\psi''$,
and $\Mc{\psi} = a_{\sigma(\pi)}.\Mc{\psi''} + \sum_{b\neq a}b_{\sigma(\pi)}.v$, 
where $v = \yes$, in case of box, and $v = \no$, in case of diamond.
We distinguish two cases. 
        \begin{enumerate}
            \item 
            $\sigma(\pi)=\ell'$ and $A(\ell')=a$; this entails that $\Mc{\psi'}=\Mc{\psi''}$.
            Fix $\ell \in \Loc$. If $\ell=\ell'$, we have $\synvone{\ell}{\sigma}{\psi}\xrightarrow{A(\ell)}(!(\range(\sigma)\setminus\{\ell\},a).\synvone{\ell}{\sigma}{\psi''}$; if $\ell\neq \ell'$, we have $\synvone{\ell}{\sigma}{\psi}\xrightarrow{A(\ell)}(?\{\ell'\},a).\synvone{\ell}{\sigma}{\psi''}+\sum_{b\neq a}(?\{\ell'\},b).v$. Thus, in both situations we can have that $\synvone{\ell}{\sigma}{\psi}\xrightarrow{A(\ell)}m_{\ell}\xrightarrow{c}\synvone{\ell}{\sigma}{\psi''}$, where $\gamma=A(\ell')\ (= a)$ whenever $c=(?\{\ell'\},\gamma)$. From \Cref{lemma:equal} 
and $\Mc{\psi'}=\Mc{\psi''}$, we conclude  $\synvone{\ell}{\sigma}{\psi'}=\synvone{\ell}{\sigma}{\psi''}$, for all $\ell\in \Loc$.
        
            \item 
            $\sigma(\pi)=\ell'$, $A(\ell')=b$, and $a\neq b$; this entails that $\Mc{\psi'}=v$. We use \Cref{lemma:threefacts} to conclude that $\synone{\psi'}=v$, for all $\ell$.
            Fix $\ell \in \Loc$. If  $\ell=\ell'$, we have $\synvone{\ell}{\sigma}{\psi}\xrightarrow{A(\ell)}(!(\range(\sigma)\setminus\{\ell\}),b).v=(!(\range(\sigma)\setminus\{\ell\},b).\synvone{\ell}{\sigma}{\psi'}$;  if $\ell\neq \ell'$, we have $\synvone{\ell}{\sigma}{\psi}\xrightarrow{A(\ell)}(?\{\ell'\},a).\synvone{\ell}{\sigma}{\psi''}+\sum_{b\neq a}(?\{\ell'\},b).v=(?\{\ell'\},a).\synvone{\ell}{\sigma}{\psi''}+\sum_{b\neq a}(?\{\ell'\},b).\synvone{\ell}{\sigma}{\psi'}$. Thus in both situations we can conclude that $\synvone{\ell}{\sigma}{\psi}\xrightarrow{A(\ell)}m_{\ell}\xrightarrow{c}\synvone{\ell}{\sigma}{\psi'}$, where $\gamma=A(\ell')$ whenever $c=(?\{\ell'\},\gamma)$. 
        \end{enumerate}
        
        Next we consider the recursive case: $\Mc{\psi}=\rec x.m\xrightarrow{A} \Mc{\psi'}$ because $m\{^{\rec x.m}/_x\}\xrightarrow{A} \Mc{\psi'}$. Thus we obtain that $\psi=\max x.\psi_1$ and $m=\Mc{\psi_1}$. As $\Mc{\psi_1}\{^{\rec x.\Mc{\psi_1}}/_x\}=\Mc{\psi_1\{^{\maxx x.\psi_1}/_x\}}$ (\Cref{lemma:recursion-substitution}), we use the induction hypothesis to obtain that, for all $\ell$, we have $\synone{\psi_1\{^{\maxx x.\psi_1}/_x\}}\xrightarrow{A(\ell)} m_\ell\boldarrowt \synone{\psi'}$. Using that 
        $\synone{\psi_1}\{^{\rec x.\synone{\psi_1}}/_x\}=\synone{\psi_1\{^{\maxx x.\psi_1}/_x\}}$ (\Cref{lemma:recursion-substitutiontwo}), we get that $\synone{\psi_1}\{^{\rec x.\synone{\psi_1}}/_x\}\xrightarrow{A(\ell)} m_\ell\boldarrowt \synone{\psi'}$. This implies that $\synone{\psi}\xrightarrow{A(\ell)} m_\ell\boldarrowt \synone{\psi'}$.
        
        Last, let $\Mc{\psi}=m\pandor n\xrightarrow{A} m'\pandor n'$ because $m\xrightarrow{A}m'$ and $n\xrightarrow{A}n'$. As $\psi$ is quantifier-free, we get that $\psi=\psi_1\op\psi_2$ (for $\op=\wedge$ if $\pandor=\pand$ and $\op=\vee$ if $\pandor=\por$), $m=\Mc{\psi_1}$ and $n=\Mc{\psi_2}$. Also, $\psi'=\psi_1'\op\psi_2'$, $m'=\Mc{\psi_1'}$ and $n'=\Mc{\psi_2'}$. From the induction hypothesis we obtain that, for all $\ell$, we have $\synone{\psi_1}\xrightarrow{A(\ell)}m_\ell^1\boldarrowt \synone{\psi_1'}$ and $\synone{\psi_2}\xrightarrow{A(\ell)}m_\ell^2\boldarrowt \synone{\psi_2'}$. Thus, for all $\ell$, we have that $\synone{\psi_1\op\psi_2}\xrightarrow{A(\ell)}m_\ell^1\pandor m_\ell^2$ and, via \Cref{lemma:merge} (that can be used thanks to \Cref{lemma:nocommunicate}), we conclude the desired $m_\ell^1\pandor m_\ell^2\boldarrowt \synone{\psi_1'}\pandor \synone{\psi_2'}$. 
        \end{proof}

         We can now prove Formula Convergence for $\psi\in\quantfree$:

    \begin{lemma}\label{lemma:actionbispsi}
        Let $\psi\in\quantfree$.
        If $\Mone{\psi}\xrightarrow{A} M \boldarrow M'$, with $M'$ that cannot communicate, and $\Mc{\psi} \xrightarrow{A} \Mcv{\psi'}{\sigma}$, for some formula $\psi'$, then $M' = \Monev{\sigma}{\psi'}$.
    \end{lemma}
    \begin{proof}
    First, since $\Mone{\psi}\xrightarrow{A} M$, we have that $\FVloc(\psi) \neq \emptyset$ and so $\sigma \neq \emptyset$.
Then, \Cref{lemma:crux} tells us that $\synone{\psi}\xrightarrow{A(\ell)}m_\ell\boldarrowt \synone{\psi'}$, for all $\ell\in \Loc$. Thus, we know that $\Mone{\psi}=\bigvee_{\ell\in \range(\sigma)}[\synone{\psi}]_\ell\xrightarrow{A}\bigvee_{\ell\in \range(\sigma)} [m_\ell]_\ell$; because $\synone{\psi'}\NOT{\xrightarrow{c}}$ for all $c$ (\Cref{lemma:nocommunicate}), we can use \Cref{lemma:psijoincom} (by taking $n_\ell=m_\ell$ in the statement of that result) to obtain that $\bigvee_{\ell \in \range (\sigma)} [m_\ell]_\ell\boldarrow\bigvee_{\ell\in \range(\sigma)}[\synone{\psi'}]_\ell=\Mone{\psi'}$ and $\Mone{\psi'}$ cannot communicate. Hence, we apply \Cref{lemma:itsallrelative} to derive that $M'\boldarrow \Mone{\psi'}$. Since $M'$ cannot communicate by assumption, we get that $M'=\Mone{\psi}$.
    \end{proof}

    Then we show a result for the synthesis function in the centralized setting specifically for $\psi$-formulas, which differs from \Cref{lemma:chain} because $\sigma$ remains constant. This allows us to prove the base case in the proof for Formula Convergence (\Cref{lemma:actionbis}) via \Cref{lemma:actionbispsi}.
    \begin{lemma}\label{lemma:chainpsi}
    If $\Mc{\psi}\xrightarrow{A} m'$, then $m'=\Mc{\psi'}$ for some $\psi'$.
    \end{lemma}
    \ifarxiversion
    \begin{proof}
    The proof is a simplified version of the proof for \Cref{lemma:chain}.
    \end{proof}
    \fi

    The next lemma is used to prove \Cref{lemma:ikweethetniet}, which in turn is used in the final steps of the proof of Formula Convergence (\Cref{lemma:actionbis}).

    \begin{lemma}\label{lemma:imaginary}
        If $\Mc{\varphi}=\Mcv{\varphi'}{\sigma'}$ and $\Loc=\{\ell\}$, then $\Mone{\varphi}=\Monev{\sigma'}{\varphi'}$.
        \end{lemma}
        \begin{proof}
        We proceed by induction on $\varphi$. In the base case we consider $\varphi\in\quantfree$: by \Cref{lemma:equal} 
and $\Loc = \{\ell\}$, we obtain $\Mone{\varphi}=[\synone{\varphi}]_\ell=[\synvone{\ell}{\sigma'}{\varphi'}]_\ell=\Monev{\sigma'}{\varphi'}$. 
        
        In the inductive case for $\varphi=\exists \pi.\varphi''$, we have $\Mc{\varphi}=\Mcv{\varphi''}{\sigma[\pi\mapsto\ell]}=\Mcv{\varphi'}{\sigma'}$. From the induction hypothesis we derive that $\Monev{\sigma[\pi\mapsto\ell]}{\varphi''}=\Monev{\sigma'}{\varphi'}$, which proves the claim, since $\Monev{\sigma[\pi\mapsto\ell]}{\varphi''} = \Mone{\varphi}$ because $\Loc = \{\ell\}$. The case for $\varphi = \forall \pi.\varphi''$ is similar
        
        Finally, let $\varphi=\varphi_1\op\varphi_2$. Hence, $\Mc{\varphi}=\Mc{\varphi_1}\pandor \Mc{\varphi_2}=\Mcv{\varphi'}{\sigma'}$ (for $\pandor=\pand$ if $\op=\wedge$ and $\pandor=\por$ if $\op=\vee$). As $\Loc=\{\ell\}$, we know that $\varphi'=\varphi_1'\op \varphi_2'$ with $\Mc{\varphi_1}=\Mcv{\varphi_1'}{\sigma'}$ and $\Mc{\varphi_2}=\Mcv{\varphi_2'}{\sigma'}$. From the induction hypothesis we conclude that $\Mone{\varphi_1}=\Monev{\sigma'}{\varphi_1'}$ and $\Mone{\varphi_2}=\Monev{\sigma'}{\varphi_2'}$. The result now follows: $\Mone{\varphi}=\Mone{\varphi_1}\op\Mone{\varphi_2}=\Monev{\sigma'}{\varphi_1'}\op \Monev{\sigma'}{\varphi_2'}=\Monev{\sigma'}{\varphi'}$.
        \end{proof}

\begin{lemma}\label{lemma:ikweethetniet}
    If $\Mc{\varphi}=\oplus_{i\in I}\Mcv{\varphi_i}{\sigma_i}$, then $\Mone{\varphi}=\bigvee_{i\in I}\Monev{\sigma_i}{\varphi_i}$.
If $\Mc{\varphi}=\otimes_{i\in I}\Mcv{\varphi_i}{\sigma_i}$, then $\Mone{\varphi}=\bigwedge_{i\in I}\Monev{\sigma_i}{\varphi_i}$.
    \end{lemma}
    \begin{proof}
We prove only the first claim (the second one is similar) and assume that $\oplus_{i\in I}\Mcv{\varphi_i}{\sigma_i}$ is maximal, i.e., no $\Mcv{\varphi_i}{\sigma_i}$ is of the form 
    $m_1\oplus m_2$.
We proceed by induction on $\varphi$. 
If $|I| = 1$ (this covers the base cases -- i.e., when $\varphi$ is an atom -- and somee inductive cases -- e.g., an existential but with $|\Loc| = 1$), then we have $\Mc{\varphi}=\Mcv{\varphi_i}{\sigma_i}$ and the result follows from \Cref{lemma:imaginary}.


If $\varphi = \varphi_1 \lor \varphi_2$, then $\Mc{\varphi} = \Mc{\varphi_1} \oplus \Mc{\varphi_2} = \oplus_{i\in I}\Mcv{\varphi_i}{\sigma_i}$.
        Then, there is some $J \subsetneq I$ such that 
        $\Mc{\varphi_1}=\oplus_{i\in J}\Mcv{\varphi_i}{\sigma_i}$ and
        $\Mc{\varphi_2}=\oplus_{i\in I\setminus J}\Mcv{\varphi_i}{\sigma_i}$.
        By the induction hypothesis, 
        $\Mone{\varphi_1}=\bigvee_{i\in J}\Monev{\sigma_i}{\varphi_i}$ and
        $\Mone{\varphi_2}=\bigvee_{i\in I\setminus J}\Monev{\sigma_i}{\varphi_i}$, and therefore 
        $\Mone{\varphi} = \Mone{\varphi_1} \vee \Mone{\varphi_2} = \bigvee_{i\in I}\Monev{\sigma_i}{\varphi_i}$.
    
    
    Now consider $\varphi=\exists \pi.\varphi'$, with $|\Loc| > 1$; then, 
    $\Mc{\varphi} = \oplus_{\ell \in \Loc} \Mcv{\varphi'}{\sigma[\pi \mapsto \ell]}$. From our assumption on the maximality of $\oplus_{i\in I}\Mcv{\varphi_i}{\sigma_i}$, there is a partition $(I_\ell)_{\ell \in \Loc}$ of $I$ such that 
    $\Mcv{\varphi'}{\sigma[\pi \mapsto \ell]} = \oplus_{i\in I_\ell}\Mcv{\varphi_i}{\sigma_i}$, for every $\ell \in \Loc$.
    By the induction hypothesis,  
    $\Monev{\sigma[\pi \mapsto \ell]}{\varphi'} = \bigvee_{i\in I_\ell}\Monev{\sigma_i}{\varphi_i}$, for every $\ell \in \Loc$, and we can conclude that
    $\Mone{\varphi} = \bigvee_{\ell\in \Loc}\Monev{\sigma[\pi \mapsto \ell]}{\varphi'} = \bigvee_{\ell\in \Loc}\bigvee_{i\in I_\ell}\Monev{\sigma_i}{\varphi_i}= \bigvee_{i\in I}\Monev{\sigma_i}{\varphi_i}$.
    \end{proof}


\begin{lemma}\label{lemma:actionbis}
    If $\Mone{\varphi}\xrightarrow{A} M \boldarrow M'$, with $M'$ that cannot communicate, and $\Mc{\varphi} \xrightarrow{A} \Mcv{\varphi'}{\sigma'}$, for some formula $\varphi'$ and environment $\sigma'$, then $M' = \Monev{\sigma'}{\varphi'}$.
\end{lemma}
\begin{proof}
We proceed by induction on $\varphi$. 
The case for $\psi$ follows from \Cref{lemma:actionbispsi}, where we know that $\sigma'=\sigma$ because of \Cref{lemma:chainpsi}.

We start with the inductive case for $\varphi = \exists \pi.\varphi''$. Since $\Monev{\sigma}{\exists \pi.\varphi''}=\bigvee_{\ell\in \Loc} \Monev{\sigma[\pi \mapsto \ell]}{\varphi''}$, we have $M=\bigvee_{\ell\in\Loc} M_\ell$, where $ \Monev{\sigma[\pi \mapsto \ell]}{\varphi''}\xrightarrow {A} M_\ell$, for each $\ell$. By Bounded Communication, we obtain that, for all $\ell$, there exists $N_\ell$ such that $M_\ell\boldarrow 
N_\ell$ and $N_\ell$ cannot communicate. We apply \Cref{lemma:joincommunication} and conclude that $M\boldarrow\bigvee_{\ell\in\Loc} N_\ell$. Then we can use \Cref{lemma:itsallrelative} to derive that $
\bigvee_{\ell\in\Loc} N_\ell
\boldarrow M'$, which implies that $
\bigvee_{\ell\in\Loc} N_\ell
= M'$ because $
\bigvee_{\ell\in\Loc} N_\ell
$ cannot communicate.
Moreover, $\Mc{\varphi} = \bigoplus_{\ell\in\Loc}\Mcv{\varphi''}{\sigma[\pi \mapsto \ell]}$ and
$\Mcv{\varphi'}{\sigma'} = \bigoplus_{\ell\in\Loc}\Mcv{\varphi_\ell}{\sigma_\ell}$, where
 $\Mcv{\varphi''}{\sigma[\pi \mapsto \ell]}\xrightarrow{A} \Mcv{\varphi_\ell}{\sigma_\ell}$, for each $\ell \in \Loc$. 
 From $\Monev{\sigma[\pi \mapsto \ell]}{\varphi''} \xrightarrow{A} M_\ell\boldarrow N_\ell$ and the induction hypothesis, we can obtain that $N_\ell=\Monev{\sigma_\ell}{\varphi_\ell}$; hence, $M' = \bigvee_{\ell\in\Loc} \Monev{\sigma_\ell}{\varphi_\ell}$ and we conclude via \Cref{lemma:ikweethetniet}, that implies that $\bigvee_{\ell \in \Loc} \Monev{\sigma_\ell}{\varphi_\ell} = \Monev{\sigma'}{\varphi'}$.

The case for $\varphi_1\vee \varphi_2$ is similar; the cases for $\forall \pi.\varphi''$ and $\varphi_1\wedge\varphi_2$ are dual.
\end{proof}
\end{textAtEnd}

\section{Conclusion}
\label{sec:concl}

We provided two methods to synthesize monitors for hyperproperties expressed as fragments of \HypermuHML.
Our first synthesis procedure constructs monitors that analyse hypertraces in a centralized manner and are guaranteed to correctly detect all violations of the respective formula, as long as it does not have a least fixed-point operator.
Our second synthesis algorithm constructs monitors that operate in a decentralized manner and communicate with one another using multicast to share relevant information between 
them.
%
The decentralized-monitor synthesis provides the same correctness guarantees as the centralized one, but is only defined for formulas with trace quantifiers that do not appear inside any fixed-point operator. 
This additional restriction, which is natural and present in many monitoring set-ups for hyperlogics, \emph{e.g.}~\cite{monitor_hyper,hyperLTL,AgrawalB16,brett17,Gustfeld21,complex_monitor_hyper}, allows us to focus on examining the intricacies of monitoring in a decentralized setting with monitor communication. 
More precisely, it allows us to fix the $\sigma$ in the synthesis function which, in turn, produces a \emph{static} set of locations with which a monitor 
can communicate. 
Despite the restriction to \PHypermuHML, our synthesis algorithm still covers properties that were previously not even expressible, hence not monitorable, in state-of-the-art hyperlogics. 
%
%

\paragraph{Future work}
Of course, the picture is still incomplete: we have a centralized-monitor synthesis procedure for an expressive fragment of \HypermuHML, 
whereas our decentralized-monitor synthesis deals with a more restricted fragment of that logic.
It is not clear if this restriction is necessary; 
for example, a different decentralized-monitor synthesis for a larger fragment might be obtained by utilizing a different communication paradigm other than multicast, which was adopted in this study.
In fact, we conjecture that broadcast communications might allow us to synthesize decentralized monitors for a larger \HypermuHML\ fragment, including formulae that mix greatest fixed-points and quantifiers, like $\varphi_a$ defined in \eqref{exp:property}; currently, monitors only send messages to the locations in the range of the specified $\sigma$.

A second interesting direction is to allow monitors to infer information from communications they did not receive. 
%
A good starting point to explore such a synthesis algorithm (and prove its correctness) can be the synthesis properties in \Cref{def:principled-synthesis}.
To fully delineate the power of decentralized monitoring, a maximality result in the spirit of those presented in~\cite{AcetoPOPL19,AcetoAFIL21:sosym} is needed, which we intend to establish in the future.

\revTwo{
A third avenue for future investigation is to synthesize monitors that detect all satisfying hypertraces for the respective dual fragments of \HypermuHML, instead of focusing on monitors that detect violations (that is, we can focus on satisfaction completeness instead of on violation completeness).
A fourth possible direction is to provide an estimation (or, at least, a lower/upper bound) of the size of synthesized centralized and decentralized monitors for a given formula; this will require non-trivial combinatorial arguments, like those in \cite{AAFIL19a}. 
Another possible direction is to extend the proposed approach to hyperlogics that adopt asynchronous semantics,  like the one adopted in~\cite{ChalupaH23,Gustfeld21}. Furthermore, it would also be challenging to understand how we can model a scenario where the number of localities
varies over time, and how we can handle the dynamic creation/deletion of the associated local monitors.
}

Finally, another relevant direction we intend to pursue in future is the development of tools for monitoring \HypermuHML\ specifications at runtime, based on the results of this article. 
We expect that our decentralised-monitor synthesis procedure can be implemented by generating a dedicated monitor for every location in a way that is very similar to the synthesis of \muHML\ monitors presented in~\cite{AttardAAFIL21,AcetoAAEFI22,AcetoAAEFI24} and implemented in the tool {\tt detectEr} available at \url{https://duncanatt.github.io/detecter/}.

\medskip

\noindent
\textbf{Related Work.}\label{sec:related_work}
To the best of our knowledge, Agrawal and Bonakdarpour were the first to study RV for hyperproperties expressed in HyperLTL in~\cite{AgrawalB16}, where they investigated monitorability for $k$-safety hyperproperties expressed in HyperLTL. They also gave a semantic characterization of monitorable $k$-safety hyperproperties, which is a natural extension to hyperproperties of the `universal version' of the classic definition of monitorability presented by Pnueli-Zaks~\cite{PnueliZ06,AcetoAFIL21:sosym}. In contrast to this work, we do not restrict ourselves to alternation-free formulas (see~\cref{exp:property}) and every monitorable formula considered by Agrawal and Bonakdarpour can be expressed in our monitorable fragment. Brett et al.~\cite{brett17} improve on the work presented in~\cite{AgrawalB16} by presenting an algorithm for monitoring the full alternation-free fragment of HyperLTL. They also highlight challenges that arise when monitoring arbitrary HyperLTL formulas, namely $(i)$ quantifier alternations, $(ii)$ inter-trace dependencies and $(iii)$ relative ordering of events across traces. 
Our decentralized-monitor synthesis addresses $(i)$ by using the number of locations as an upper bound on the number of traces, and $(ii)$ and $(iii)$ via synchronized multicasts.

In~\cite{monitor_hyper}, Finkbeiner et al.\  investigate RV for HyperLTL~\cite{hyperLTL} formulas w.r.t.\ three different input classes, namely the bounded sequential, the unbounded sequential and the parallel classes.
They also develop the monitoring tool RVHyper~\cite{FinkbeinerHST18} based on the sequential algorithms developed for those input classes.
The parallel class is closest to our set-up, since it consists in a \emph{fixed} number of system executions that are processed synchronously. 

Beutner et al.~\cite{BeutnerFFM24} study runtime monitoring for  \hyperLTLtwofp, a temporal logic that is interpreted over sets of \emph{finite} traces of \emph{equal length}. 
Unlike \hyperLTLtwo~\cite{BeutnerFFM23}, \hyperLTLtwofp permits quantification under temporal operators, which is also allowed in our logic \HypermuHML. 
In contrast to HyperLTL, \hyperLTLtwofp features second-order quantification over sets of finite traces and can express properties like common knowledge.

In~\cite{Gustfeld21}, Gustfeld et al.~study automated analysis techniques for asynchronous hyperproperties and propose a novel automata-theoretic framework, the so-called alternating asynchronous parity automata, together with the fixed-point logic $H_\mu$ for expressing asynchronous hyperproperties. 
The logic $H_\mu$ has commonalities with \PHypermuHML,
%
but 
it only allows for prenex formulas; 
moreover, 
its semantics 
progresses asynchronously on each trace.
%
Properties such as ``an atomic proposition does not occur at a certain level in the tree (of traces)'' are not expressible in their logic $H_\mu$, but 
can be described in \HypermuHML. 
%
%

Chalupa and Henzinger~\cite{ChalupaH23} explore the potential of 
monitoring for hyperproperties using prefix transducers. 
They develop a transducer language, called prefix expressions, give it an operational semantics over a hypertrace (reminiscent of the semantics in Section~\ref{sec:decMon}) and then implement 
it to assess the induced overheads.
They show how transducers can use the writing capabilities as a method for monitor synchronization across traces, akin to the monitor communication and verdict aggregation 
of \Cref{sec:decMon}.
%
Since transducers are, in principle, more powerful that passive monitors, 
additional guarantees 
are required to ensure that they do not interfere unnecessarily with system executions.

\bibliographystyle{ACM-Reference-Format}
\bibliography{bibliography.bib}

\ifextendedversion
\clearpage
\appendix

\section{Proofs for Soundness of Centralized Monitors}
\label{app:soundness}
\printProofs[soundness]

\section{Proofs for Correct Synthesis of Decentralized Monitors}
\label{app:corr}
\printProofs[corr]


\subsection{Proofs for Verdict Agreement of $\Monev{-}{\cdot}$}
\label{app:va}
\printProofs[va]

 \subsection{Proofs for Verdict Irrevocability of $\Monev{-}{\cdot}$} 
 \label{app:vi}
 \printProofs[vi]

\subsection{Proofs for Reactivity of $\Monev{-}{\cdot}$}
\label{app:react}
\printProofs[react]

\subsection{Proofs for Bounded Communication of $\Monev{-}{\cdot}$}
\label{app:bc}
 \printProofs[bc]

\subsection{Proofs for Processing-Communication Alternation of $\Monev{-}{\cdot}$}
\label{app:pca}
 \printProofs[pca]

 \subsection{Proofs for Formula Convergence of $\Monev{-}{\cdot}$} 
 \label{app:fc}
 \printProofs[fc]
\fi 

\end{document}